\def\ARXIVVERSION{1}
\def\PAPERWITHSUPPLEMENT{1}
\def\PAPERWITHREFERENCEGUIDE{1}
\newif\ifrefereelayout
\refereelayoutfalse   
\ifrefereelayout
  \documentclass[AER,reviewmode]{AEA}
\else
  \documentclass[AER,finalmode]{AEA}
\fi
\ifdefined\ARXIVVERSION
  \newcommand{\paperversionurl}{https://arxiv.org/pdf/2607.05802}
\else
  \newcommand{\paperversionurl}{https://cashman.science/assets/pdf/safe-collective-expression.pdf}
\fi
\newif\ifanonymizeforpeerreview
\anonymizeforpeerreviewfalse   
\newif\ifshowacknowledgments
\showacknowledgmentstrue     
\newif\ifbuildwithsupplement
\ifdefined\PAPERWITHSUPPLEMENT
  \buildwithsupplementtrue
\else
  \buildwithsupplementfalse
\fi

\usepackage{natbib}
\usepackage{amsmath}
\usepackage{amsfonts}
\usepackage{amssymb}
\usepackage{bm}
\usepackage{booktabs}
\usepackage{array}
\usepackage{longtable}
\makeatletter
\def\LT@makecaption#1#2#3{%
  \LT@mcol\LT@cols c{\hbox to\z@{\hss\parbox[t]\LTcapwidth{%
    \sbox\@tempboxa{#1{#2}#3}%
    \ifdim\wd\@tempboxa>\hsize #1{#2}#3\else
    \hbox to\hsize{\hfil\box\@tempboxa\hfil}\fi
    \endgraf\vskip\baselineskip}\hss}}}
\makeatother
\usepackage{graphicx}
\usepackage{tikz}
\usepackage[letterpaper,margin=1in]{geometry}
\usepackage{xr-hyper}
\usepackage{subfiles}
\usepackage[colorlinks=true,allcolors=blue]{hyperref}

\makeatletter
\let\AC@original@linkfile\hyper@linkfile
\def\hyper@linkfile#1#2#3{%
  \begingroup\AC@original@linkfile{#1}{#2}{#3}\endgroup}
\makeatother

\makeatletter
\let\AC@original@appsect\@appsect
\def\@appsect#1#2#3#4#5#6[#7]#8{%
  \NR@gettitle{#7}%
  \AC@original@appsect{#1}{#2}{#3}{#4}{#5}{#6}[{#7}]{\csname the#1\endcsname.\ #8}%
}
\makeatother

\let\ACbodysection\section

\newtheorem{assumption}{Assumption}
\newtheorem{definition}{Definition}
\newtheorem{lemma}{Lemma}
\newtheorem{proposition}{Proposition}
\newtheorem{theorem}{Theorem}
\newtheorem{corollary}{Corollary}
\newtheorem{remark}{Remark}

\newcommand{\Pfam}{\mathcal{P}}
\newcommand{\Cs}{C^{-}}
\newcommand{\Cp}{C^{+}}
\newcommand{\appendixparttitle}[1]{%
  \clearpage
  \begin{center}
  {\scshape\sectionsize #1}
  \end{center}
  \vspace{12pt}
}

\newcommand{\runinhead}[1]{\par\textit{#1}---\ignorespaces}

\ifSubfilesClassLoaded{%
  \externaldocument[][nocite]{safe_coalitions}[safe_coalitions.pdf]%
}{%
  \ifbuildwithsupplement\else
    \externaldocument[][nocite]{safe_coalitions_online_appendix}[safe_coalitions_online_appendix.pdf]%
  \fi
}

\providecommand{\theHsection}{}
\providecommand{\theHsubsection}{}
\providecommand{\theHequation}{}
\providecommand{\theHfigure}{}
\providecommand{\theHtable}{}
\providecommand{\theHassumption}{}
\providecommand{\theHdefinition}{}
\providecommand{\theHlemma}{}
\providecommand{\theHproposition}{}
\providecommand{\theHtheorem}{}
\providecommand{\theHcorollary}{}
\providecommand{\theHremark}{}
\newcommand{\startsupplement}{%
  \appendix
  \setcounter{section}{0}%
  \setcounter{subsection}{0}%
  \setcounter{equation}{0}%
  \setcounter{figure}{0}%
  \setcounter{table}{0}%
  \setcounter{assumption}{0}%
  \setcounter{definition}{0}%
  \setcounter{lemma}{0}%
  \setcounter{proposition}{0}%
  \setcounter{theorem}{0}%
  \setcounter{corollary}{0}%
  \setcounter{remark}{0}%
  \renewcommand{\thesection}{S\arabic{section}}%
  \renewcommand{\thesubsection}{\arabic{subsection}}%
  \renewcommand{\theequation}{S\arabic{equation}}%
  \renewcommand{\thefigure}{S\arabic{figure}}%
  \renewcommand{\thetable}{S\arabic{table}}%
  \renewcommand{\theassumption}{S\arabic{assumption}}%
  \renewcommand{\thedefinition}{S\arabic{definition}}%
  \renewcommand{\thelemma}{S\arabic{lemma}}%
  \renewcommand{\theproposition}{S\arabic{proposition}}%
  \renewcommand{\thetheorem}{S\arabic{theorem}}%
  \renewcommand{\thecorollary}{S\arabic{corollary}}%
  \renewcommand{\theremark}{S\arabic{remark}}%
  \renewcommand{\theHsection}{supplement.\arabic{section}}%
  \renewcommand{\theHsubsection}{supplement.\arabic{section}.\arabic{subsection}}%
  \renewcommand{\theHequation}{supplement.\arabic{equation}}%
  \renewcommand{\theHfigure}{supplement.\arabic{figure}}%
  \renewcommand{\theHtable}{supplement.\arabic{table}}%
  \renewcommand{\theHassumption}{supplement.\arabic{assumption}}%
  \renewcommand{\theHdefinition}{supplement.\arabic{definition}}%
  \renewcommand{\theHlemma}{supplement.\arabic{lemma}}%
  \renewcommand{\theHproposition}{supplement.\arabic{proposition}}%
  \renewcommand{\theHtheorem}{supplement.\arabic{theorem}}%
  \renewcommand{\theHcorollary}{supplement.\arabic{corollary}}%
  \renewcommand{\theHremark}{supplement.\arabic{remark}}%
}

\draftSpacing{1.5}

\renewenvironment{proof}[1][Proof]{\noindent{\scshape #1:}\par}{\par\par}

\makeatletter
\def\ps@headings{%
  \leftskip=0pt\rightskip=0pt
  \let\@oddfoot\@empty
  \let\@evenfoot\@empty
  \def\@oddhead{{\small\textit{\MakeUppercase{\@shortTitle\hfil\thepage}}}}%
  \def\@evenhead{{\small\textit{\MakeUppercase{\thepage\hfil\@shortTitle}}}}%
}
\makeatother

\begin{document}

\title{Safe Collective Expression\\with Social Assurance Contracts}
\shortTitle{Safe Collective Expression with Social Assurance Contracts}
\issueName{}
\ifanonymizeforpeerreview
\author{Anonymous}
\else
\author{Matthew Cashman\thanks{MIT Sloan, 100 Main Street, Cambridge, MA, USA. Email: \href{mailto:cashman@mit.edu}{cashman@mit.edu}.}\\[6pt]
\normalfont\normalsize \today\ (\href{\paperversionurl}{latest version here})}
\fi

\begin{abstract}
Controversial views often go unspoken because of fears of retaliation. Speaking up one at a time may encourage others to join in, but retaliation against early participants can stop that cascade before a group large enough to protect itself forms. However, a sufficiently large group could speak up safely if only they spoke up together. A Social Assurance Contract collects endorsements of a controversial view in private and publishes them only when the group is large enough to protect itself. Such contracts can tunnel under the barrier that stops public cascades, forming coalitions that one-by-one public organizing cannot reach.

(JEL D71, D74, D82, D83, J51)
\end{abstract}


\maketitle

\phantomsection\label{sec:intro}

Many people hold views they are not willing to share publicly. Their views may fall outside the range of publicly acceptable opinion, be opposed by an employer, a university, or a hostile public, or concern allegations that invite retaliation. Being publicly identified as endorsing a controversial view can be dangerous enough that a person stays silent---unless, perhaps, enough other people were to speak up with her. A \textit{Social Assurance Contract} lets people arrange to speak up together. The contract specifies a controversial statement (such as an objection to an employer's policy), and then collects endorsements: signatures that give permission to publish that statement under the signer's name. It holds these endorsements in private and publishes them all at once, but only when the group would provide the protection each signer requires. At a firm, such a contract might look like this:

\begin{quote}
\textsc{We, the undersigned, believe [controversial belief] and ask the firm to [requested change]. This letter and its signatures will be published together \textsl{only if at least 200 employees sign.} Until then, no one can see who has signed. If the threshold is not met, the letter and its signatures are destroyed.}
\end{quote}

The contract's \emph{publication condition} specifies when endorsements may be published and whose endorsements may appear. In this example, it requires at least 200 employees to sign. \emph{Failure privacy} means that if the publication condition is never met, no signer's identity is revealed.

This paper focuses on social and professional retaliation in response to the public expression of controversial ideas. The motivating settings are workplaces, universities, and public debate, where speaking as a group can deter retaliation or help its members withstand it. Classic accounts of self-censorship emphasize how beliefs about others' views interact with fear of isolation or retaliation to shape willingness to speak \citep{kuranPrivateTruthsPublic1997,noelle-neumannSpiralSilenceTheory1974}. Information can change the circumstances under which people are willing to speak up, though. For instance, a manager who learns that most employees oppose a policy may become less willing to punish someone who criticizes it. Yet information is not always the missing ingredient: everyone may know that many people privately oppose an established practice, but, simultaneously, it can be the case that no one is willing to become its first public critic.

Indeed, public enforcement of norms can survive despite the underlying norm being unpopular. Third parties punish norm violations even when they were not harmed themselves \mbox{\citep{fehrThirdpartyPunishmentSocial2004}}, and private opposition and public enforcement can even coexist in the same person. Punishing norm violators \textit{and} people who fail to punish can stabilize any individually costly practice---whether or not it benefits the group \citep{boydPunishmentAllowsEvolution1992}, and a weak tendency to copy common behavior when learning from others is enough to sustain punishment of those who fail to punish \citep{henrichWhyPeoplePunish2001}.

Reputational concerns provide another reason for public enforcement. When association with a dissenter damages a person's reputation, even people who privately agree with the dissenter may ostracize her to preserve their own standing \citep{tiroleDigitalDystopia2021}. Consistent with this reputational motive, in the lab we observe people who conformed under pressure later criticizing a lone dissenter in public \mbox{\citep{willerFalseEnforcementUnpopular2009}}.

These mechanisms explain why learning that a view is widely shared need not make it safe to express. The question here is how supporters can turn that private agreement into a protective public coalition. A Social Assurance Contract arranges the expressive act beforehand, forms a coalition large enough to protect itself, and only then publishes its endorsements.

I compare four ways of acting on latent support for a controversial statement in the context of a firm. First is the base case, organizing in public through an \textit{open letter} whose signatures are published as they arrive. Each new name is exposed before the next signer joins. Second, a \emph{call to action} asks employees to add their names publicly at a specified time, without collecting endorsements beforehand. Each person must bet that enough others will follow through. Third, a \textit{credible poll} verifies that respondents are eligible employees, for instance through a firm login, but publishes only the number who support the letter. Finally, a Social Assurance Contract assembles an embargoed letter, publishing its signatures together only when the group would protect every signer.

The first, second, and fourth procedures publish endorsements of the same statement to the same readers, while the poll publishes a count and names no one. What differs is when identities become public and what each person must believe about others before acting. To isolate these differences, I hold fixed the eligible population, its supporters, the opposition, the protection each published group provides, and the readers' response.

The group that speaks together matters because retaliation against an identifiable person, such as dismissal, blacklisting, or public censure, may be diluted or deterred when others speak too. I call the set of people who publicly endorse the statement the \textit{coalition}. It protects a member when speaking with that group is at least as good for her as remaining silent, though protection need not mean immunity from attack. In Section~\ref{sec:target}, for example, expression can be worthwhile even after accounting for the expected cost of retaliation.

Two assumptions govern how coalitions form. The first is safety in numbers: adding public participants never makes an existing member less safe. This lets a contract open to every eligible signer accept a new signature without rechecking, each time, that everyone already committed would still be protected. The second is follow-through uncertainty: even when several people plan to publish their endorsements together, any one person may follow through while others do not. I impose the same safety standard on public organizing and on the contract: no participant is ever named in a coalition that does not protect her, including when others do not follow through. This regret-free standard is appropriate when retaliation falls on individuals, cannot be undone, and is difficult to compensate.

\runinhead{Cascade versus tunnel.} These requirements allow us to distinguish a safe \textit{destination} from a safe \textit{path}. In this model, each new signature means public organizing must pass through an interim coalition that protects everyone already exposed. Each new signature also adds to the group's protection and encourages others to join, so the public cascade of signatures can grow. However, it stops when no additional person is willing to sign. This stopping point is the \textit{exposure barrier}. A Social Assurance Contract can tunnel beneath this barrier to a larger self-protecting group without exposing the people who would otherwise have to join first (Figure~\ref{fig:cascade_tunnel}).

\begin{figure}[t]
\centering
\begin{tikzpicture}[scale=0.95]
  \draw[->] (0,0) -- (6.7,0);
  \draw[->] (0,0) -- (0,6.4) node[above] {\footnotesize Protected mass $G_S(y)$};
  \node[below] at (3.35,-0.46) {\footnotesize Current public mass $y$};
  \draw[dash dot,gray!70] (0,0) -- (6,6);
  \draw[very thick] plot[smooth] coordinates
    {(0.6,1.7) (1,2.0) (2.5,2.5) (3.3,2.95) (4,4) (4.62,4.95) (5.5,5.5) (6.2,5.65)};
  \node[right] at (6.2,5.65) {\footnotesize Protected-size curve};
  \fill (2.5,2.5) circle (2.2pt);
  \draw[fill=white] (4,4) circle (2.2pt);
  \fill (5.5,5.5) circle (2.2pt);
  \draw[dotted] (1,0) -- (1,1);
  \draw[dotted] (2.5,0) -- (2.5,2.5);
  \draw[dotted] (5.5,0) -- (5.5,5.5);
  \draw (1,0.09) -- (1,-0.09);
  \draw (2.5,0.09) -- (2.5,-0.09);
  \draw (5.5,0.09) -- (5.5,-0.09);
  \node[below] at (1,0) {\footnotesize $b$};
  \node[below] at (2.5,0) {\footnotesize $y^-$};
  \node[below] at (5.5,0) {\footnotesize $y^+$};
\node[below] at (4,0) {\footnotesize $y^u$};
  \node[below right] at (2.5,2.5) {\footnotesize $\Cs$};
  \node[above left] at (5.5,5.5) {\footnotesize $\Cp$};
  \draw[dotted] (4,0) -- (4,4);
  \draw[|-|] (2.5,0.3) -- (4,0.3);
\node[above] at (3.25,0.32) {\footnotesize $y^u-y^-$};
  \draw[<-] (3.5,3.05) -- (4.35,2.6) node[right] {\footnotesize Exposure barrier};
  \draw[->,thick] (1,1) -- (1,2.0) -- (2.0,2.0) -- (2.0,2.33) -- (2.45,2.33);
  \node at (2.35,1.7) {\footnotesize Cascade};
  \draw[->,thick,densely dashed] (2.5,2.5) to[bend left=28] (5.5,5.5);
  \node at (4.95,3.9) {\footnotesize Tunnel};
\end{tikzpicture}
\caption{Cascade versus tunnel}
\label{fig:cascade_tunnel}
\par\medskip\noindent\begin{minipage}{\textwidth}\footnotesize \emph{Notes:} The horizontal axis $y$ is the total mass of support (the weighted number of people, simply the headcount when everyone counts equally) already public, including the vanguard $b$. The vertical axis $G_S(y)=b+H_S(y)$ is the total mass that would be public if every supporter protected at mass $y$ joined. Crossings with the diagonal are fixed points. Growth resumes from just above the hollow crossing, the unstable middle crossing $y^u$, but not from the filled crossings; the bracket marks the width of the exposure barrier, $y^u-y^-$. Starting from $b$, public organizing reaches the lower crossing $y^-$ and stalls at the exposure barrier. With all relevant endorsements collected, the Social Assurance Contract reaches the largest self-protecting coalition at $y^+$, adding mass $y^+-y^-$.
\end{minipage}
\end{figure}
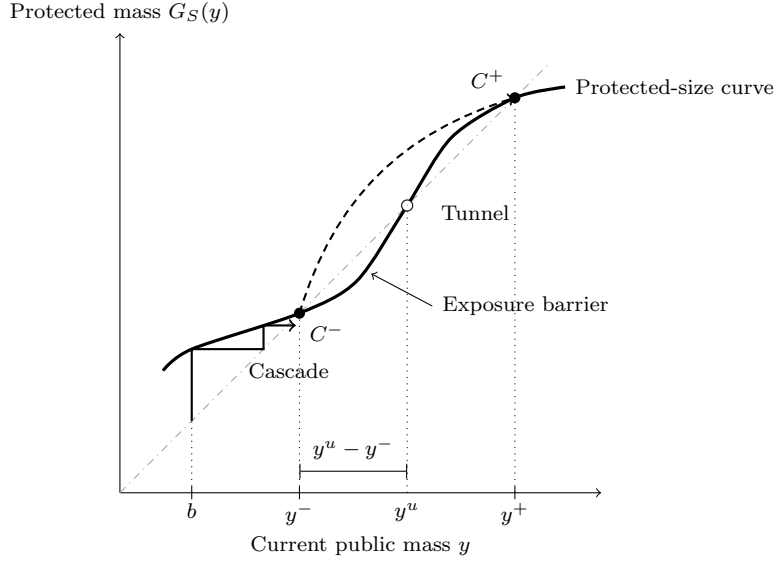

\runinhead{Main results.} Starting from the \textit{vanguard}, a self-protecting group whose endorsements are already public, Section~\ref{sec:model} identifies both the largest coalition that can form through regret-free public entry and the largest that a contract can assemble from the endorsements it receives. With all supporters signing, the contract reaches the largest protective coalition available in the population. Its additional members beyond the public cascade are the \textit{hidden supporters}. Section~\ref{sec:representation_necessity} then establishes what crossing the exposure barrier safely requires. When any person may follow through without the others, a crossing must rely on endorsements received privately before their joint publication. After the exposure barrier has been crossed with such a mechanism, public organizing can resume.

The comparison changes when participants accept some risk of speaking without enough others to protect them. With sufficiently reliable turnout, a call to action could cross the exposure barrier with high probability (Section~\ref{sec:positive_tolerance_frontier}). A contract avoids this turnout gamble, but collecting endorsements privately does create a risk of leaks. The publication rule itself also does not guarantee that supporters will sign: protection applies to endorsements actually received, not intentions to sign. Sections~\ref{sec:custody_frontier} and~\ref{sec:scope_axioms} address custody and participation.

The institutions also differ in whom they bring into public view. When protection depends only on the group's size and all supporters sign, hidden supporters require more protection than the public cascade's members. Section~\ref{sec:target} connects those requirements to the opposition's decision about whether to retaliate. In that setting, the additional members have more to lose or less to gain from speaking. Bringing them into view can therefore both raise the amount of retaliation observed and reduce total suppression of expression, because silence previously kept them from becoming targets.

\runinhead{Related literature.} Social Assurance Contracts adapt a familiar monetary idea to the social domain. Monetary assurance contracts and provision-point mechanisms collect contributions conditionally and return the money if total participation falls short of what is required \citep{bagnoliProvisionPublicGoods1989,tabarrokPrivateProvisionPublic1998,zubrickasProvisionPointMechanism2014}. Suppose a filmmaker needs \$20,000 to produce a documentary. Supporters contribute to a fund that releases the money to the filmmaker only if it reaches \$20,000 by a specified deadline. If it does not, their contributions are refunded. Contingent contracts and dynamic contribution games address related participation and commitment problems \citep{segalContractingExternalities1999,admatiJointProjectsCommitment1991}. In contrast to crowdfunding, once a complainant, dissenter, or whistleblower is publicly identified, returning a deposit cannot make her anonymous again. Information escrows for allegations of wrongdoing collect information privately and release it conditionally \citep{ayresInformationEscrows2012}. Building on that work, this paper asks which protective coalitions such an institution can assemble beyond the reach of the public cascade.

Another strand in the literature asks how information changes willingness to speak or act. Public signals, what people know about one another's knowledge, and strategic control of information can change expectations about others' participation \citep{lohmannDynamicsInformationalCascades1994,chweStructureStrategyCollective1999,edmondInformationManipulationCoordination2013}. Within this broader concern about beliefs, the social image literature examines how concern about others' opinions changes the costs of public expression \citep{benabouIncentivesProsocialBehavior2006,aliImageInformationChanging2020,bursztynMisperceivedSocialNorms2020}.

\citet{braghieriThresholdDisclosureCollective2026} connect information to conditional public identification and provide experimental evidence for the paper's two behavioral premises: first, that publicly linking people to their views discourages controversial expression, and second, that people condition their willingness to have their vote revealed on the share of voters who agree. The authors asked university students to vote on whether transgender women should be allowed to compete in women's collegiate sports and varied whether each student's name was linked to her vote. Their threshold-disclosure treatment lets each voter choose the share of like-minded votes required for her own vote to be revealed. That share counts agreeing votes even when those votes remain private. That counting rule matters because, in the present model, protection depends on the coalition that actually becomes public. Proposition~\ref{prop:bbf} in the Supplemental Appendix shows what happens when that dependence is absent: if protection depends only on a statistic outside the coalition's control, the cascade and tunnel reach the same group. Unlike Braghieri and coauthors' model, in which the vote share depends on how people vote, this boundary comparison holds the statistic fixed. Section~\ref{sec:target} discusses the experimental evidence from their study.

The public cascade considered here follows the threshold models of \citet{granovetterThresholdModelsCollective1978} and \citet{schellingMicromotivesMacrobehavior1978}: people join as the participating group reaches their individual thresholds. Starting with an initial active group, this process grows until it forms a group that attracts no further members, a fixed point. This is also the critical-mass account of collective action \citep{marwellCriticalMassCollective1993}.\footnote{It is not the informational cascade of \citet{bikhchandaniTheoryFadsFashion1992}, in which people abandon private signals to follow predecessors.} Standard fixed-point methods identify the smallest and largest such groups when adding participants never makes others less willing to join \citep{topkisSupermodularityComplementarity1998}. Related results for supermodular games, in which others' actions increase the incentive to act, bound learning dynamics by the smallest and largest equilibria \citep{milgromRationalizabilityLearningEquilibrium1990}. This paper adds the requirement that each public step protect the people exposed at that step. The lower fixed point is then a limit on safe public expression, not merely an outcome selected by learning.

\citet{tabarrokSimpleModelCrime1997} models how greater criminal activity strains enforcement, lowering individual punishment risk and encouraging further activity. He extends this logic to strikes and revolutions even when participants are well informed. That is safety in numbers arising from strained enforcement. This paper takes safety in numbers as given and asks how supporters can assemble a protective group without exposing its early members before that protection exists.

The closest structural comparison is \citet{galeMonotoneGamesPositive2001}. Gale studies irreversible public participation in which one person's action makes participation more attractive to others, but he evaluates only the terminal action profile. This means an early participant can accept temporary exposure on the strength of equilibrium expectations that others will follow. Here, attribution is irreversible and retaliation may arrive before later participants do, so an early participant cannot rely on others' promises. Remark~\ref{rem:gale} in Section~\ref{sec:necessity} states the contrast precisely.

The divide-and-conquer literature studies the same barrier from the principal's side. Discriminatory contracts and targeted subsidies can induce participation in coordination problems by rewarding the agents who matter most to the rest \citep{segalCoordinationDiscriminationContracting2003,winterIncentivesDiscrimination2004,halacRaisingCapitalHeterogeneous2020,sakovicsWhoMattersCoordination2012}. In a public recruitment process, the opposition can instead concentrate retaliation on whoever is exposed first. The tunnel avoids that exposure because no new endorsement appears in public until the coalition that protects it has formed.

\section{The Model: Protective Coalitions, Cascade, and Tunnel}\label{sec:model}

Which protective groups can be assembled without leaving anyone unprotected along the way? First, I consider a small example that shows why a safe destination need not have a safe public path leading to it. The general model then identifies the largest self-protecting group, establishes the limit of public organizing, and gives the contract's publication rule.

\subsection{A Worked Example}\label{sec:working_example}

Suppose one person has already signed a public letter opposing an employer's policy and finds speaking up worthwhile even alone. Four others support the statement but need company before they are willing to appear on the letter. Counting the person already public and each new signer herself, the first supporter needs at least two public names, the second needs three, and the last two each need five in order for signing the letter to be worthwhile.

Public organizing can recruit the first supporter, bringing the total to two, and then the second, bringing it to three. It stops there. Either of the last two would bring the total only to four if they signed alone, short of the five each one needs. Yet both would be protected if they appeared together, bringing the total to five. The destination offers sufficient protection from retribution, but the next public step does not.

A call to action asking the last two to sign at noon does not remove this difficulty: either could follow through while the other does not. A Social Assurance Contract instead collects both endorsements privately and publishes neither until it holds both. Joint publication then raises the public total from three to five without exposing either person in the smaller group.

Now add a fifth supporter who requires six public names. Once the contract has published the two endorsements, this person can safely join through the public letter: her own signature brings the total from five to six. The contract need not collect her endorsement. Private custody is needed to cross the exposure barrier, after which ordinary public organizing could resume---though it is possible for private custody to handle these names, too.

The example separates three tasks: finding a protective group, reaching it without stranding early participants, and collecting the endorsements needed to cross the barrier. The general model keeps these tasks separate while allowing people to differ in how much protection they need and how much protection they provide.

\subsection{Protective Coalitions}\label{sec:safe_coalitions}

\runinhead{Primitives.} Let $C_0$ be the vanguard, the group whose support is already public when recruitment begins, and let $N$ be a fixed finite set of other people eligible to join. Thus $C_0\cap N=\varnothing$. The \emph{supporter set} $S\subseteq N$ contains those willing to have their names on the published statement if that coalition protects them. A supporter has not necessarily signed or committed. Section~\ref{sec:scope_axioms} establishes when supporters sign, and the Supplemental Appendix separately analyzes public participation.

I measure groups by their \emph{mass}: this is the number of people in a coalition, weighted when some count for more than others, and simply the headcount when everyone counts equally. Each person $i\in C_0\cup N$ has positive mass $a_i$, and a coalition has total mass $m(C)=\sum_{i\in C}a_i$. Write $b=m(C_0)\ge0$ for the vanguard's mass. A public coalition includes the vanguard and any supporters who have joined it, so $C_0\subseteq C\subseteq C_0\cup S$. Its total mass $m(C)$ therefore already includes the vanguard. The vanguard may be empty.

The vanguard gives public organizing a place to start. If no one is safe alone and there is no vanguard, organizing in public cannot begin even though some coalition would protect every member. Even in an environment deeply hostile to dissent, a Social Assurance Contract can still tunnel to the largest protective coalition.

Two groups outside the coalition matter. The \emph{opposition} is whoever can punish a person for publicly supporting the statement: an employer, a university, an accused party, or a hostile public. It can be a single decision-maker or a diffuse group. The model records which coalitions protect people from its response, and Section~\ref{sec:target} later models the opposition's choice of whether to retaliate. \emph{Allies} are whoever would act to protect the people named as supporters, from sympathy or from duty. These could be a court, a labor board, the press, or the public.

To collect endorsements from supporters and manage disclosure to these audiences, the contract requires an \textit{architect} and an \textit{administrator}. The \emph{architect} chooses the statement, the population eligible to sign, and the publication condition, while the \emph{administrator} receives endorsements and publishes them. These may be the same person or a cryptographic system that requires trusting no single party.\footnote{A companion paper discusses practical implementation, including cryptographic systems for collecting endorsements privately and publishing them conditionally \citep{cashmanConditionalDisclosureCoordination}.}

\runinhead{Protection.} For each person $i\in C_0\cup S$, the \emph{protection family} $\Pfam_i\subseteq\{C:C_0\subseteq C\subseteq C_0\cup S,\ i\in C\}$ contains the public coalitions in which she is at least as well off speaking as remaining silent. I take the vanguard to be self-protecting at the outset: $C_0\in\Pfam_i$ for every $i\in C_0$. Under monotone protection below, its members remain protected as the coalition grows. Recruitment therefore needs to check protection only for newcomers, although the guarantee covers the entire coalition. Retaliation falls on identifiable people. In public organizing and contract release, safety is evaluated when both the statement and the signer become public. For the calls to action in Section~\ref{sec:positive_tolerance_frontier}, it is evaluated against those who actually participate in the appointed round. Both evaluations use the same protection families, and neither counts participation in later rounds. The coalition $C$ therefore records who has spoken publicly. Protection families and the opposition remain fixed while the coalition forms. Legal reforms, a change in audience beliefs, or different opposition would change which coalitions protect whom and would therefore define a new protection environment.

A group may protect each of its members, but a person deciding whether to speak does not necessarily know whether that group will form. When answering a public call to action, she must judge whether the others she needs will follow through. A Social Assurance Contract instead waits until it has the endorsements needed for a protective coalition before publishing them together. Whether she finds signing worthwhile remains a separate question.

\begin{assumption}[Monotone protection]\label{ass:M}
For all $i$ and all $C,C'$ with $C\in\Pfam_i$, $C\subseteq C'$, and $i\in C'$: $C'\in\Pfam_i$.
\end{assumption}

\runinhead{Monotone protection.} Monotone protection is safety in numbers: adding people to a public coalition never makes an existing member less safe. It is also what makes a Social Assurance Contract manageable. Once the architect fixes the statement to be published and the eligible class, the administrator can accept another eligible signer without first checking that every incumbent will still be protected at the publication condition, or asking the incumbents to consent again. If a new member could make others less safe, the administrator would instead have to decide whom to admit (and potentially turn people away, which could leak information about who is already in the contract), or reveal information about who has already signed to potential endorsers. Though new entrants can in principle make incumbents less safe, a Social Assurance Contract does not seem workable without monotone protection. Section~\ref{sec:target} gives an economic explanation. In that model, larger coalitions are harder to suppress, so adding members can both discourage retaliation and make it less likely to succeed. Section~\ref{sec:scope} studies what changes when an added name can make others less safe.

Monotone protection does not require everyone to contribute equally. The protection a group provides can depend on who belongs to it. A senior employee, corroborating witness, tenured professor, or legally protected organizer may add more protection than another person of equal mass. For a simpler case, suppose that total public mass alone determines protection. This is the \emph{scalar protection case} below: one number, the group's mass, is enough to evaluate each member's safety, so we can represent the model with a single curve. The general formulation instead records every named coalition that protects each person, allowing identity and coalition composition to matter. Section~\ref{sec:scope} extends the analysis to cases in which stronger opposition, strategic allegations, or guilt by association make an additional member harmful to others.

\begin{definition}
A coalition $C$ with $C_0\subseteq C\subseteq C_0\cup S$ is \emph{self-protecting} if it protects every member, that is, if $C\in\Pfam_i$ for every $i\in C$. I refer to such groups as \emph{protective coalitions}. The vanguard $C_0$ is self-protecting by the starting condition. If it is empty, that condition is vacuous. A \emph{regret-free expansion} maintains protection at every step as new members join the public coalition.
\end{definition}

\runinhead{The scalar protection case.} Here, one number, the protection requirement $\rho_i$, summarizes each person's protection family. Person $i$ is then safe in coalition $C$ exactly when its total public mass reaches her requirement: $C\in\Pfam_i$ iff $m(C)\ge\rho_i$. Vanguard members have $\rho_i\le b$, so the initial coalition already protects them. The function $H_S(y)=m\{i\in S:\rho_i\le y\}$ is the total mass of supporters protected when public mass is $y$.

In the \emph{nonatomic} case, agents form a continuum, each individual has negligible mass, and $H_S$ is continuous over the relevant interval. A self-protecting coalition that includes everyone protected at its size satisfies
\begin{equation}\label{eq:granovetter}
y = b + H_S(y),
\end{equation}
the threshold-cascade equation of \citet{granovetterThresholdModelsCollective1978}. At this fixed point, the coalition includes everyone protected at that mass, and its mass protects no one else.

Figure~\ref{fig:cascade_tunnel} graphs the right side as $G_S(y)=b+H_S(y)$. Crossings with the diagonal are fixed points. Write $\Cs$ for the coalition reached by the public cascade and $\Cp$ for the largest self-protecting coalition. The construction below defines each precisely. Then $y^-=m(\Cs)$ and $y^+=m(\Cp)$ denote the smallest and largest crossings. A crossing is locally stable if a small increase above it causes no further growth; just above such a crossing, $G_S$ lies below the diagonal.

Outside the continuum, this graphical reading needs one refinement. With finitely many people, entrant $i$ contributes her own mass when she joins. She is safe when $m(C)+a_i\ge\rho_i$. The public cascade therefore uses the lower activation threshold $\rho_i-a_i$, whereas a coalition released together must satisfy the original requirement $\rho_i$ (Section~\ref{sec:gap}).

\runinhead{Safe paths and safe destinations.} Two rules keep the path question separate from the destination question. Each takes a coalition as its input and returns a set of people. Mathematically, such a rule is an \emph{operator}. The safe-expansion operator
\[
T(C) = C \cup \{\, i\in S\setminus C : C\cup\{i\}\in\Pfam_i \,\}
\]
adds exactly the people who can safely join the currently public set $C$. It represents one round of public expression. The second operator
\[
\Gamma(C) = C_0\cup\{\, i\in S : C\cup\{i\}\in\Pfam_i \,\}
\]
keeps the vanguard and collects every other supporter who would be safe in $C$ or after joining it. Both operators act on coalitions between $C_0$ and $C_0\cup S$. $\Gamma$ is the static safety check. For $i\in C$, $C\cup\{i\}=C$, so $C$ is self-protecting iff $C\subseteq\Gamma(C)$. It is a fixed point, $C=\Gamma(C)$, when it protects its members and no one else can safely join.

The operators answer different questions. $\Gamma$ asks who would be safe if a candidate coalition appeared at once; its greatest fixed point gives the largest safe destination. $T$ asks who can safely join the coalition already public; iterating it traces a safe path. Because a new participant may be the only newcomer to appear, $T$ tests the existing coalition plus her name. A technology that guarantees joint publication removes the need for this one-at-a-time test. Two institutions facing the same supporters and protection families can therefore reach different coalitions.

\begin{lemma}[Monotonicity]\label{lem:mono}
Under Assumption~\ref{ass:M}, $C\subseteq C'$ implies $\Gamma(C)\subseteq\Gamma(C')$ and $T(C)\subseteq T(C')$, and $T$ is inflationary ($C\subseteq T(C)$).
\end{lemma}

Because $T$ only adds people and $S$ is finite, the sequence $C_0\subseteq T(C_0)\subseteq T^2(C_0)\subseteq\cdots$ stops changing after at most $|S|$ rounds. It stops at the \emph{cascade coalition} $\Cs(S)=T^{|S|}(C_0)$. The exposure barrier is where the safe path ends: at $\Cs(S)$, no one outside it is safe to join in public. The largest self-protecting coalition instead combines all protective coalitions:

\[
\Cp(S) = \bigcup\{\, C:C_0\subseteq C\subseteq C_0\cup S,\ C \text{ is self-protecting} \,\}.
\]

For any available supporter set $W\subseteq S$, use $\Cs(W)$ and $\Cp(W)$ for the same constructions with $W$ in place of $S$, keeping $C_0$ fixed. In particular, $\Cp(\varnothing)=C_0$: without new signatures, the public coalition remains the vanguard.

\runinhead{The largest protective coalition.} In the example in Section~\ref{sec:working_example}, write $v$ for the initial signer, so $C_0=\{v\}$ and $b=1$. Everyone has unit mass, the vanguard member has $\rho_v\le1$, and the four potential recruits have requirements $\rho=(2,3,5,5)$. The cascade reaches $\Cs=\{v,1,2\}$, whereas $\Cp=C_0\cup S=\{v,1,2,3,4\}$ contains the vanguard and all four recruits. The next theorem establishes that a unique largest protective coalition exists under monotone protection, even when protection depends on the identities of the members as well as their mass.

\begin{theorem}[Largest Self-Protecting Coalition]\label{thm:largest}
Under Assumption~\ref{ass:M}, the union of any nonempty family of self-protecting coalitions is self-protecting; hence $\Cp(S)$ is the unique largest self-protecting coalition. Any permission-respecting release rule---one that publishes no name without that agent's permission---that satisfies regret-free safety (formalized in Section~\ref{sec:res}) can expand the public coalition only to a subset of $\Cp(S)$.
\end{theorem}

Theorem~\ref{thm:largest} answers the destination question. Under monotone protection, protective coalitions can be combined, so their union $\Cp$ is protective and contains every other protective coalition. Whether its potential recruits sign is a separate question, taken up in Section~\ref{sec:scope_axioms}. Before turning to whether public organizing can reach $\Cp$ and what an institution must hold to do so, consider why monotonicity matters.

Without monotone protection, two protective coalitions may be incompatible: their union can leave a member unprotected. A contract must then choose among groups, because no unique largest coalition exists. Supplemental Appendix Section~\ref{sec:scope} gives an example and characterizes which coalitions require private collection in this setting.

\subsection{Regret-Free Safety}\label{sec:res}

Public organizing and conditional release face the same safety requirement in this setting, but differ in whether the institution can receive an endorsement before exposing its author.

For a potential recruit $i\in S$, \emph{unprotected exposure} means being identified as a supporter without the protection she requires from the public coalition at the time her safety is assessed, defined below. Let $r_i$ denote the probability of this event she will tolerate, and let an \emph{admissible realization} be a pattern of who goes public and who does not, as described below, on which the institution promises protection. An escrow, lawyer, or organizer can offer this disclosure guarantee without assuming that everyone rejects every risky prospect. Section~\ref{sec:positive_tolerance_frontier} allows $r_i>0$ and interprets this tolerance through the benefit of expression and the harm from unprotected exposure.

\runinhead{Follow-through and verification.} People may intend to speak together, but some may not follow through. The danger is that a person publishes her endorsement while others she was counting on do not.

Time is divided into dates $t=1,2,\dots$. The set $C_t$ contains everyone whose statement is public under her name by the end of date $t$, including the vanguard present at date zero. Because publication is irreversible, $C_0\subseteq C_1\subseteq\cdots$. At date $t$, the set $A_t\subseteq S\setminus C_{t-1}$ contains the people who choose to try to publish the statement under their names. That decision need not lead to publication. A person may change her mind, or a practical obstacle may prevent her statement from appearing. I say that a public action \emph{completes} when the statement becomes public under its author's name. The set $R_t\subseteq A_t$ contains the people whose public actions are completed, so $C_t=C_{t-1}\cup R_t$. For a private commitment, completion instead means that the endorsement reaches its intended recipient. A co-participant is \emph{verified before exposure} when the institution has received and certified her signature rather than relying on her stated intention. Verification establishes identity and ties the signature to the exact statement and publication condition. An \emph{unverified attempt} is a person's attempt to publish her own endorsement before the institution holds the other endorsements needed to protect her. The assumptions below state which patterns of completed and uncompleted actions the institution must handle.

\begin{assumption}[Follow-through uncertainty]\label{ass:U}
Unless the institution has verified co-participants before exposure, any subset of a date's attempted public actions may not be completed: for every $R\subseteq A_t$, the realization in which exactly the actions by $R$ are completed is admissible. In particular, for every $i\in A_t$, the realization in which only $i$'s public action is completed, $R_t=\{i\}$, is admissible.
\end{assumption}

The assumption concerns possibility, not probability. Under the disclosure guarantee, the institution must protect a participant even when she is the only person who follows through. Coordination may make that outcome less likely, but the guarantee must still cover it. The bounds below (Lemma~\ref{lem:reduction}, Theorem~\ref{thm:cascade}(i), and Theorem~\ref{thm:necessity}) use only that case, so they hold under any weaker assumption that keeps it.

A Social Assurance Contract avoids relying on those intentions by collecting endorsements before publishing them together. Section~\ref{sec:necessity} shows why this guarantee requires private custody, and the later analysis allows people to risk speaking without enough protection when others do not follow through.

Organizers have other useful tools. Deniable inquiries and polls can reveal support, repeated interaction can build trust, and costly pledges can make later action more likely. These tools may reduce the expected cost of speaking alone or reduce the probability that others fail to join her, but they do not guarantee protection at publication unless they either make lone expression safe or set up joint publication of sufficient endorsements to make a protective coalition.

\begin{assumption}[Regret-free safety]\label{ass:RES}
For every date $t$, every admissible realization path, and every $i\in C_t$: $C_t\in\Pfam_i$.
\end{assumption}

This assumption is the institution's disclosure guarantee: no participant is ever named in a coalition that fails to protect her. Requiring zero probability of unprotected exposure, as in Section~\ref{sec:positive_tolerance_frontier}, is slightly weaker: it ignores outcomes of the administrator's randomization that have probability zero. Supplemental Appendix Section~\ref{sec:acceptable_risk_supp} shows that the two otherwise coincide. Accepting some chance of unprotected exposure (Section~\ref{sec:positive_tolerance_frontier}) relaxes the guarantee further. Expected-utility protection (Section~\ref{sec:target}) instead specifies what counts as a protective coalition, namely that speaking with that group can be worthwhile despite some chance of retaliation. The guarantee then requires each coalition in which a person becomes public to meet that criterion. It does not allow a poor turnout to be averaged against a favorable one.

To test this guarantee, we need not examine every possible group separately. Under monotonicity, the singleton outcome in Assumption~\ref{ass:U} is enough to test every realization.

\begin{lemma}[Reduction to singleton safety]\label{lem:reduction}
Consider a process in which every current attempt is unverified before exposure. Under Assumptions~\ref{ass:M} and~\ref{ass:U}, the process is regret-free if and only if for every date $t$ and every $i\in A_t$, $C_{t-1}\cup\{i\}\in\Pfam_i$; equivalently, iff $C_t\subseteq T(C_{t-1})$ for every realization.
\end{lemma}

Each person going public must therefore be safe with the existing public coalition plus herself, whether she moves on her own or as part of a planned simultaneous action, and the next theorem turns this test into the limit of public organizing.

\subsection{Cascade and Tunnel}\label{sec:cascade_tunnel}

The model now has a common safety standard and two candidate destinations: the largest safe destination and the coalition that public organizing can build. The next theorem identifies how far regret-free public organizing can go. Its attainment part uses a ``greedy'' process, in which everyone who is willing to join the current public coalition joins at once, and this repeats until no one else is willing to join.

\begin{theorem}[Cascade Feasibility]\label{thm:cascade}
Under Assumptions~\ref{ass:M} and~\ref{ass:U}:
\begin{enumerate}
\item[(i)] (\emph{Bound.}) Every decentralized public process described in Section~\ref{sec:res}, whose transitions use only prior public exposure and unverified current attempts, has $C_t\subseteq\Cs(S)$ for every $t$ and every realization when it satisfies regret-freeness.
\item[(ii)] (\emph{Attainment.}) The greedy process $A_t=T(C_{t-1})\setminus C_{t-1}$ satisfies regret-freeness and, in the realization in which every attempted action is completed, reaches exactly $\Cs(S)$ in at most $|S|$ dates.
\item[(iii)] (\emph{Lattice position.}) $\Cs(S)$ is the least fixed point of $\Gamma$ and $\Cp(S)$ is the greatest; the fixed points of $\Gamma$ form a complete lattice on the domain $\{C:C_0\subseteq C\subseteq C_0\cup S\}$.
\end{enumerate}

\end{theorem}

Theorem~\ref{thm:cascade} gives the exact limit of regret-free public organizing. No safe one-by-one process can pass $\Cs$, and the greedy process, which immediately invites everyone who can safely join, reaches it. Section~\ref{sec:scope_axioms} reports why strategic public expression selects this coalition (Supplemental Appendix Theorem~\ref{thm:selection}).

To go beyond this public cascade, the tunnel requires the published coalition to protect every member. To formalize the tunnel, let $V\subseteq S$ be the people whose valid signatures the administrator has received. The release rule publishes the signed statements of a subset $D(V)\subseteq V$ together, choosing it so that $C_0\cup D(V)$ is self-protecting. The vanguard is already public, so $D(V)$ records only the newly published endorsements. The rule publishes no other names, and all unreleased signatures remain private.

\begin{assumption}[All-or-nothing release]\label{ass:atomic}
The administrator can publish the endorsements of a coalition in a single event, rather than as a sequence of individual exposures subject to Assumption~\ref{ass:U}.
\end{assumption}

All-or-nothing release is one public event. Having received every endorsement does not itself guarantee this output: the administrator must also be able to publish the protecting coalition together to the relevant audience. Without that ability, a partial publication is a leak, covered by Section~\ref{sec:custody_frontier}. The object being published matters here: a jointly published signed statement supplies the modeled expression itself, whereas a promise to attend a later strike or protest still leaves a separate execution stage. Lemma~\ref{lem:custody} derives the prior custody required for joint publication within the stated environment.

\begin{theorem}[Tunnel Implementation]\label{thm:tunnel}
Under Assumptions~\ref{ass:M} and~\ref{ass:atomic}, the rule $D(V)=\Cp(V)\setminus C_0$ satisfies regret-freeness for every realized participation $V$ and produces the public coalition $\Cp(S)$ when all supporters sign. The difference $m(\Cp(S))-m(\Cs(S))$ is nonnegative and is strictly positive iff $\Cs(S)\subsetneq\Cp(S)$, equivalently iff $\Gamma$ has multiple fixed points.
\end{theorem}

The release rule is constructive when protection can be checked. Start with all received signers, $V_0=V$, and remove anyone the current group does not protect:
\[
V_{k+1}=\{i\in V_k:C_0\cup V_k\in\Pfam_i\}.
\]
Repeat until no one is removed, always counting the vanguard's protection. Every self-protecting coalition $C\subseteq C_0\cup V$ survives every round: if its recruits $C\setminus C_0$ are contained in $V_k$, monotonicity makes $C_0\cup V_k$ protect them. The sequence stops after at most $|V|$ deletion rounds. If $V_\ast$ is its terminal signer set, $C_0\cup V_\ast$ is self-protecting and contains every other protective coalition available from these signers, so it equals $\Cp(V)$. The endorsements published are $D(V)=V_\ast$. This computes the release from known protection families. It does not elicit them. Proposition~\ref{prop:fullsign} separately addresses participation under the costless signing conditions.

Applying this rule requires a practical protection test and a credible administrator. Section~\ref{sec:scope_axioms} discusses these requirements alongside the decision to sign.

\runinhead{Credible polls.} An administrator could verify signatures and publish their \textit{count} instead of the signers' identities. The result is the credible poll defined in the Introduction. In the scalar protection case a count-only mechanism can privately assemble the upper-fixed-point coalition and publish its mass $y^+$ (Corollary~\ref{prop:count}). A count can change what the opposition believes, but it does not tell supporters whom they can coordinate with---the critical distinction. Joint publication creates a coalition that can act together, which is why named coalitions are the object here (Section~\ref{sec:count_release}).

\section{Crossing the Exposure Barrier}\label{sec:necessity}

The limit on safe public organizing and the contract's release rule answer the first question: which protective coalitions can form? This section answers the second: what must an institution do differently to cross the exposure barrier? Section~\ref{sec:representation_necessity} shows why a crossing that guarantees protection requires \textit{private custody} before publication. Private custody means receiving and safeguarding endorsements that permit publication of a specified statement under a stated publication condition. Sections~\ref{sec:positive_tolerance_frontier} and~\ref{sec:custody_frontier} then compare the risks of a public call to action and private collection when people accept some chance of speaking without enough protection, and Section~\ref{sec:scope_axioms} brings together participation, practical publication rules, and the institutional conditions behind these comparisons.

\subsection{Representation and Guaranteed Protection}\label{sec:representation_necessity}

The coalition model identifies protective coalitions, and we now ask what an institution must do to reach one without leaving people unprotected along the way. Social Assurance Contracts are designed for environments in which retaliation is triggered when a person's endorsement of a specified controversial statement is publicly attributed to her, but before the publication condition is met the endorsement is kept beyond the opposition's reach. If the opposition can observe or punish based on that private endorsement itself, the mechanism cannot provide failure privacy.

To identify what a safe crossing requires, I now specify a general environment describing the actions available to participants and institutions. The same rules apply whether endorsements are handled by a lawyer, a server, a cryptographic protocol, or a trusted person. The question is what an institution must have received before it can publish a statement under several people's names at once. Each action, such as endorsing a statement through a Social Assurance Contract or attending a protest, has an author and a recipient. It can be completed while other actions are not, and its effects cannot be undone once it is complete. For each action, what it would make public is fixed in advance: a post under a person's name, if published, makes that name public and nothing else. Information moves only through completed actions, so a person or device can act only on the public history and on the private commitments it has actually received.

Under these rules, a private commitment reaches its recipient without exposing its author; a public action or a release makes designated names visible; and first naming anyone other than oneself requires a permitting action from that person that the institution has received. The lemmas show that this action is a private commitment, still unpublished when the release happens. Holding the opposition and the protection families fixed isolates the technological question: can an institution guarantee joint attribution without first receiving private commitments? Assumption~\ref{ass:causal} in Appendix~\ref{app:representation} states this environment formally, as environment E1, and Supplemental Appendix Section~\ref{sec:representation_details} discusses its scope.

This account of attribution distinguishes publication of a signed statement from suspicions of dissent. A signed statement directly records the signer's support and can be cited, circulated, or used by institutions. In comparison, an allegation (or an inference to that effect) can be denied, though suspicion may still identify a supporter well enough to trigger retaliation. The model records that event as a leak rather than as successful public expression, because it can harm the person without her joining the public coalition. More generally, a \emph{leak} is any route by which the opposition learns, or can infer, that a particular person has endorsed the statement before release: theft or carelessness in custody, inference from metadata or a side channel, or suspicion strong enough to trigger retaliation. The model treats every such identification the same way.

Two lemmas, stated in Appendix~\ref{app:representation} and proved in Appendix~\ref{app:proofs}, do the work. The first, the attribution dichotomy (Lemma~\ref{lem:dichotomy}), says that every completed action by a participant in this environment is one of two things: an \emph{exposure action}, which identifies its author when it is completed, or a \emph{private commitment}, which reaches a recipient who can verify it while its author stays hidden. The institution's own publications are a third kind of action, a device action that names others, and the environment's permission and receipt rules govern it. The second, the custody lemma (Lemma~\ref{lem:custody}), says that an action first publishing endorsements under several people's names at once requires the releasing party to hold their still-private commitments, except possibly its own. Endorsements already public need no new private commitment. Without such custody, every first attribution is a public action that can be completed on its own, governed by Assumption~\ref{ass:U}.

A cryptographic bulletin board, an embargoed letter, and a trusted delegate use the same technology. Each receives private commitments---a pseudonymous filing, a signed instruction, an encrypted message---and controls the one release that reveals their authors together. By contrast, separate public posts backed by a common promise remain separate attempts because any one of them can be the only one to appear. In each custody example, later publication requires permission tied to the statement and the publication condition. An action that instead changes legal protection or audience beliefs changes the protection environment. It does not provide another way to publish the same endorsements. The definition below calls any person or device controlling this process the administrator.

\begin{definition}[Expression mechanisms and the public-only class]\label{def:mechanism}
Within the general environment, an \emph{expression mechanism} is an extensive form game between an administrator and $S$. At each time point agents attempt \emph{exposure actions}, which publish their expressions under their identities when completed, or \emph{private commitments}, which the administrator verifies on receipt and keeps in private custody without attribution. The administrator observes the completed exposure actions and the private commitments it received. Agents observe the nondecreasing public state $C_h$, starting from the self-protecting vanguard $C_0$. The mechanism may change the public coalition only in response to completed exposure actions or private commitments it received. The admissibility correspondence allows any subset of attempted actions to remain uncompleted, including the case in which only one completes. A mechanism is \emph{public-only} if its publications condition only on completed exposure actions. Every mechanism is permission-respecting: a new expression becomes public under a recruit's identity only through her own completed exposure action or a joint publication she permitted. Regret-free safety applies at every admissible realization.
\end{definition}

\begin{theorem}[Representation of Regret-Free Crossings]\label{thm:necessity}
Under Assumptions~\ref{ass:M} and~\ref{ass:U}, with mechanisms as in Definition~\ref{def:mechanism}, suppose $\Cs(S)\subsetneq\Cp(S)$.
\begin{enumerate}
\item[(i)] No public-only mechanism satisfying regret-freeness reaches, in any realization, a public coalition outside $\{C:C_0\subseteq C\subseteq\Cs(S)\}$; in particular none implements $\Cp(S)$.
\item[(ii)] Any regret-free mechanism that in some realization reaches a public coalition $D\not\subseteq\Cs(S)$ contains a \emph{private conditional-release stage}: at the first completed action that moves the public set outside $\Cs(S)$, the executing device has already received a still-private commitment from every agent the action newly names other than possibly its own author, and this set of commitments is nonempty.
\end{enumerate}
\end{theorem}

The theorem turns the example's private release into a requirement on institutions. Part (i) rules out every public-only outcome containing anyone beyond $\Cs$, not just a particular order of public entry. Part (ii) identifies the change needed at the first crossing: the releasing party must already hold still-private permission for the statements it publishes under others' names. A lawyer, a trusted delegate, and a software protocol can differ in many respects while performing this same operation. As the fifth supporter in Section~\ref{sec:working_example} illustrates, custody need not extend to people who can join safely after the crossing.

This requirement concerns protection whenever someone becomes public, not immunity from retaliation. A coalition may meet its members' protection requirements even though speaking still carries an expected retaliation cost, provided the benefit of expression is at least as great. Allowing an additional chance that turnout fails to meet those requirements changes the comparison. A public call to action can then cross the exposure barrier by accepting that chance. Section~\ref{sec:positive_tolerance_frontier} measures this alternative, and Section~\ref{sec:custody_frontier} discusses the risk of leaks from private custody.

How much of this result the environment supplies deserves consideration. Allowing one person's public action to be completed without the others is the follow-through uncertainty that bounds the public cascade. A genuinely indivisible collective act that either attributes the whole group or attributes no one lies outside that clause and can cross without prior custody.

The other clauses govern the institution's information and permission. Requiring a device to act only on what it has received rules out release conditioned on commitments that never reached the releasing party, and the permission clause rules out naming anyone without her permission. Theorem~\ref{thm:necessity} is therefore a representation result within the general environment, not an impossibility claim about every conceivable institution: a doxxer or a hostile newspaper that names people without their permission crosses the barrier trivially, and the theorem's scope is institutions that publish only with consent. Within that class, private custody is the common core of every regret-free crossing; other mechanisms may add stages around it, but none can avoid it. These classifications turn on what an institution does, whatever it is called. Probabilistic attribution enters the analysis by allowing some unprotected exposure and does not create a third kind of action. A side channel that identifies a signer before release is a leak. Likewise, a platform that accepts authenticated signatures and publishes them only above a threshold is an administrator relative to the opposition, however public it looks to its users.

Theorem~\ref{thm:necessity}'s stage has three properties that can be perturbed one at a time: privacy, conditionality, and verification. Joint publication is the event that those three make safe. Supplemental Appendix Section~\ref{sec:tightness_supp} classifies familiar institutions by these features, and Remark~\ref{prop:tightness} there shows why each is necessary.

\begin{remark}[Terminal-payoff trust versus per-exposure safety]\label{rem:gale}
In \citet{galeMonotoneGamesPositive2001}, payoffs depend on the final profile of irreversible actions. An early participant can therefore rely on later participation when evaluating the outcome. Regret-free safety instead requires protection whenever someone becomes public, because retaliation can arrive before later entrants. A protective destination need not make the path to it safe.
\end{remark}

Information and private custody address different obstacles. If an employer learns that most of the staff agree, it may retaliate less, making some people willing to speak with a smaller group. That changes the protection environment. A contract instead changes how a group forms within a given environment. Supplemental Appendix Section~\ref{sec:design} separates these effects formally.

\subsection{When a Public Call to Action Can Cross the Exposure Barrier}\label{sec:positive_tolerance_frontier}

A private contract is not the only way to get supporters past the exposure barrier when they are willing to risk speaking without enough protection. A public call to action can also succeed because enough people may in fact follow through to protect everyone who appears. The question is when the risk of too little turnout becomes small enough to accept. The earlier criterion instead requires every realized coalition to protect its members. That can mean making expression worthwhile in expectation while some chance of retaliation remains.

\runinhead{Scheduled action rounds.}\label{par:scheduled_rounds} People decide whether to answer before learning who else follows through. Those who arrive late do not count for this action: protection depends on the people actually participating in the appointed round, together with the coalition already public. Later participation cannot repair a shortfall in that round. Formally, fix all attempts before any completion draw, with no further participant decisions or device actions until after the round-end assessment. If the preceding coalition is $C$ and exactly the attempts in $R$ complete this round, a completing participant $i$ is unprotected exactly when $C\cup R\notin\Pfam_i$. Safety is assessed against everyone who completes in the round, not against the order in which they complete.

This convention retains E1's individual actions, causal information, and permission rules (Assumption~\ref{ass:causal}). Only the risk assessment of a scheduled public round changes. Outside these rounds, protection is still assessed at each exposure event, including device releases. A scheduled time is not all-or-nothing completion: every subset of the attempted group remains possible, and nobody collects the group's endorsements beforehand. At zero tolerance the necessity theorem still applies to scheduled rounds, except on outcomes of probability zero (Supplemental Appendix Proposition~\ref{prop:interim_necessity}(iii)). The positive-risk bounds below concern scheduled rounds.

For this calculation, each attempted public action completes independently with probability $1-\varepsilon\in(0,1)$, after the person has chosen to act. If an attempt does not complete within the round, it is costless and reveals no new name. People who remain willing to act can try again in a later round, with the same independent chance of completion. Monotone protection continues to hold.

A call to action is \emph{acceptable} to person $i$ if her probability of unprotected exposure, conditional on her own public action being completed, is at most her tolerance $r_i$. One interpretation gives her a benefit $e_i$ from expression and a loss $h_i$ from being exposed without protection. With probability $p$ of that loss, expression pays $e_i-p h_i$, so she accepts when $p\le r_i=e_i/h_i$. A larger benefit or smaller harm supports more risk.

Here the harm occurs exactly when the group does not protect her. Section~\ref{sec:target} instead allows retaliation to succeed with some probability even in a protective coalition, defining protection by expected payoffs. The present calculation asks whether a person accepts a chance that turnout will fail to meet her protection requirement.

\runinhead{Turnout in response to a call to action.} A call to action asks people to act at once: a strike date, a scheduled protest, or an announcement that everyone should post at noon. Because the call to action is broadcast rather than whispered, it names no supporters. Each person decides privately whether to answer, and nobody holds a list that could leak. The uncertainty instead concerns who will show up. Rather than tunneling beneath the barrier, the group tries to clear it in one coordinated leap, relying on scale for protection. When enough people follow through, the realized coalition is self-protecting even though no individual's participation is guaranteed. The person who issues the call to action is taken to belong to the already-public, self-protecting vanguard. If issuing the call to action would first expose a new organizer, her protection would also have to be checked.

This absence of a list distinguishes the call to action from one organized by soliciting particular people and keeping a list of who agreed. Such a list is a set of supporters in someone's custody, so it carries the custody risk of a contract together with the follow-through risk of a call to action. That arrangement combines parts of the two institutions and is not the list-free alternative considered here. Supplemental Appendix Section~\ref{sec:dominance_supp} bounds the set of recruits that any acceptable mechanism exposes with positive probability when each requires some chance of protection. With leak-free custody and free private retries, a contract attains the largest protective coalition with probability one, without unprotected exposure, if supporters keep signing until their endorsements arrive.

To calculate the risk of answering, begin with a finite group. Put $y^-=m(\Cs)$, and for a target $y_t>y^-$ let $B(y_t)=\{i\in S\setminus\Cs:\rho_i\le y_t\}$ be the \emph{same-date group}: the set, defined for the analysis, of supporters outside the cascade coalition whom the announced target would protect in this scheduled round. ``Same-date'' means the same appointed action round, not separate actions pooled across a calendar day. The organizer broadcasts the call to action and need not know who belongs to $B(y_t)$. The calculation assumes that people know their own requirements and that exactly those in $B(y_t)$ answer.

The group must do more than protect its members at the target: it must also make further safe growth possible. Call $y_t$ a \emph{bridging target} if, once all of $B(y_t)$ are public, ordinary safe expansion reaches $\Cp$: $T^{|S|}(\Cs\cup B(y_t))=\Cp$. Write $m_B=m(B(y_t))$, $s(y_t)=y^-+m_B-y_t$ for the mass available above the target when the whole group appears, and $\Delta(y_t)=(1-\varepsilon)m_B-(y_t-y^-)=(1-\varepsilon)s(y_t)-\varepsilon(y_t-y^-)$ for expected safety slack after allowing for people who do not follow through.

This finite calculation has a graphical counterpart. In the three-crossing nonatomic graph, every $y_t\in(y^u,y^+]$ is the limiting analogue of a bridging target, $m_B=G_S(y_t)-y^-$, and $s(y_t)=G_S(y_t)-y_t$. The proposition itself concerns a finite group and its finite destination. The continuous graph supplies the corresponding geometry and adds no agents to the economy.

The calculation concerns the specified answering group. Outside this calculation, people with higher requirements may also answer. Under monotone protection their participation cannot make members of $B(y_t)$ less safe, but the proposition does not protect an additional respondent unless realized public mass reaches her own requirement.

\begin{proposition}[The call to action]\label{prop:call}
In a finite scalar protection case with scheduled-round assessment and independent completion as above, fix a bridging target $y_t$ whose same-date group $B=B(y_t)$ has $n$ equal-mass members:
\begin{enumerate}
\item[(i)] (\emph{Mean condition.}) The expected realized public mass is $y^-+(1-\varepsilon)m_B=y_t+\Delta(y_t)$: mean realized mass clears the target exactly when $(1-\varepsilon)\,s(y_t)\ge\varepsilon\,(y_t-y^-)$. This is a condition on the mean; part (ii) bounds the probability of an unprotected outcome.
\item[(ii)] (\emph{Concentration.}) With $n$ same-date group members of equal mass $m_B/n$ and $\Delta=\Delta(y_t)>0$, every member's unprotected probability conditional on her own public action being completed is at most $\exp(-2n\Delta^2/m_B^2)$, as is the probability that the realized mass falls short of $y_t$; the call to action is acceptable to every member once $n\ge(m_B^2/2\Delta^2)\ln(1/\underline{r})$, where $\underline{r}=\min_{i\in B}r_i>0$ is the same-date group's minimum tolerance. Along equal-mass refinements with fixed $m_B$ and $\Delta>0$, the unprotected probability vanishes.
\item[(iii)] (\emph{Composition.}) Under the positive-slack condition in part (ii), if the realized public mass reaches $y_t$, every member who did not become public on her first attempt can try again safely. Once all of $B$ are public, safe expansion reaches $\Cp$. Thus a public-only composite of the call to action, safe retries, and safe expansion reaches the finite economy's tunnel destination except with probability at most $\exp(-2n\Delta^2/m_B^2)$.
\end{enumerate}
\end{proposition}

Calls to action lie between public organizing and Social Assurance Contracts. For the result above, expected turnout needs to exceed the participation required for protection, and that extra participation provides a buffer against people not following through. Turnout here is measured by total public mass, including the vanguard. Follow-through is independent across the people who have decided to answer: learning whether one appears tells us nothing about whether the others do. When many people each supply a small share of the group's mass, total turnout is unlikely to fall far below its expected level. The \emph{concentration bound} above puts an upper limit on the probability of that shortfall. Holding the group's total mass and expected margin above the target fixed, dividing that mass among more equal-sized participants makes this upper limit smaller. The risk can therefore fall below any fixed positive tolerance, although the possibility of appearing alone remains, so the process does not become regret-free at $r=0$. A Social Assurance Contract instead checks who has signed before deciding whom to publish, so participants need not rely on predictions of turnout.

A call to action can carry little risk in the calculation above yet seem too risky to prospective participants. It relies on forecasts of turnout, whereas a Social Assurance Contract checks endorsements it has received. A credible poll can improve those forecasts, but a count of private agreement need not establish willingness to act, and even accurate forecasts do not guarantee follow-through. Pessimism can deter participation, but a smaller responding group may still protect itself and allow further public expansion.

\subsection{Turnout Risk versus Custody Risk}\label{sec:custody_frontier}

The primary risk to someone answering a call to action is being stood up. The Social Assurance Contract instead waits for the necessary endorsements, but what it holds could leak. A harmful leak identifies a signer to the opposition without the protection she requires. A partial public release can expose too few people to protect one another, and even theft of the whole list by the opposition does not supply the protection that publication to allies and the wider audience would. Supplemental Appendix Section~\ref{sec:frontier_supp} gives formal bounds combining custody and follow-through risk.

Whether people follow through together can change this comparison. A shared threat that makes many people stay home while others still appear can invalidate the bound derived under independent follow-through. Correlation need not be harmful, however. Under scheduled-round assessment, if a protective group completes the round together or not at all, no participant faces a protection shortfall. That pattern of follow-through rules out the outcomes in which only one person appears, which Assumption~\ref{ass:U} allows, and announcing a common time does not produce it. Nor does it turn separate actions into one joint publication. Choosing among calls to action, contracts, and hybrids would also require their costs, delays, and participation responses.

\runinhead{Custody in practice.} Protecting private endorsements requires limiting who can read them before publication. Threshold cryptography can keep the list unreadable to any single party, including the administrator, while allowing the publication condition to be checked \citep{evansPragmaticIntroductionSecure2018,breuerSeCritMassThresholdSecret2024,mangipudiCollusionDeterrentThresholdInformation2023}. Supplemental Appendix Section~\ref{sec:custody_practice} discusses these safeguards and the information that recruitment can reveal.

Secure storage is not enough if recruitment or signing identifies supporters to the opposition. Protecting that part of the process is another requirement on implementation. Both the contract and the call to action can announce their existence publicly. The call to action's remaining advantage is that it holds no list.

\subsection{Participation and Institutional Conditions}\label{sec:scope_axioms}

The results so far compare what institutions can publish and the risks their participants bear. Applying those comparisons also requires knowing who will sign, how the administrator checks protection, and whether private collection changes the environment it faces. This subsection addresses those questions in that order.

\runinhead{Participation.} Three questions must be distinguished: First, does the contract protect everyone whose endorsement it publishes? Second, is there an equilibrium in which enough people sign for publication? Third, does the model predict that people will choose that equilibrium? Theorem~\ref{thm:tunnel} answers the first question. The other two depend on what signing costs and what participants believe.

Under the \emph{costless signing conditions}, signing costs nothing, stays hidden until release, never leaks, and is governed by a credible publication rule. A hidden signature cannot change anyone's behavior, protection, or the audience's response before publication. Consider a one-shot simultaneous game in which each recruit signs or abstains without observing others' choices. Every chosen signature reaches the administrator, which then publishes $D(V)=\Cp(V)\setminus C_0$. An unpublished recruit earns zero. A published recruit earns $u_i(\Cp(V))$, where $u_i(C)\ge0$ exactly when $C$ protects her. \emph{Strict safe surplus} means $u_i(C)>0$ in every coalition that protects her.

\begin{proposition}[Signing under costless signing conditions]\label{prop:fullsign}
Under Assumptions~\ref{ass:M} and~\ref{ass:atomic}, in the signing game just defined:
\begin{enumerate}
\item[(i)] Signing is weakly better than abstaining against every profile of others' choices, with payoffs evaluated at the resulting coalition. Everyone signing is a Nash equilibrium.
\item[(ii)] Under strict safe surplus, signing weakly dominates abstaining for every $i\in\Cp(S)\setminus C_0$. An equilibrium avoids weakly dominated strategies if and only if all these recruits sign, in which case the full public outcome is $\Cp(S)$. Recruits outside $\Cp(S)$ can never be published and are indifferent. All-abstain is also an equilibrium when no recruit can safely join the vanguard alone. The selection is by dominance, not uniqueness.
\end{enumerate}
\end{proposition}

The full-publication conclusion assumes that signatures arrive. If a costless, private signing attempt may fail to reach the administrator, attempting still pays at least as much as abstaining, but everyone attempting need not deliver every endorsement. The release rule protects everyone it publishes from among the endorsements received. Collecting every endorsement additionally requires that people keep trying until their submissions arrive, as in the collection process of Supplemental Appendix Section~\ref{sec:dominance_supp}.

Even a small unreimbursed signing cost can change the selection. Take two potential recruits who cannot safely join the vanguard individually but are protected if both join, with benefit $v>0$ from safe publication and signing cost $c\in(0,v)$. Both signing remains an equilibrium, but both abstaining is now a strict equilibrium: signing alone incurs $c$ and releases nothing. If person $i$ assigns probability $p_i$ to the other signing, her expected gain is $p_i v-c$. A cost small relative to exposure harm need not be small relative to that expected gain. Cryptographic custody can reduce signing risk without making this participation comparison disappear. The companion paper \citep{cashmanConditionalDisclosureCoordination} studies entry with signing frictions through a global game in the tradition of \citet{carlssonGlobalGamesEquilibrium1993}, \citet{morrisGlobalGamesTheory2003}, and \citet{angeletosDynamicGlobalGames2007}. Its comparison, in which all supporters share one protection requirement, is separate from this paper's guarantee for whatever valid signer set $V$ is actually received.

In the corresponding game of public participation, analyzed in the Supplemental Appendix, players choose among strategies that are safe under every completion pattern, and everyone who can safely join acts immediately, reaching $\Cs$ almost surely (Theorem~\ref{thm:selection}).

\runinhead{Implementing the publication condition.} The set $V$ contains at most one currently valid signature from each eligible person. A signature with an institutional login, ORCID account, or cryptographic certificate credibly records her permission to publish the statement under her name if the announced publication condition is met. It verifies who endorsed what, not sincere support or coalition safety, and enters $V$ only after whatever screening the architect requires.

The administrator must check whether the current signers and the vanguard together would form a protective coalition. The general rule $D(V)=\Cp(V)\setminus C_0$ describes what is feasible, not what an administrator can observe. A practical contract must therefore preannounce a test the administrator can apply: a total threshold, class-specific quotas, or both. For example, 500 signatures including 100 from tenured faculty could trigger publication. These tests implement coarse versions of the protection families. A finer version asks each signer to report her protection requirement, much as the subjects of \citet{braghieriThresholdDisclosureCollective2026} chose their own disclosure thresholds. In the scalar protection case, truthful reporting is optimal. A signer who knows only an upper bound can report that bound, which may shrink the released coalition but cannot place her in a coalition that does not protect her (Supplemental Appendix Proposition~\ref{prop:report}). Whether that option is useful depends on how well potential signers know their own requirements, which remains unclear.

Even with a usable protection test the administrator must hide unreleased signatures and obey the publication rule, which creates a credibility problem \citep{akbarpourCredibleAuctionsTrilemma2020}. This is related to the entrepreneurial moral hazard that \citet{strauszTheoryCrowdfundingMechanism2017} studies in all-or-nothing crowdfunding, where the party that collects the pledges can divert them and the platform is assumed to enforce the rules. For endorsements, cryptographic custody can constrain the administrator's discretion and shift the residual risk to credentials, key custody, code, and the devices people sign from. Section~\ref{sec:custody_frontier} treats the custody risk that remains.

The formal comparison also presupposes the domain stated in Section~\ref{sec:representation_necessity}, in which the opposition responds to a coalition going public, while private permission stays out of its reach. Social and professional retaliation often has this form, since a signed public statement gives an employer, organization, or public audience evidence that a rumor or third-party allegation does not. Within this domain, the formal comparison holds fixed the opposition, its willingness to punish, and the protection each coalition provides. This isolates what the timing of attribution changes.

Supplemental Appendix Section~\ref{sec:custody_practice} discusses threats to these conditions, including infiltration and false contracts.

\section{Who Becomes Public, and How the Opposition Responds}\label{sec:target}

Who benefits from reaching a coalition that public organizing cannot assemble? The earlier analysis leaves this question in terms of supporters' protection requirements. This section gives those requirements an economic interpretation through the opposition's decision to fight or fold, and then compares who becomes public and what happens to their expression. It begins with the assumption that larger coalitions are harder to suppress because their greater resilience can both discourage retaliation and make it less likely to succeed. Here I model the opposition as a single decision-maker choosing whether to retaliate. It could be an employer deciding whether to fire workers, an accused person deciding whether to intimidate witnesses, or a political faction deciding whether to repress critics.

Let $q(C)\in[0,1]$ be the probability that coalition $C$ survives retaliation. Adding members weakly raises this probability. The opposition loses $d>0$ if the coalition survives and pays a retaliation cost $\kappa_T\in(0,d)$ if it fights. Retaliation succeeds with probability $1-q(C)$, so fighting produces an expected gain of $d[1-q(C)]$. After observing the public coalition, the opposition fights when this gain exceeds the cost $\kappa_T$ and folds otherwise. Equivalently, the opposition has a break-even survival probability $q^\ast=1-\kappa_T/d$ and fights coalitions with resilience below that threshold. For the comparison, this threshold is drawn once from $F$ and known to the opposition; participants evaluate release before learning it. The same draw and distribution are used when comparing coalitions. This is a comparison of responses to given public coalitions, not a repeated game in which participants update $F$ from attacks during recruitment.
\begin{equation}\label{eq:fold}
C\ \text{draws no retaliation iff}\ q(C)\ge q^\ast;\qquad
\Pr(C\ \text{is fought})=1-F\bigl(q(C)\bigr).
\end{equation}

The probability $1-F(q)$ reflects participants' uncertainty about the opposition's cost of retaliation. The opposition itself does not randomize: given its cost and the coalition, it chooses whether fighting is worthwhile. If it does fight, the separate probability $1-q$ describes whether it succeeds. The overall probability of suppression is therefore $\pi(q)=(1-F(q))(1-q)$. As a coalition grows, the increase in $q$ matters twice: it makes the opposition less likely to fight and makes any fight less likely to succeed.

To translate this suppression probability into protection, suppose person $i$ gains $e_i$ from expression and loses $h_i$ if the coalition is suppressed. She accepts publication when $e_i\ge\pi(q(C))h_i$. The starting condition requires $e_i\ge\pi(q(C_0))h_i$ for every vanguard member $i\in C_0$, so the initial coalition already protects them under this criterion. This criterion averages over the opposition's type and the outcome of a fight. A person can therefore prefer joining in advance yet be worse off than silence if suppression occurs and its harm exceeds her benefit from expression. If the published coalition is also uncertain, expected utility is $e_i-h_i\,\mathbb E[\pi(q(C))]\ge0$, rather than the binary risk bound of Section~\ref{sec:positive_tolerance_frontier}. Regret-free disclosure requires every coalition in which a person becomes public to meet this expected-payoff criterion. The guarantee concerns who appears with her, not how a subsequent fight turns out. Lemma~\ref{lem:interim} in Appendix~\ref{app:proofs} applies the coalition results to these expected-payoff protection families.

With these protection families established, the selection result in the Supplemental Appendix (Theorem~\ref{thm:selection}) requires, in addition, strictly positive expected surplus in every protecting coalition. The \emph{fight-or-fold model} derives monotone protection from monotone resilience while holding the opposition's stakes and costs fixed; it does not derive safety in numbers for unrestricted strategic retaliation.

\runinhead{Who the tunnel brings into view.} The additional public participants are not an arbitrary sample of silent supporters. In the scalar protection case, the public cascade has already recruited everyone whom its size protects. The tunnel can therefore bring in the people who need more protection (who need a lower probability of successful suppression before speaking is worthwhile), which the next result connects to the benefit of expression and the harm from retaliation.

\begin{corollary}[Composition of released coalitions]\label{cor:composition}
In the nonatomic scalar protection case, every hidden supporter has a strictly higher protection requirement than every member of the cascade coalition, so $\rho_i>y^-\ge\rho_j$ for all $i\in\Cp\setminus\Cs$ and $j\in\Cs$. When the cascade coalition has positive mass, the full coalition formed through a Social Assurance Contract therefore has a weakly higher average requirement than the cascade coalition grown from the same vanguard, and strictly higher whenever the hidden supporters have positive mass. In the fight-or-fold model, every hidden supporter has a strictly lower ratio of expressive benefit to exposure harm than every member of the cascade coalition, so the tunnel's additional members are those with the most to lose or the least to gain from speaking.
\end{corollary}

The composition comparison holds the population, vanguard, protection environment, and full signing fixed. Under those conditions, conditional release includes people who require more protection than the supporters recruited through the public cascade, and than the already-public vanguard members. In an application these may be junior or precarious participants, but selective signing can change the composition in that the people most exposed by speaking may also be most reluctant to trust administrators' custody of their endorsements. An empirical comparison should therefore record who was eligible and who signed, as well as who appeared in a successful release.

Their entry also changes what records of retaliation tell us about suppression.

\runinhead{The visibility--suppression reversal.} Replacing public organizing with conditional release can simultaneously produce \textit{more} observed retaliation and \textit{less} suppression by bringing previously silent people into view. Hold the environment and the full-signing population fixed. Let $G=\Cp\setminus\Cs$ be the hidden supporters whom the tunnel reaches but the public cascade does not. Under public organizing, these people remain silent. Their expression is completely suppressed, and the data show no retaliation against them because they never become targets. Conditional release instead publishes them as part of $\Cp$. The opposition then fights with probability $1-F(q(\Cp))$ and successfully suppresses the coalition with probability $\pi(q(\Cp))$. The proposition assumes that this last probability lies strictly between zero and one.

\begin{proposition}[Visibility--suppression reversal]\label{thm:rd}
Fix the fight-or-fold primitives above, full signing, a nonempty set $G=\Cp\setminus\Cs$ of hidden supporters, and $\pi(q(\Cp))\in(0,1)$. For the hidden supporters, the model comparison from public organizing to conditional release \emph{raises} observed retaliation---from $0$ to $1-F(q(\Cp))$ in fight attempts, or to $\pi(q(\Cp))$ in successful punishments---while \emph{lowering} effective suppression of their expression from $1$ to $\pi(q(\Cp))$. The two move in opposite directions.
\end{proposition}

More recorded retaliation need not mean that more expression has been suppressed. Suppose fear keeps a hundred employees from signing an open letter, so management need not punish anyone. If a Social Assurance Contract lets them publish together and withstand management's retaliation, attacks become visible while their views are no longer silenced. Records of retaliation alone miss the silence produced by fear. Other interventions that induce previously silent people to speak can produce the same pattern. Proposition~\ref{prop:permember} in the Supplemental Appendix states the result person by person.

This comparison changes the institution while holding the population fixed. Neither Proposition~\ref{thm:rd} nor its per-member version in the Supplemental Appendix varies the publication threshold within contracts, where a higher threshold can also reduce the chance of release. Supplemental Appendix Section~\ref{sec:experiment_supp} sets out the corresponding measurement questions, including how to distinguish the effects of better turnout information from those of conditional publication.

\runinhead{Experimental evidence.} The experiment in \citet{braghieriThresholdDisclosureCollective2026} supports two behavioral premises of this comparison: attribution suppresses controversial expression, and people condition their willingness to have their vote revealed on the share of voters who agree. In a sample of 298 UC Berkeley undergraduates, students voted on a proposal to allow transgender women to compete in women's collegiate sports, with the aggregate result to be shared with the Chancellor's office. Random assignment placed students under one of three rules. Their votes remained private, were publicly linked to names, or were linked to names only if the share voting the same way met a threshold chosen by that voter. Participants' names appeared on lists in every treatment, but their votes did not necessarily appear alongside them. Voting against inclusion was the socially controversial position in this population. Comparing public with private voting identifies the effect of linking names to votes: linking them raised abstention and cut expression of the socially controversial position by roughly half. Under threshold disclosure, participation and controversial expression were close to their private-voting levels. The chosen thresholds provide the second piece of evidence, documenting willingness to disclose one's vote conditional on how many others agree, whether or not those others' votes become public. Unlike the treatment comparison, the comparison of chosen thresholds by vote direction is descriptive rather than randomized. Section~\ref{sec:experiment_supp} gives further details.

\section{Discussion}\label{sec:conclusion}

Social Assurance Contracts can change the public discussion by enabling coalitions that cannot form safely through open organizing. The main obstacle to public organizing is the exposure of early participants to retaliation before enough others have joined. The contract changes that by collecting endorsements privately and publishing them together only when the resulting coalition protects its members. Within the environment studied here, crossing the exposure barrier while guaranteeing protection requires this combination of private commitments and joint publication. The distinction between cascade and tunnel is therefore about how a coalition forms, not only about how much support exists.

For institutional design, the first question is whether supporters need to become publicly identified at all. A credible poll may suffice when an employer or university needs evidence of support before changing a policy. An open letter serves a different purpose when the identities of its signers give the statement credibility, demonstrate their willingness to stand behind it, or, most importantly, allow them to organize further activity.

When named expression is needed, the choice between a call to action and a contract turns partly on the reliability of participation. When protection is assessed against participation in the appointed round, a call to action can protect well if the expected answering group comfortably exceeds the required size, people follow through independently, and no single person supplies a large share of the group's protection. Some supporters can then fail to appear without leaving the others unprotected in that round. A contract instead requires collecting endorsements before publishing anything, and so avoids relying on promised turnout but also creates a private record that must be protected. The comparison is therefore between the risk that others will not follow through and the risk that private endorsements will leak, alongside the costs and delays of each arrangement.

Private custody need not mean giving an administrator a readable list and trusting that person to keep it secret. Threshold cryptography can keep endorsements unreadable to any single custodian while allowing the publication condition to be checked. Such arrangements can greatly reduce dependence on custodial discretion. Protection must also cover recruitment and signing, since secure storage accomplishes little if the opposition can identify supporters as they enter the contract. The practical objective is to prevent premature identification throughout collection, while making the promised joint publication credible to potential signers.

These design lessons concern what participants should do given their goals, not whether their goals deserve support. A coalition can enable its members to spread false allegations or impose other costs on outsiders. Conversely, making previously suppressed expression public can benefit people who never signed. A larger protective coalition is therefore not, by itself, a measure of social improvement. Evaluating the institution requires considering what the coalition says and does, whom it affects, and what organizing it costs.

Several questions for future work follow directly from these choices. First, how cheaply can an institution make private custody and its publication rule credible? If the technology is close at hand, adopting Social Assurance Contracts is a matter of activating latent demand rather than of engineering. Second, which supporters will sign when participation takes effort or carries some risk, and how should an institution handle strategic or false submissions? Third, how does the strategic environment evolve over repeated campaigns, including competition among similar contracts? Finally, the welfare consequences of such contracts remain largely open. Answering these questions would help determine which arrangements to use in particular settings.

\ifanonymizeforpeerreview\else
\ifshowacknowledgments
\section*{Acknowledgments}
I wish to thank Adam Bear, Rahul Bhui, Caspar Kaiser, Drazen Prelec, Birger Wernerfelt, Zachary Wojtowicz, and Erez Yoeli for many helpful comments and suggestions. I am especially indebted to Alex Tabarrok for his substantial feedback on an earlier version of this paper.
\fi
\fi

\appendix
\makeatletter\renewcommand\thesubsection{\@arabic\c@subsection}\makeatother

\appendixparttitle{Appendix}

\section{The General Environment and the Representation Lemmas}\label{app:representation}

This appendix states the environment and the two lemmas that Section~\ref{sec:representation_necessity} uses in words. Their proofs are in Appendix~\ref{app:proofs}.

As in the article body, an action \emph{completes} when it takes effect. A public action completes when the statement becomes public under its author's name, and a private commitment completes when it reaches its recipient. These definitions concern what happened, regardless of whether noncompletion reflects a change of mind or a practical obstacle.

\begin{assumption}[Causal action histories (environment E1)]\label{ass:causal}
\begin{enumerate}
\item[(i)] Each attempted action has one author and is a separate event. Once completed, it is irreversible. Every subset of a date's attempted individual actions is an admissible completion outcome: the environment must allow that exactly those actions complete and the remaining actions do not complete. Histories record completions in their realized order, including within a calendar date, and public attribution updates at each completion before any device action at that date.
\item[(ii)] Information moves through completed actions. Each action has a designated recipient that can verify its completion. With the opposition held fixed, each action has a deterministic named set and attribution status, realized at completion. Receipt by a device is not itself attribution.
\item[(iii)] A device---a person or automated system that executes an action---chooses its action and named set using only the public history and completed actions it received strictly before execution. If an unattributed action is absent from that received record, the admissibility correspondence contains a twin history with the same public and received record in which that action never completed; the device takes the same action at both histories.
\item[(iv)] Permission binds action by action: an action cannot name another agent unless that agent has completed an action permitting publication on the realized history.
\end{enumerate}
\end{assumption}

Clause (iii) carries the main restriction: a device can condition release only on what it has received, not on a commitment that completed elsewhere but never reached it. A device may be a composite institution, such as a protocol run by several parties or a threshold-cryptographic system. Its received record is then the union of what its components hold, and clause (iii) requires publication to be causally downstream of commitments held somewhere within it. The two lemmas use the two action types that Definition~\ref{def:mechanism} makes precise.

\begin{lemma}[The attribution dichotomy]\label{lem:dichotomy}
Consider any technology in E1 that publishes a person's identity only with her permission and assigns a definite attribution consequence to each completed action. Every action by a participant in $S$ can be described using the two action types in Definition~\ref{def:mechanism}; the institution's own publications are device actions, governed by clauses (iii) and (iv) of Assumption~\ref{ass:causal} rather than by this classification. For every possible realization, the new description has the same actions, payoffs, path of public identities, and released sets. Specifically: (i) a participant's completion that makes its author identifiable is an \emph{exposure action}; (ii) a participant's completed action that remains unattributed but can be verified by a recipient is a \emph{private commitment}; and (iii) the possible completion patterns include every subset of a date's attempts, as Assumption~\ref{ass:U} and Definition~\ref{def:mechanism} require. Thus that definition describes, rather than restricts, technologies that publish identities only with permission at fixed protection families.
\end{lemma}

\begin{lemma}[Joint publication requires custody of commitments]\label{lem:custody}
Maintain the preceding primitives. Fix a history $h$, its attempted profile $A_h$, and the family $\mathcal U(A_h)$ of realizations the institution must withstand there. Let $B$ be a set of at least two recruits whose endorsements are not yet public at $h$. A realization $R\in\mathcal U(A_h)$ is a \emph{joint-completion event} over $B$ if it first attributes all of $B$ at one completion and no realization in $\mathcal U(A_h)$ first attributes some but not all members of $B$.
\begin{enumerate}
\item[(a)] Any permission-respecting technology realizing such an event operates a private conditional-release stage: one device action jointly attributes $B$, and before executing it the device has received a completed, still-unattributed private commitment from every member of $B$ other than its author. No device action first names such a member at a history lacking that commitment.
\item[(b)] Conversely, a failure-private stage in which a device has custody of completed, unattributed private commitments and releases them by one action, and at which no member of the released coalition attempts an exposure action, realizes a joint-completion event at every release history whose newly attributed coalition has at least two members.
\end{enumerate}
Without custody of completed commitments, every first attribution is an individually completing action governed by Assumption~\ref{ass:U}.
\end{lemma}

\begin{remark}[Correlation and privacy risk]\label{rem:imperfect}
The lemma concerns guaranteed joint attribution, because admissibility is a set of realizations rather than a probability law. Correlating separate attempts while preserving singleton outcomes cannot remove the singleton safety requirement. Proposition~\ref{prop:call} quantifies risk under independent completion and scheduled-round assessment. A different, perfectly correlated all-or-none completion law can eliminate round-end shortfall for a protective group, but violates Assumption~\ref{ass:U} and does not turn separate completions into one attribution event. Arbitrary correlation need not reduce risk. Supplemental Appendix Proposition~\ref{prop:robust_frontier} treats leaks during a true custody stage.
\end{remark}

\section{Proofs}\label{app:proofs}

This appendix proves every formal result stated in the article body and in Appendix~\ref{app:representation}. The proofs are ordered so that each supporting result appears before the argument that uses it. The Supplemental Appendix contains extensions and auxiliary results.

\begin{proof}[Proof of Lemma~\ref{lem:mono}]
Every member of $C_0$ belongs to both $\Gamma(C)$ and $\Gamma(C')$. For a recruit, suppose $C\cup\{i\}\in\Pfam_i$ and $C\subseteq C'$. Then $C\cup\{i\}\subseteq C'\cup\{i\}$, so Assumption~\ref{ass:M} gives $C'\cup\{i\}\in\Pfam_i$. It follows that $i\in\Gamma(C')$ and, if $i\notin C'$, that $i\in T(C')$. The definition of $T$ immediately gives $C\subseteq T(C)$.
\end{proof}

\begin{proof}[Proof of Theorem~\ref{thm:largest}]
Take any nonempty collection of self-protecting coalitions $\{C_j\}$ and let $U=\bigcup_j C_j$. Then $C_0\subseteq U$. If $i\in U$, then $i$ belongs to some $C_j$. Because $C_j$ is self-protecting, $C_j\in\Pfam_i$. Since $C_j\subseteq U$ and $i\in U$, Assumption~\ref{ass:M} gives $U\in\Pfam_i$. This holds for every $i\in U$, so $U$ is self-protecting. There is at least one self-protecting coalition, namely $C_0$. The union of all self-protecting coalitions is therefore itself self-protecting and contains every such coalition. Thus $\Cp(S)$ is well defined and is the largest self-protecting coalition. Any safe rule must leave a self-protecting public coalition, so that coalition must be a subset of $\Cp(S)$.
\end{proof}

\begin{proof}[Proof of Lemma~\ref{lem:reduction}]
($\Leftarrow$) Take any realization in which $i\in R_t$. By hypothesis, $C_{t-1}\cup\{i\}\in\Pfam_i$. Because $C_{t-1}\cup\{i\}\subseteq C_{t-1}\cup R_t=C_t$, Assumption~\ref{ass:M} gives $C_t\in\Pfam_i$. Members of $C_{t-1}$ also remain safe because the public set only grows. Starting from the self-protecting $C_0$ and applying this argument at each date proves regret-free safety.

($\Rightarrow$) Fix any $i\in A_t$. Current attempts are unverified before exposure, so Assumption~\ref{ass:U} allows the realization $R_t=\{i\}$. Regret-free safety in that realization requires $C_{t-1}\cup\{i\}\in\Pfam_i$. Finally, requiring this condition for every $i\in A_t$ is equivalent, by definition, to $C_t\subseteq T(C_{t-1})$.
\end{proof}

\begin{proof}[Proof of Theorem~\ref{thm:cascade}]
(i) Lemma~\ref{lem:reduction} gives $C_t\subseteq T(C_{t-1})$. Lemma~\ref{lem:mono} and induction then give $C_t\subseteq T^t(C_0)\subseteq\Cs$.

(ii) Greedy intentions satisfy the condition in Lemma~\ref{lem:reduction} by construction. When every intended public action is completed, $C_t=T^t(C_0)$. Because the finite set can grow at most $|S|$ times, this sequence reaches $\Cs$ within $|S|$ steps.

(iii) Vanguard members are protected in $C_0$, and each other member of $\Cs$ was safe to add when she entered the iteration. Assumption~\ref{ass:M} keeps her safe as the set grows, so $\Cs$ is self-protecting. Because the iteration has also stopped, $\Cs$ is closed. Hence $\Cs$ is a fixed point.

To show that it is the least fixed point, take any $C^\ast=\Gamma(C^\ast)$ in the domain $C_0\subseteq C^\ast\subseteq C_0\cup S$. The induction starts from $C_0\subseteq C^\ast$. Suppose $T^n(C_0)\subseteq C^\ast$. Vanguard members already belong to $C^\ast$. For any other $i\in T^{n+1}(C_0)$, we have $T^n(C_0)\cup\{i\}\in\Pfam_i$. Assumption~\ref{ass:M} then gives $C^\ast\cup\{i\}\in\Pfam_i$, so $i\in\Gamma(C^\ast)=C^\ast$. Induction gives $\Cs\subseteq C^\ast$.

To show that $\Cp$ is the greatest fixed point, first note that $\Cp$ is self-protecting by Theorem~\ref{thm:largest}, so $\Cp\subseteq\Gamma(\Cp)$. Conversely, if $i\in\Gamma(\Cp)$, then $\Cp\cup\{i\}\in\Pfam_i$. The coalition $\Cp\cup\{i\}$ is self-protecting because Assumption~\ref{ass:M} keeps the existing members safe. Maximality of $\Cp$ therefore requires $\Cp\cup\{i\}\subseteq\Cp$, which implies $i\in\Cp$. Thus $\Cp=\Gamma(\Cp)$. Every fixed point is self-protecting and therefore lies in $\Cp$. The domain $\{C:C_0\subseteq C\subseteq C_0\cup S\}$ is a complete lattice with bottom $C_0$, and $\Gamma$ maps it into itself. Tarski's theorem gives a complete lattice of fixed points.
\end{proof}

\begin{proof}[Proof of Theorem~\ref{thm:tunnel}]
Failure privacy creates no new exposure before release, so the public coalition stays at $C_0$ until the endorsements in $D(V)=\Cp(V)\setminus C_0$ are published together. It then becomes $C_0\cup D(V)=\Cp(V)$. By Theorem~\ref{thm:largest}, applied to supporter set $V$, $\Cp(V)$ is self-protecting. Every person exposed at release is therefore protected in every realization, which proves regret-free safety. Assumption~\ref{ass:M} makes $V\mapsto\Cp(V)$ monotone, so $\Cp(V)\subseteq\Cp(S)$, with equality when $V=S$. Theorem~\ref{thm:largest} supplies the upper bound for every mechanism, and Theorem~\ref{thm:cascade}(iii) supplies the strict comparison.
\end{proof}

\begin{proof}[Proof of Lemma~\ref{lem:dichotomy}]
\emph{Claim 1 (exhaustive action partition).} Fix an action $a$ by a participant in $S$. Under deterministic attribution in Assumption~\ref{ass:causal}(ii), completion either attributes the author to the opposition or it does not. In the first case, completion is an exposure action. In the second, the author remains private, and a recipient can verify the completed action without making the author public. The action is therefore a private commitment. If a later action attributes the author, that later action is classified when it completes. The two cases are exhaustive for participants' actions and do not overlap; device actions are not classified, since a release can name others without identifying its author.

\emph{Claim 2 (pathwise translation).} Keep the same actions, authors, recipients, completion outcomes, and within-date order at every history. On each public and received-action record, let the administrator copy the device's action rule. Claim 1 changes only the label attached to each action. The translated mechanism therefore produces the same public identities, released sets, and private action costs in every realization. No later result or payoff in E1 depends on the label itself.

\emph{Claim 3 (admissibility).} Assumption~\ref{ass:causal}(i) allows every subset of a date's attempted individual actions to complete. This includes every realization required by Assumption~\ref{ass:U}, in which any subset of a date's actions is completed, and every complementary noncompletion realization required by Definition~\ref{def:mechanism}. Claim 2 preserves which actions complete, so the translated mechanism has the same realizations.

Claims 1--3 establish parts (i)--(iii) and equivalence in every realization.
\end{proof}

\begin{proof}[Proof of Lemma~\ref{lem:custody}]
\emph{Claim 1 (one carrier action).} Let realization $R$ contain completion $c$, which jointly attributes $B$. Because actions complete individually, $c$ belongs to one action $a^\ast$ by one author $w$. Attribution occurs when that action completes. Thus $a^\ast$, rather than separate self-attributing actions, names every $j\in B\setminus\{w\}$. Separate actions cannot produce the stipulated joint event: Assumption~\ref{ass:causal}(i) allows one of them to complete without the others, which would attribute only part of $B$.

\emph{Claim 2 (permission by completed commitment).} Action-by-action permission requires a completed permitting action from every $j\in B\setminus\{w\}$ before $a^\ast$ names her; Claim 4 shows that it is a private commitment. If $w\notin B$, this requirement covers all of $B$. If $w\in B$, it covers $B\setminus\{w\}$. Because $|B|\ge2$, at least one person must have completed such an action.

\emph{Claim 3 (receipt, not mere completion).} Fix one such $j$. Her earlier commitment has not attributed her because $c$ is her first attribution. Suppose the device had not received that commitment before executing $a^\ast$. Assumption~\ref{ass:causal}(iii) then supplies a twin history with the same public and received-action record in which the commitment never completed. The device must take the same action at the twin history, so it names $j$ without a completed commitment from her. That contradicts action-by-action permission. The device must therefore have received $j$'s completed commitment before executing $a^\ast$. The same argument rules out any device action that names $j$ when its received record lacks her commitment.

\emph{Claim 4 (private custody).} Because $c$ is $j$'s first attribution, her earlier commitment remained hidden from the opposition. Receipt by the device is not public attribution. Lemma~\ref{lem:dichotomy} therefore classifies the earlier action as a completed private commitment. Claims 1--4 show that one device action names the group and that, for each name, the device had already received a completed but still private commitment. This proves part (a).

\emph{Claim 5 (converse).} At a release history on a failure-private stage, no member of the new coalition $D$ became public through the commitment in the device's custody. The device attempts one release action, and by hypothesis no member of $D$ attempts an exposure action at that history, so the release is the only action able to name members of $D$ there. If the release action is not completed, no one in $D$ is attributed. If it is completed, it attributes all of $D$ at once; Assumption~\ref{ass:atomic} rules out partial release. Every release with $|D|\ge2$ is therefore a joint-completion event. This proves part (b).

Without custody of commitments, the device's received record contains no completed private commitment whose author remains unattributed. Claims 2--4 then rule out an action that first attributes several names. Only individual exposure actions remain, along with the realizations in which only one public action is completed. This proves the final statement.
\end{proof}

\begin{proof}[Proof of Theorem~\ref{thm:necessity}]
\emph{Claim 1 (public-only local bound).} In a public-only mechanism, a person's exposure action can depend only on earlier public exposure and on whether her own action completes. Assumption~\ref{ass:U} allows her action to complete while every other current action remains uncompleted. Regret-free safety therefore requires $C_{h^-}\cup\{i\}\in\Pfam_i$ for each newly exposed $i$. Every public transition must consequently satisfy $C_h\subseteq T(C_{h^-})$.

\emph{Claim 2 (public-only induction).} The full public set starts at the self-protecting vanguard $C_0$. Claim 1 and monotonicity of $T$ then give $C_h\subseteq\Cs$ at every history, by the same induction used in Theorem~\ref{thm:cascade}(i). This proves part (i).

\emph{Claim 3 (first crossing action).} For part (ii), order the history by completed actions as Assumption~\ref{ass:causal}(i) requires, not merely by calendar date. If several actions complete at the same time, their recorded order determines the preceding public set without changing any action or payoff. Let $a^\ast$ be the first completed action that moves the public set outside $\Cs$. Let $C_{h^-}\subseteq\Cs$ be the public set just before it, and let $R$ be the coalition that $a^\ast$ first attributes. Choose any $i\in R\setminus\Cs$.

\emph{Claim 4 (the crossing requires another named participant).} Agent $i$ is not safe with only the preceding public set and herself. If she were, then $C_{h^-}\cup\{i\}\in\Pfam_i$. Monotonicity would imply $\Cs\cup\{i\}\in\Pfam_i$, and hence $i\in\Gamma(\Cs)=\Cs$, a contradiction. Regret-free safety when $a^\ast$ completes instead requires $C_{h^-}\cup R\in\Pfam_i$. The coalition $R$ must therefore contain at least one other named participant.

\emph{Claim 5 (represented release stage).} Let $w$ be the author of $a^\ast$. Action-by-action permission requires a completed permitting action from every $j\in R\setminus\{w\}$. These actions remained unattributed because $a^\ast$ is the action that first names their authors, so Lemma~\ref{lem:dichotomy} classifies them as private commitments. Suppose the device had not received one of them before execution. The twin history in Assumption~\ref{ass:causal}(iii) would then make the same action name $j$ even though her commitment never completed, violating permission. The device must therefore have received every required commitment before executing $a^\ast$. If $w\notin R$, the device had received commitments from all of $R$; if $w\in R$, it had received commitments from all of $R\setminus\{w\}$. Claim 4 ensures that this set is nonempty. One action thus releases the coalition, with each name conditioned on a completed private commitment. This is the stage required by part (ii).
\end{proof}

\begin{proof}[Proof of Proposition~\ref{prop:call}]
Number the same-date group members $1,\dots,n$, each with equal mass $a=m_B/n$. To record who becomes public in the scheduled round, let $X_j\in\{0,a\}$ be the mass contributed by member $j$: independently across members, $X_j=a$ with probability $1-\varepsilon$ and $X_j=0$ otherwise. The public mass at the round-end assessment is $Y=y^-+\sum_j X_j$. No participant decision or device action intervenes during these draws, and subsequent rounds do not enter this risk event.

(i) Its expectation is $\mathbb E[Y]=y^-+(1-\varepsilon)m_B$. Because $m_B=(y_t-y^-)+s(y_t)$, the expected amount above the target is $\mathbb E[Y]-y_t=(1-\varepsilon)s(y_t)-\varepsilon(y_t-y^-)=\Delta(y_t)$.

(ii) Fix a member $i$ whose protection requirement satisfies $\rho_i\le y_t$. If $n=1$, then conditional on her completion, $X_i=a$ and the realized mass is $y^-+m_B>y_t\ge\rho_i$, where the strict inequality follows from $\Delta>0$. She is therefore protected with probability one, and the stated bound holds.

Now let $n\ge2$. Conditional on $X_i=a$, agent $i$ is unprotected only if $y^-+a+\sum_{j\ne i}X_j<\rho_i$. Because $\rho_i\le y_t$, this event is contained in $\{\sum_{j\ne i}X_j<(y_t-y^-)-a\}$. The sum has mean $(1-\varepsilon)(n-1)a$. The distance from this mean to the threshold is $t^*=(1-\varepsilon)(n-1)a+a-(y_t-y^-)=\Delta+\varepsilon m_B/n\ge\Delta$. Hoeffding's inequality for the $n-1$ independent variables, each ranging from $0$ to $a$, gives
\[
\Pr(\text{unprotected}\mid X_i=a)\;\le\;\exp\!\Bigl(-\tfrac{2t^{*2}}{(n-1)a^2}\Bigr)\;\le\;\exp\!\Bigl(-\tfrac{2n\Delta^2}{m_B^2}\Bigr),
\]
where the second inequality uses $t^*\ge\Delta$ and $(n-1)a^2\le m_B^2/n$.

To bound the unconditional probability as well, consider the shortfall event $\{Y<y_t\}=\{\sum_j X_j<m_B-s(y_t)\}$. Across all $n$ variables, the distance from the mean to this threshold is exactly $\Delta$. Hoeffding's inequality therefore gives the same bound, $\exp(-2\Delta^2/(na^2))=\exp(-2n\Delta^2/m_B^2)$. With unequal masses $a_j\le\bar a$, the same argument gives $\exp(-2t^{*2}/\sum_j a_j^2)$. The conditional bound is no greater than $\underline{r}$ exactly when $n\ge(m_B^2/2\Delta^2)\ln(1/\underline{r})$, and it converges to zero along equal-mass refinements with fixed $m_B$ and $\Delta>0$.

(iii) On the event $\{Y\ge y_t\}$, every $j\in B$ whose first public action was not completed has $\rho_j\le y_t\le Y$. She can therefore retry without risking unprotected exposure. With repeated independent attempts, each member eventually becomes public almost surely, and monotone protection keeps her safe as the public set grows. Once all of $B$ are public, the definition of a bridging target gives $T^{|S|}(\Cs\cup B)=\Cp$. Greedy safe expansion with retries implements this iteration almost surely. The only positive-probability way for the composite to miss $\Cp$ is therefore the initial shortfall event, whose probability part (ii) bounds.
\end{proof}

\begin{proof}[Proof of Proposition~\ref{prop:fullsign}]
(i) Fix a recruit $i$ and the others' signatures $V_{-i}$. Abstaining leaves her unpublished and pays zero. Signing also pays zero if her endorsement remains unpublished; otherwise the release rule places her in a self-protecting coalition, paying $u_i(\Cp(V_{-i}\cup\{i\}))\ge0$. Thus signing is weakly better against every profile, and everyone signing is a Nash equilibrium. A failed private submission likewise pays zero, which proves the weak comparison for incomplete attempts.

(ii) If everyone else signs, each $i\in\Cp(S)\setminus C_0$ obtains $u_i(\Cp(S))>0$ by signing and zero by abstaining. Abstaining is therefore weakly dominated for these recruits. For $i\notin\Cp(S)$, the inclusion $\Cp(V)\subseteq\Cp(S)$ prevents publication under every profile, so both choices always pay zero. Signing is never weakly dominated.

Every equilibrium avoiding weakly dominated strategies must therefore have $\Cp(S)\setminus C_0\subseteq V\subseteq S$. Then $\Cp(S)$ is available among the self-protecting coalitions between $C_0$ and $C_0\cup V$, so maximality gives $\Cp(V)=\Cp(S)$. Conversely, every such profile is a Nash equilibrium avoiding weakly dominated strategies: the released recruits strictly prefer signing, and the others are indifferent. If no recruit can safely join the vanguard alone, a lone signature releases no new name, so all-abstain is also an equilibrium.
\end{proof}

\begin{lemma}[Expected-payoff safety inherits the lattice]\label{lem:interim}
Fix the expected-payoff criterion: $i$ accepts release in $C$ iff $e_i\ge\pi(q(C))\,h_i$, with $\pi$ nonincreasing and $q$ nondecreasing in the coalition. Define $\Pfam^{\mathrm{int}}_i=\{C:C_0\subseteq C\subseteq C_0\cup S,\ i\in C,\ e_i\ge\pi(q(C))\,h_i\}$ for $i\in C_0\cup S$, and maintain $e_i\ge\pi(q(C_0))h_i$ for all $i\in C_0$. Then each $\Pfam^{\mathrm{int}}_i$ satisfies Assumption~\ref{ass:M}, so $\Gamma$, $T$, $\Cs$, and $\Cp$ are defined as before. Theorems~\ref{thm:largest}, \ref{thm:cascade}, \ref{thm:tunnel}, and~\ref{thm:necessity} apply with $\Pfam^{\mathrm{int}}_i$ in place of $\Pfam_i$ under each theorem's original additional conditions: follow-through uncertainty for the cascade; all-or-nothing release for the tunnel; and follow-through uncertainty, E1, and Definition~\ref{def:mechanism} for the representation theorem. The expected payoff $e_i-\pi(q(C))\,h_i$ is nondecreasing in the coalition. Supplemental Appendix Theorem~\ref{thm:selection} therefore also carries over under its completion technology and separate strict-safe-surplus hypothesis, that is, if $e_i>\pi(q(C))\,h_i$ for each recruit $i\in S$ on every coalition treated as protecting her.
\end{lemma}

\begin{proof}[Proof of Lemma~\ref{lem:interim}]
Suppose $C\subseteq C'$, $i\in C$, and $e_i\ge\pi(q(C))h_i$. A larger coalition is at least as resilient, so $q(C')\ge q(C)$. Because $\pi$ is nonincreasing, $\pi(q(C'))\le\pi(q(C))$. Hence $e_i\ge\pi(q(C'))h_i$. Each induced family $\Pfam^{\mathrm{int}}_i$ is therefore closed under supersets that contain $i$, as Assumption~\ref{ass:M} requires. Replacing the original protection families with these induced families changes none of the completion, all-or-nothing-release, E1, or mechanism primitives. The starting condition makes $C_0$ self-protecting under the induced families as well. The earlier proofs thus apply whenever their other assumptions hold. Finally, the expected payoff $e_i-\pi(q(C))h_i$ is nondecreasing because $\pi\circ q$ is nonincreasing. This gives the payoff monotonicity required by Supplemental Appendix Theorem~\ref{thm:selection}; its completion technology and strict-safe-surplus assumptions must still be imposed separately.
\end{proof}

\begin{proof}[Proof of Proposition~\ref{thm:rd}]
Public organizing produces $\Cs$. Because $G\cap\Cs=\emptyset$, every $i\in G$ remains silent. Her expression never becomes public, so effective suppression equals one. At the same time, the opposition never retaliates against her because no public coalition identifies her. Conditional release instead produces $\Cp$. This coalition is self-protecting, so $e_i\ge\pi(q(\Cp))h_i$ for every member, and everyone in $G$ accepts publication. Equation~\ref{eq:fold} gives a fight probability of $1-F(q(\Cp))$. Conditional on a fight, suppression succeeds with probability $1-q(\Cp)$. The unconditional suppression probability is therefore $(1-F(q(\Cp)))(1-q(\Cp))=\pi(q(\Cp))$. Moving from public organizing to conditional release raises observed retaliation from zero and lowers effective suppression from one to $\pi(q(\Cp))$. Both changes are strict because $\pi(q(\Cp))\in(0,1)$.
\end{proof}

\begin{proof}[Proof of Corollary~\ref{cor:composition}]
In the nonatomic scalar protection case, a fixed point $y=b+H_S(y)$ contains exactly the supporters with $\rho_i\le y$. The safe-expansion iteration from the vanguard reaches the least fixed point $y^-$, while Theorem~\ref{thm:largest} identifies the greatest fixed point $y^+$ with the largest self-protecting coalition. Hence $\Cs=C_0\cup\{i\in S:\rho_i\le y^-\}$ and $\Cp=C_0\cup\{i\in S:\rho_i\le y^+\}$. Every $i\in\Cp\setminus\Cs$ therefore has $\rho_i>y^-$. Each recruited member $j\in\Cs\setminus C_0$ has $\rho_j\le y^-$, while every vanguard member has $\rho_j\le b\le y^-$. Thus the ordering holds for all members of the full cascade coalition. When $m(\Cs)>0$, the mass-weighted average requirement over $\Cp=\Cs\cup(\Cp\setminus\Cs)$ is a convex combination of the averages over the two pieces, which gives the inequality, strict when $m(\Cp\setminus\Cs)>0$. In the fight-or-fold model, a hidden supporter $i$ has $\Cs\cup\{i\}\notin\Pfam_i$ because $\Cs$ is a fixed point of $\Gamma$, so $e_i/h_i<\pi(q(\Cs\cup\{i\}))\le\pi(q(\Cs))$, where the second inequality holds because $q$ is nondecreasing and $\pi$ is nonincreasing. A cascade member $j$ has $\Cs\in\Pfam_j$, so $e_j/h_j\ge\pi(q(\Cs))$. Hence $e_i/h_i<e_j/h_j$ for every such pair.
\end{proof}

\clearpage
\let\section\ACbodysection
\bibliographystyle{aea}
\bibliography{safe_coalitions}

\ifbuildwithsupplement
  \ifdefined\PAPERWITHREFERENCEGUIDE
    \clearpage
    \begingroup
    \documentclass[safe_coalitions.tex]{subfiles}

\externaldocument[][nocite]{safe_coalitions_online_appendix}[safe_coalitions_online_appendix.pdf]

\begin{document}

\ifSubfilesClassLoaded{%
  \title{Reference Guide to ``Safe Collective Expression with Social Assurance Contracts''}
  \shortTitle{Reference Guide}
  \issueName{}
  \JEL{}
  \Keywords{}
  \ifanonymizeforpeerreview
    \author{Anonymous}
  \else
    \author{Matthew Cashman}
  \fi
  \date{\today}
  \begin{abstract}
  Notation and a lexicon of defined terms for the article and its Supplemental Appendix.
  \end{abstract}
  \maketitle
}{%
  \appendixparttitle{Reference Guide}
  Notation and a lexicon of defined terms for the article and its Supplemental Appendix.

  \let\section\ACbodysection
  \counterwithout{table}{section}
}

\setcounter{section}{0}
\setcounter{table}{0}
\renewcommand{\thesection}{\Roman{section}}
\renewcommand{\theHsection}{reference.\arabic{section}}
\renewcommand{\thetable}{R\arabic{table}}
\renewcommand{\theHtable}{reference.\arabic{table}}

Section references beginning with S point to the Supplemental Appendix; other section references point to the article, and ``Intro.'' to its introduction.

\section{Notation}\label{sec:reference_tables}

The table records where each symbol is introduced. As in the article and its appendix, a person \emph{follows through} when her attempt \emph{completes}: a public action completes when the statement becomes public under its author's name, and a private action completes when the commitment reaches its designated recipient.

{\footnotesize
\begin{longtable}{@{}>{\raggedright\arraybackslash}p{1.55in}>{\raggedright\arraybackslash}p{3.45in}>{\raggedright\arraybackslash}p{0.8in}@{}}
\caption{Notation, with the section that introduces each symbol.}\label{tab:notation}\\
\toprule
Symbol or term & Meaning & Where\\
\midrule
\endfirsthead
\multicolumn{3}{@{}l}{\emph{Table~\ref{tab:notation}, continued}}\\
\toprule
Symbol or term & Meaning & Where\\
\midrule
\endhead
\midrule
\multicolumn{3}{@{}r}{\emph{continued on the next page}}\\
\endfoot
\bottomrule
\endlastfoot
\multicolumn{3}{@{}l}{\emph{Notation: protective coalitions}}\\
$N,\;S$ & eligible prospective recruits, disjoint from $C_0$; supporter set $S\subseteq N$, those willing to be published if the published coalition protects them & \S\ref{sec:safe_coalitions}\\
$a_i$, $m(C)=\sum_{i\in C}a_i$ & a person's mass and a coalition's mass (the headcount when everyone counts equally) & \S\ref{sec:safe_coalitions}\\
$C_0,\;b$ & self-protecting initial public coalition and its mass $b=m(C_0)$ & \S\ref{sec:safe_coalitions}\\
$\Pfam_i$ & protection family: the full coalitions containing $i$ in which $i$ is safe, for $i\in C_0\cup S$ & \S\ref{sec:safe_coalitions}\\
$\rho_i$ & scalar protection requirement ($C\in\Pfam_i$ iff $m(C)\ge\rho_i$, for $C\ni i$); $\rho_i\le b$ for $i\in C_0$ & \S\ref{sec:safe_coalitions}\\
$T(C)$ & $C\cup\{i\in S\setminus C:C\cup\{i\}\in\Pfam_i\}$, the current coalition plus safe entrants & \S\ref{sec:safe_coalitions}\\
$\Gamma(C)$ & $C_0\cup\{i\in S:C\cup\{i\}\in\Pfam_i\}$; its fixed points are self-protecting and closed to safe entrants & \S\ref{sec:safe_coalitions}\\
$\Cs,\;\Cp$ & full cascade coalition $\Cs(S)=T^{|S|}(C_0)$ and largest self-protecting coalition; $\Cp(W)$ is the largest such $C$ with $C_0\subseteq C\subseteq C_0\cup W$ & \S\ref{sec:safe_coalitions}\\
$y$, $G_S(y)=b+H_S(y)$ & public mass; the mass that would be public if every supporter protected at $y$ joined (Figure~\ref{fig:cascade_tunnel}) & \S\ref{sec:safe_coalitions}\\
$y^-,\,y^u,\,y^+$ & crossings of $G_S$ with the diagonal: cascade mass, unstable middle crossing, largest safe mass & \S\ref{sec:safe_coalitions}\\
$C_t,\;A_t,\;R_t$ & the public coalition after date $t$; those who try to go public at $t$; those who follow through & \S\ref{sec:res}\\
$V$, $D(V)$ & the currently valid signers; the new endorsements released, $D(V)=\Cp(V)\setminus C_0$, with full outcome $\Cp(V)$ & \S\ref{sec:cascade_tunnel}\\
$r_i$ & risk tolerance: the probability of unprotected exposure $i$ accepts, conditional on exposure (none allowed at $r_i=0$); $e_i/h_i$ in the binary harm model & \S\ref{sec:res}\\
\multicolumn{3}{@{}l}{\emph{Notation: custody, risk, and beliefs}}\\
$C_\infty$ & the final public coalition & \S\ref{sec:dominance_supp}\\
$B_i$ & all other attempted public respondents in $i$'s appointed action round & \S\ref{sec:custody_risk_supp}\\
$\varepsilon$ & the probability that an attempt does not complete & \S\ref{sec:positive_tolerance_frontier}\\
$y_t$, $m_B$ & a call to action's target mass and the mass of its group & \S\ref{sec:positive_tolerance_frontier}\\
$s(y)$, $\Delta(y)$ & excess safe mass at target $y$; expected safety slack after incomplete attempts & \S\ref{sec:positive_tolerance_frontier}\\
$\underline{r}$ & the same-date group's minimum tolerance, $\min_{i\in B}r_i$ & \S\ref{sec:positive_tolerance_frontier}\\
$A$, $A_h$ & a scheduled round's attempted profile, fixed before its completion draws (frozen) & \S\ref{sec:acceptable_risk_supp}\\
$X_i$, $H_{ih}$, $\alpha_{ih}$ & the event that $i$ is exposed; the event that step $h$ is reached and her attempt completes; its share $\Pr(H_{ih}\mid X_i)$ & \S\ref{sec:acceptable_risk_supp}\\
$\delta$, $\delta^{\mathrm h}_i$ & leak probability; the harmful-leak rate for $i$ & \S\ref{sec:custody_risk_supp}\\
$m_i$ & the share of $i$'s exposure on which a step is reached without a leak and her own action completes & \S\ref{sec:custody_risk_supp}\\
$E(M)$ & the vanguard plus every recruit that mechanism $M$ exposes with positive probability & \S\ref{sec:dominance_supp}\\
$\mu$, $\mathcal M$, $\gamma$, $\Pfam_i(\mu)$ & audience belief, its attainable set, the strength of its effect, and belief-dependent protection & \S\ref{sec:design}\\
\multicolumn{3}{@{}l}{\emph{Notation: selection and fight or fold}}\\
$z_{it}$, $p_i$ & whether $i$ becomes public at date $t$; the completion probability of her attempt & \S\ref{sec:selection_supp}\\
$u_i(C)$, $\beta$ & payoff from expression in $C$, nonnegative exactly when $C$ protects $i$; discount factor & \S\ref{sec:scope_axioms}; \S\ref{sec:selection_supp}\\
$\hat\rho_i$ & a signer's reported protection requirement, to which the release rule is applied & \S\ref{sec:reporting_requirements}\\
$q(C)$ & resilience: the probability that coalition $C$ survives retaliation & \S\ref{sec:target}\\
$d$, $\kappa_T$, $q^\ast$, $F$ & the opposition's stakes, retaliation cost, break-even resilience $1-\kappa_T/d$, and the distribution of $q^\ast$ & \S\ref{sec:target}\\
$\pi(q)=(1-F(q))(1-q)$ & the probability that a coalition of resilience $q$ is suppressed & \S\ref{sec:target}\\
$e_i$, $h_i$ & $i$'s gain from expression; her loss if the coalition is suppressed & \S\ref{sec:target}\\
$\Pfam^{\mathrm{int}}_i$ & the expected-payoff protection family, the coalitions with $e_i\ge\pi(q(C))h_i$ & App.~\ref{app:proofs}\\
$\Lambda$, $\psi(m)$, $d_i$, $p_i(m)$, $\mu_i$ & per-member model: the distribution of the opposition's cost type; the cost scale of punishing a member of a mass-$m$ coalition; the harm member $i$ imposes; her punishment probability $\Lambda(d_i/\psi(m))$; the mass at which she accepts release & \S\ref{app:opponent}\\
\multicolumn{3}{@{}l}{\emph{Notation: supplement extensions}}\\
$\check\rho_i$, $K_S(y)$ & activation threshold $\rho_i-a_i$ for public entry with atomic mass; the mass of recruits whose activation threshold is at most $y$ & \S\ref{sec:gap}\\
$\Gamma_\mu$, $\Cs(\mu)$, $\Cp(\mu)$ & the operator and its least and greatest fixed points in the protection environment induced by belief $\mu$ & \S\ref{sec:design}\\
$\sigma$, $\sigma^\ast_i$, $C^\ast(\sigma)$ & an exogenous conditioning statistic, individual thresholds on it, and the coalition it protects & \S\ref{sec:bbf}\\
$\Sigma$, $\mathcal K(C_0)$ & the admissible public sets under nonmonotone protection; the coalitions safely reachable from the vanguard by one-person steps & \S\ref{sec:scope}\\
\end{longtable}
}

\section{Terms of Art}\label{sec:reference_lexicon}

The table records the paper's defined terms with the section that introduces each. Symbols are in Table~\ref{tab:notation}.

{\footnotesize
\begin{longtable}{@{}>{\raggedright\arraybackslash}p{1.55in}>{\raggedright\arraybackslash}p{3.35in}>{\raggedright\arraybackslash}p{0.9in}@{}}
\caption{Defined terms, with the section that introduces each.}\label{tab:lexicon}\\
\toprule
Term & Meaning & Where\\
\midrule
\endfirsthead
\multicolumn{3}{@{}l}{\emph{Table~\ref{tab:lexicon}, continued}}\\
\toprule
Term & Meaning & Where\\
\midrule
\endhead
\midrule
\multicolumn{3}{@{}r}{\emph{continued on the next page}}\\
\endfoot
\bottomrule
\endlastfoot
\multicolumn{3}{@{}l}{\emph{Institutions}}\\
Social Assurance Contract & an institution that collects endorsements of a statement in private and publishes them together only when the resulting coalition would protect every signer & Intro.\\
statement; endorsement & the controversial text to be published; a signature giving permission to publish it under the signer's name & Intro.\\
publication condition & the rule fixing when endorsements may be published and whose may appear, in the simplest case a number of signatures & Intro.\\
failure privacy & if the publication condition is never met, no signer's identity is revealed & Intro.\\
open letter & the base case: signatures published as they arrive, each exposed before the next joins & Intro.\\
call to action & asks people to add their names publicly at a specified time, without collecting endorsements beforehand, so each must bet that enough others follow through & Intro.; \S\ref{sec:positive_tolerance_frontier}\\
credible poll & verifies that respondents are eligible and publishes only their number; with private custody it certifies a self-protecting mass without naming anyone & Intro.; \S\ref{sec:cascade_tunnel}; \S\ref{sec:count_release}\\
cascade; tunnel & public organizing that adds names one at a time and stops at the exposure barrier; private collection and joint publication that reaches a larger protective coalition beneath it & Intro.\\
exposure barrier & the point at which the public cascade stops because no further person can join safely & Intro.; \S\ref{sec:safe_coalitions}\\
hidden supporters & members of the largest protective coalition whom the public cascade cannot reach & Intro.\\
\multicolumn{3}{@{}l}{\emph{Coalitions and protection}}\\
coalition & the set of people whose endorsement of the statement is public, including the vanguard & Intro.; \S\ref{sec:safe_coalitions}\\
protection & a coalition protects a member when speaking with it is at least as good for her as staying silent; it need not mean immunity from retaliation & Intro.; \S\ref{sec:safe_coalitions}\\
self-protecting (protective) coalition & a coalition that protects every member & \S\ref{sec:safe_coalitions}\\
vanguard & the group already public when recruitment begins, taken to be self-protecting & Intro.; \S\ref{sec:safe_coalitions}\\
supporter set; eligible population & those willing to be named if the published coalition protects them; the fixed set of people who may join & \S\ref{sec:safe_coalitions}\\
mass & headcount, weighted when some people count for more than others & \S\ref{sec:safe_coalitions}\\
opposition; allies & whoever can punish public support for the statement; whoever would act to protect the people named & \S\ref{sec:safe_coalitions}\\
architect; administrator & chooses the statement, the eligible population, and the publication condition; receives endorsements and publishes them & \S\ref{sec:safe_coalitions}\\
protection family & the public coalitions containing a person in which she is protected & \S\ref{sec:safe_coalitions}\\
monotone protection (safety in numbers) & adding public participants never makes an existing member less safe (Assumption~\ref{ass:M}) & Intro.; \S\ref{sec:safe_coalitions}\\
scalar protection case & protection depends only on total public mass, through one requirement per person & \S\ref{sec:safe_coalitions}\\
nonatomic case & a continuum of agents, each of negligible mass & \S\ref{sec:safe_coalitions}\\
safe-expansion operator; static safety check & $T$ adds exactly the people who can safely join the current public set, one round of public expression; $\Gamma$ keeps the vanguard and collects every supporter who would be safe in a candidate coalition or after joining it & \S\ref{sec:safe_coalitions}\\
cascade coalition; largest self-protecting coalition & where iterating $T$ from the vanguard stops; the union of all protective coalitions & \S\ref{sec:safe_coalitions}\\
regret-free expansion & growth that maintains protection at every step & \S\ref{sec:safe_coalitions}\\
\multicolumn{3}{@{}l}{\emph{Safety and follow-through}}\\
unprotected exposure; tolerance & being identified without the protection one requires; the probability of it a person accepts & \S\ref{sec:res}\\
follow through; complete & a person follows through when her attempt completes: a public action when her statement becomes public under her name, a private commitment when it reaches its recipient & \S\ref{sec:res}\\
follow-through uncertainty & unless co-participants are verified before exposure, any subset of a date's attempted public actions may fail to complete, including all but one (Assumption~\ref{ass:U}) & Intro.; \S\ref{sec:res}\\
verified before exposure; unverified attempt & a co-participant whose signature the institution has received and certified; an attempt to publish one's own endorsement before the institution holds the others needed & \S\ref{sec:res}\\
regret-free safety & no participant is ever named in a coalition that does not protect her, on any admissible realization (Assumption~\ref{ass:RES}) & Intro.; \S\ref{sec:res}\\
admissible realization & a pattern of who goes public and who does not that the institution's guarantee must cover & \S\ref{sec:res}\\
greedy process & everyone who can safely join the current public coalition joins at once, repeatedly & \S\ref{sec:cascade_tunnel}\\
all-or-nothing release & the administrator can publish a coalition's endorsements in a single event rather than as a sequence of individual exposures (Assumption~\ref{ass:atomic}) & \S\ref{sec:cascade_tunnel}\\
release rule; valid signatures & the administrator publishes together a subset $D(V)$ of the valid signatures $V$ it has received, chosen so that $C_0\cup D(V)$ is self-protecting; the tunnel uses $D(V)=\Cp(V)\setminus C_0$ & \S\ref{sec:cascade_tunnel}\\
\multicolumn{3}{@{}l}{\emph{Crossing the barrier}}\\
private custody & receiving and safeguarding endorsements that permit publication under a stated condition & \S\ref{sec:necessity}\\
general environment (E1) & actions are individual and irreversible with every completion subset admissible, information moves only through completed actions, devices act only on what they have received, and naming requires permission action by action (Assumption~\ref{ass:causal}) & \S\ref{sec:representation_necessity}; App.~\ref{app:representation}\\
exposure action; private commitment; device action & an action that identifies its author on completion; one that reaches a recipient while its author stays hidden; an institution's own publication & \S\ref{sec:representation_necessity}\\
leak & any route by which the opposition learns or infers a person's endorsement before release & \S\ref{sec:representation_necessity}\\
expression mechanism; public-only mechanism & the game between an administrator and the recruits; a mechanism whose publications condition only on completed exposure actions (Definition~\ref{def:mechanism}) & \S\ref{sec:representation_necessity}\\
private conditional-release stage & the first completed action that moves the public set past the cascade coalition, whose executing device already holds a still-private commitment from every agent it newly names other than possibly its author (Theorem~\ref{thm:necessity}) & \S\ref{sec:representation_necessity}\\
\multicolumn{3}{@{}l}{\emph{Accepting risk}}\\
scheduled action round & a call to action whose protection is assessed against the coalition already public plus everyone who completes in the appointed round; late arrivals do not count and cannot repair a shortfall & \S\ref{sec:positive_tolerance_frontier}\\
acceptable & a call to action or mechanism whose probability of unprotected exposure, conditional on being exposed, is at most each person's tolerance & \S\ref{sec:positive_tolerance_frontier}; \S\ref{sec:acceptable_risk_supp}\\
same-date group; bridging target; safety slack & the supporters outside the cascade whom a target would protect; a target from which safe expansion reaches $\Cp$; expected mass above the target after incomplete attempts & \S\ref{sec:positive_tolerance_frontier}\\
concentration bound & an upper limit on the probability that realized turnout falls short of the target; it shrinks as the group's mass is spread over more people & \S\ref{sec:positive_tolerance_frontier}\\
turnout risk; custody risk & the risk that others do not appear; the risk that held endorsements leak & \S\ref{sec:custody_frontier}\\
harmful leak; $\delta$-failure-private stage & a leak that identifies someone without protection; a stage that reveals commitments only through a protecting release unless a leak occurs & \S\ref{sec:custody_frontier}; \S\ref{sec:custody_risk_supp}\\
\multicolumn{3}{@{}l}{\emph{Participation and custody}}\\
costless signing conditions & signing costs nothing, stays hidden until release, never leaks, and is governed by a credible publication rule & \S\ref{sec:scope_axioms}\\
strict safe surplus & publication in a coalition that protects a person pays her strictly more than silence & \S\ref{sec:scope_axioms}\\
infiltration; false contract & signing to inflate the count while planning to recant after release; a contract that harvests names for the opposition, publishing a short list or never publishing as promised & \S\ref{sec:custody_practice}\\
\multicolumn{3}{@{}l}{\emph{The opposition}}\\
fight-or-fold model & the opposition, a single decision-maker, retaliates when the expected gain exceeds its cost & \S\ref{sec:target}\\
resilience; break-even survival probability; suppression probability & the chance a coalition survives retaliation; the resilience below which the opposition fights; the overall chance that expression is suppressed & \S\ref{sec:target}\\
expected-payoff protection & a coalition protects a member when her gain from expression covers her expected loss from suppression, $e_i\ge\pi(q(C))h_i$ & \S\ref{sec:target}\\
visibility--suppression reversal & bringing hidden supporters into view raises observed retaliation while lowering suppression of their expression & \S\ref{sec:target}\\
\multicolumn{3}{@{}l}{\emph{Supplement extensions}}\\
exposure step; frozen attempted profile; master seed & an exposure event or an entire scheduled round, assessed before its completion draws; a round's attempts, fixed before the draws; the randomization of administrator and participants & \S\ref{sec:acceptable_risk_supp}\\
final safe mass; confinement & the mass of public participants whom the final coalition protects; when each recruit requires a positive chance of protection, the vanguard plus the recruits exposed with positive probability is itself protective and lies in $\Cp$ & \S\ref{sec:dominance_supp}\\
robustly admissible; public participation game & a Markov strategy that keeps its user safe under every pattern of completed and incomplete attempts; the dynamic game of unverified public attempts in which strategies depend only on the current public coalition, promises carry no weight, and incomplete attempts are costless and unobserved & \S\ref{sec:selection_supp}\\
safe errors; class menus & reporting an upper bound on one's requirement cannot cause unprotected publication; classes with preset requirements stand in for reports & \S\ref{sec:reporting_requirements}\\
exposure barrier of positive depth; locally blocked; activation threshold & a level above the cascade at which protected mass falls short of public mass; a lower crossing with a shortfall immediately above it; a requirement net of the person's own mass & \S\ref{sec:gap}\\
exogenous-statistic limit & when protection depends only on a statistic that is primitive and independent of who acts, the cascade and the tunnel reach the same coalition at each value of the statistic & \S\ref{sec:bbf}\\
safely reachable coalitions; same-date group admissibility & coalitions reachable from the vanguard by safe one-person steps; every subset and order of a round's attempts is admissible (Assumption~\ref{ass:batch}) & \S\ref{sec:scope}\\
\end{longtable}
}

\end{document}

    \clearpage
    \endgroup
  \fi
  \clearpage
  \documentclass[safe_coalitions.tex]{subfiles}

\begin{document}

\ifSubfilesClassLoaded{%
  \title{Supplemental Appendix to ``Safe Collective Expression with Social Assurance Contracts''}
  \shortTitle{Supplemental Appendix}
  \issueName{}
  \JEL{}
  \Keywords{}
  \ifanonymizeforpeerreview
    \author{Anonymous}
  \else
    \author{Matthew Cashman\thanks{MIT Sloan, 100 Main Street, Cambridge, MA, USA. Email: \href{mailto:cashman@mit.edu}{cashman@mit.edu}.}}
  \fi
  \date{\today}
  \begin{abstract}
  This Supplemental Appendix contains supporting results, proofs, and extensions for the article.
  \end{abstract}
  \maketitle
}{ }

\startsupplement

\ifSubfilesClassLoaded{ }{%
  \appendixparttitle{Supplemental Appendix}
}

The appendix uses the article's terminology, except in the credible-poll variant below. To ``name'' an agent is to publish her specified statement, report, or other expression under her name. Accordingly, the public coalition is the set of people whose expressions have been published. Publication itself is the act of expression, rather than an announcement followed by a later choice to act.

Protection means that speaking with the published coalition is at least as good as silence. In the fight-or-fold model, that comparison uses expected payoffs, so a protective coalition can still suffer retaliation. Regret-free disclosure guarantees protection at every publication, not a favorable outcome in every fight. The risk calculations for scheduled public calls to action use the article's round-end assessment convention: only participation in the appointed round counts, with no credit for participation in later rounds.

The fixed initial public coalition $C_0$ is disjoint from the eligible prospective recruits $N$, with $S\subseteq N$ and $b=m(C_0)$. Every public coalition is a full coalition $C_0\subseteq C\subseteq C_0\cup S$, and protection families are defined on these full coalitions for all $i\in C_0\cup S$. The vanguard is self-protecting: $C_0\in\Pfam_i$ for every $i\in C_0$, with this condition maintained in each comparison or information environment considered below. Monotone protection preserves its protection as recruits join. In the scalar protection case, requirements are defined for all these agents, with $\rho_i\le b$ for $i\in C_0$. In the article's fight-or-fold model, expected-payoff protection is defined for all $i\in C_0\cup S$, with $e_i\ge\pi(q(C_0))h_i$ for initial members. Signing and public attempts are choices of recruits in $S$. The vanguard, already public, makes no such choice. The contract releases new endorsements $D(V)=\Cp(V)\setminus C_0\subseteq V$, giving full public outcome $C_0\cup D(V)=\Cp(V)$ and mass $m(\Cp(V))=b+m(D(V))$.

The supplement has three parts. The first supplies supporting arguments, the second examines robustness and empirical interpretation, and the third develops further extensions. Section labels and cross-references identify the individual results within each part.

\part*{\normalfont\large\bfseries I. Supporting Arguments}
\addcontentsline{toc}{part}{I. Supporting Arguments}

The article's Appendix~\ref{app:proofs} contains the proofs of every result stated in the article. This part supplies the additional arguments cited there: the institutional classification, details of the action-based representation, risk bounds, and results for verification and participation. The separate \emph{Reference Guide} collects the notation and the paper's defined terms.

\section{Institutional Features and Tightness}\label{sec:tightness_supp}

This section gives the institutional classification and the tightness remark summarized in Section~\ref{sec:representation_necessity} of the article.

To distinguish the institutions, consider what each records and releases. A poll records preferences rather than permission to publish; a pledge list records permission but a fixed publication date can expose a short list; a credible poll certifies support without naming anyone, and an allegation escrow publishes compatible reports together only when the group protects its members. These differences mean that the labels used in practice matter less than the order in which people are exposed and what the institution holds: whether it receives completed private commitments, conditions release on a protecting coalition, and publishes the attributable expressions together. That order determines whether the bound on regret-free public organizing (Theorem~\ref{thm:cascade}), the necessity theorem (Theorem~\ref{thm:necessity}), or the credible-poll result (Corollary~\ref{prop:count}) applies.

\begin{remark}[Each feature is necessary]\label{prop:tightness}
On a two-agent instance with an empty vanguard, neither agent safe alone, and the pair safe together, each of the three properties---private, conditional, verified---is individually necessary. Dropping conditionality permits publication of a short list, dropping privacy, so that verification happens in public, exposes the first person alone, and publishing one person's statement because the other has merely promised to sign can leave her visible when the promise is not kept. The same two-person example shows that dropping any one of the three properties breaks the guarantee against unprotected exposure.
\end{remark}

\section{Action-Space Representation Details}\label{sec:representation_details}

This section spells out the scope of Section~\ref{sec:necessity}. Attribution is defined relative to an audience. Revealing an attribute without identifying a person can change beliefs or future protection requirements, as discussed in Section~\ref{sec:design}, but does not enlarge the public coalition. An identity that becomes public without any completed action carrying that information is a leak and enters the rate of harmful leaks in Proposition~\ref{prop:robust_frontier}.

Lemma~\ref{lem:dichotomy} covers actions with deterministic attribution. It does not split one action with random attribution into separate exposure and commitment actions. Such a split would conflict with the assumption that every subset of attempted actions can complete: it would introduce a realization in which both new actions complete even when the original lottery had no corresponding outcome. Random attribution instead extends the risk analysis. At zero tolerance, unprotected attribution must have probability zero. This says nothing about individual outcomes of a continuous randomization, each of which has probability zero. At positive tolerance, its probability enters the person's risk of unprotected exposure. These observations do not produce an exact two-action representation.

The classification also applies to informal institutions: the device in Lemma~\ref{lem:custody} can be a person. A colleague who takes private custody of endorsements supporting a specified statement and then publishes the co-signed letter provides a custody stage. The still-private-commitment requirement concerns newly attributed recruits, not vanguard members or previously public recruits whose endorsements are reprinted. Trust affects which endorsements the delegate receives. The commitments she holds, together with the protection families, bound the coalition she can safely release. A delegate holding endorsements differs from people merely coordinating their own attempts: Remark~\ref{rem:imperfect} in the article distinguishes correlation from joint attribution. Under Assumption~\ref{ass:U}, coordinating separate attempts still permits a lone completion.

\section{Frontier Supporting Results}\label{sec:frontier_supp}

A public call to action risks incomplete turnout, whereas a contract risks premature disclosure from private custody. This section puts both risks on the same scale so they can be compared against a participant's tolerance for unprotected exposure.

\subsection{Acceptable Exposure Risk}\label{sec:acceptable_risk_supp}

Maintain monotone protection, Assumption~\ref{ass:U}, and the causal information and permission rules of Definition~\ref{def:mechanism} and E1. Extend their risk assessment to scheduled public rounds as in Section~\ref{sec:positive_tolerance_frontier}: fix all public attempts before their completion draws (a \emph{frozen} attempted profile), allow no decisions or device actions before the round ends, and assess a completing participant against the preceding public set plus everyone who completes that round. Incomplete or late attempts are costless and unattributed this round. All subsets remain possible; a scheduled round is not an all-or-nothing release. Outside these rounds, assess protection at each ordinary E1 exposure event. Later rounds cannot erase an assessed shortfall.

Each attempted public action and private submission completes independently with probability $1-\varepsilon\in(0,1)$, after the choice to act. An \emph{exposure step} is an ordinary exposure event or an entire scheduled public round. Its risk is computed before the relevant completion draws, and a round's assessment precedes any device action at its boundary. Steps are countably indexed, with finitely many attempts at each step, and first exposure and assessment occur at finite steps. A measurable master seed, independent of completion draws, includes all administrator and participant randomization. It may have a continuous distribution. Conditioning on the seed determines every contingent choice and completion-order rule, leaving only the completion indicators random. The same-date groups below consist of attempts in one scheduled round, not separate rounds on the same calendar date.

A mechanism is \emph{acceptable} at tolerance profile $r=(r_i)_{i\in S}$ if $\Pr(i\text{ exposed unprotected}\mid i\text{ exposed})\le r_i$ for every recruit. Unprotected exposure is assessed at her first identifying event, or at the end of her first identifying scheduled round, using the convention above. For someone the mechanism never newly exposes, including each vanguard member, take this risk to be zero. The vanguard remains protected by the starting condition and monotonicity. Unless specified otherwise, private solicitation is invisible and does not leak.

Regret-free safety requires every realized public coalition to meet the participant's protection criterion, which can already account for the benefit of speaking. With positive risk tolerance, she can instead accept some chance that the group will not protect her because the gains from expression make that risk worthwhile. Let $e_i$ be person $i$'s benefit from expression and $h_i$ her loss from unprotected exposure. Her tolerance is $r_i=e_i/h_i$: exposure is worth accepting when $e_i-p\,h_i\ge0$, that is, when the probability $p$ of unprotected exposure is at most $e_i/h_i$, so a larger benefit or smaller harm supports more risk. At a given exposure step, person $i$ receives the benefit only if her public action is completed, and in that case she also suffers the harm if the realized coalition does not protect her, so the local calculation conditions on her becoming public. For a whole mechanism, acceptability includes every route by which she can be exposed.

This binary risk calculation differs from the fight-or-fold model in Section~\ref{sec:target}, which gives one economic setting in which these benefits and harms arise. Here the harm is incurred exactly when the realized coalition does not protect her, so acceptability bounds the probability of that event, whereas there suppression occurs with probability $\pi(q(C))$ in any coalition and the criterion is expected utility, $e_i-h_i\,\mathbb E[\pi(q(C))\mid i\text{ exposed}]\ge0$. The two coincide when suppression is certain in an unprotective coalition and impossible in a protective one. These risk bounds use the binary model.

To bound risk, consider a participant who becomes public while everyone she relies on fails to appear. The next result weighs that event by how often the mechanism exposes her at the relevant step. Other patterns of incomplete turnout can also leave her unprotected, so this is a lower bound on her risk. The article's scalar call-to-action result instead uses the full distribution of who appears.

\begin{proposition}[Necessity at positive tolerance]\label{prop:interim_necessity}
(i) If $\Cs(S)\subsetneq\Cp(S)$ and $r_i\ge 1-(1-\varepsilon)^{|B|}$ for every agent $i$ of the nonempty same-date group $B=\Cp(S)\setminus\Cs(S)$, a single scheduled unverified call to action by that group is acceptable to every member and realizes $\Cp(S)$ with probability $(1-\varepsilon)^{|B|}>0$: with tolerant participants, a private stage is not necessary to reach beyond the cascade. (ii) Fix a mechanism and agent $i\in S$. Let $X_i$ be the event that $i$ is exposed, with $\Pr(X_i)>0$. Suppose that at an exposure step $h$, before its completion draws, $i$ is not yet public and $C\cup\{i\}\notin\Pfam_i$, where $C$ is the preceding public set. Her own unverified exposure attempt is accompanied by a fixed group $B_h\ne\varnothing$ of all other attempted public participants who could add names by her assessment. No other new names or device releases can occur before that assessment; in a scheduled round with frozen attempted profile $A_h$, $B_h=A_h\setminus\{i\}$. Let $H_{ih}$ be the event that $h$ is reached and $i$'s attempt completes, and write $\alpha_{ih}=\Pr(H_{ih}\mid X_i)$. Then
\[
\Pr(i\text{ is exposed unprotected}\mid X_i)\ \ge\ \alpha_{ih}\varepsilon^{|B_h|}.
\]
Consequently acceptability requires $\alpha_{ih}\varepsilon^{|B_h|}\le r_i$. For pairwise disjoint exposure events at these steps, the corresponding lower bounds add. Safe exposure on other branches can make $\alpha_{ih}<1$, so $r_i<\varepsilon^{|B_h|}$ alone does not rule out the exposure step. (iii) At $r=0$, for almost every master seed, every finite completion history of the fixed-seed mechanism is safe at its assessment times. Each scheduled round then admits a regret-free E1 refinement into individual completions, so Theorem~\ref{thm:necessity} applies for those seeds: any crossing requires its private conditional-release stage. The conclusion holds for almost every seed, not for every seed. Along any sequence of acceptable mechanisms with $r^n\to0$, each agent's mechanism-wide unprotected-exposure risk converges to zero, although it may remain positive at every $n$.
\end{proposition}

\begin{proof}[Proof of Proposition~\ref{prop:interim_necessity}]
(i) From public state $\Cs(S)$, have everyone in $B=\Cp(S)\setminus\Cs(S)$ attempt to go public in one scheduled round. Conditional on person $i$'s own attempt completing, she can be unprotected at the round-end assessment only if at least one member of $B\setminus\{i\}$ does not follow through in that round. This event has probability $1-(1-\varepsilon)^{|B|-1}\le1-(1-\varepsilon)^{|B|}\le r_i$, so she accepts. If everyone follows through, the public set is $\Cp(S)\in\Pfam_i$ for every member. That event has probability $(1-\varepsilon)^{|B|}$.

(ii) The attempted profile is fixed before the draws, so conditional on $H_{ih}$, independence gives probability $\varepsilon^{|B_h|}$ that no member of $B_h$ follows through. Since these are all the other possible new public participants before assessment and no device release intervenes, the assessed coalition is then exactly $C\cup\{i\}\notin\Pfam_i$. Moreover $H_{ih}\subseteq X_i$, so conditioning additionally on $X_i$ does not change these draws within $H_{ih}$. The harmful event therefore has conditional probability at least $\Pr(H_{ih}\mid X_i)\varepsilon^{|B_h|}=\alpha_{ih}\varepsilon^{|B_h|}$. Acceptability gives the inequality. The bounds add across disjoint exposure events at these steps. Safe exposure on other branches can make $\alpha_{ih}<1$ without changing the conditional noncompletion probability at this step.

(iii) Let $Z$ be the master seed, including participant as well as administrator randomization. Zero conditional risk implies zero unconditional probability of unprotected exposure for every recruit, including those never exposed. Since there are finitely many recruits, Fubini's theorem gives, for almost every $Z=z$, zero probability of any recruit's unprotected exposure conditional on that seed. This does not require $\Pr(Z=z)>0$. For a fixed seed, the remaining tree of finite completion histories is countable: each finite history has finitely many attempts and finitely many completion subsets. Every consistent finite prefix has strictly positive conditional probability, a finite product of factors $\varepsilon$ and $1-\varepsilon$. Thus an unsafe assessment at any such prefix would contradict conditional probability-zero harm. Equivalently, the countable union of finite-prefix failure events is null only if every such event is absent for almost every seed. This leaves one full-measure set of seeds on which all finite completion histories are safe.

Fix a seed in this set and a scheduled round reached after a finite history with public coalition $C$. For each attempted participant $i$, the outcome in which only her action completes has positive conditional probability, so $C\cup\{i\}\in\Pfam_i$. The preceding coalition is self-protecting by induction from $C_0$. Monotonicity therefore protects every member of every ordered prefix of every completed subset. Refine the round into its individual completions, retaining the original completion record and allowing the original continuation only after the whole round resolves and is assessed. This refinement preserves boundary public and received histories and copies the original device actions and subsequent decisions. E1's causal information and permission rules are unchanged, and every refined exposure is safe. Ordinary E1 events outside the round were already assessed individually. Theorem~\ref{thm:necessity} therefore applies to the refinement. A first crossing cannot occur within such a public round, since each of its safe singleton steps remains inside $\Cs$; its necessary private release is consequently also a stage of the original mechanism. This is a safety and crossing argument, not payoff equivalence for arbitrary within-round utilities, and it does not transfer positive-risk round-end bounds to the refinement.

All-path regret-freeness implies probability-zero harm, but the converse here holds only for the full-measure set of seeds just obtained. Infinite paths need not have positive probability. Their safety follows from the safety of each finite exposure prefix. Eventual completion of the release does not follow. In particular, a path on which a missing endorsement never arrives can prevent release even though free retries succeed almost surely. Finally, along $r^n\to0$, acceptability bounds each mechanism-wide conditional risk by $r_i^n$, proving convergence without requiring any positive-tolerance mechanism to eliminate risk.
\end{proof}

\subsection{Custody and Crossing Risk Bounds}\label{sec:custody_risk_supp}

For the sufficiency result, supporters sign when solicited, and the administrator continues requesting any missing endorsement until it arrives. Requests and incomplete submissions are costless, and recruitment reveals no attributable information. Independent positive completion probabilities then give eventual receipt in finite time almost surely. Retries do not overcome a persistent refusal to sign. The following extension allows private custody to leak.

Real repositories of private commitments can leak (though this threat may be minimal with modern cryptographic systems). Let $\delta$ bound the probability of an unauthorized leak outside the modeled completion process. The leak process is independent of completion draws and the mechanism's master seed. A stage is $\delta$-failure-private if it reveals commitments only through a protecting release unless a leak occurs. A leak is harmful when it first identifies someone to the opposition without the protection she requires. Let $\delta^{\mathrm h}_i$ be the part of person $i$'s total exposure risk caused by such leaks, conditional on her being exposed. This term accounts for who learns the identities: disclosure to the opposition alone need not supply the protection that publication to allies and the wider audience would supply.

\begin{proposition}[Custody and crossing risk bounds]\label{prop:robust_frontier}
\begin{enumerate}
\item[(i)] (\emph{Sufficiency.}) Let the administrator tunnel toward a self-protecting target $C$: it solicits exactly the new members $C\setminus C_0$ until it has received every one of their endorsements, and then publishes those endorsements to form $C$. On a $\delta$-failure-private stage, each $i\in C\setminus C_0$ is eventually named almost surely, through the protecting release unless a leak names her first, and she is unprotected only when a leak exposes her as part of a coalition that does not protect her. Hence her conditional risk is $\delta^{\mathrm h}_i\le\delta$, and the mechanism is acceptable exactly when $r_i\ge\delta^{\mathrm h}_i$; the common bound $r_i\ge\delta$ is sufficient.
\item[(ii)] (\emph{Necessity, with rates.}) Conversely, consider any mechanism acceptable to its participants and a recruit $i$ exposed with positive probability. Fix an exposure step and a frozen attempted profile $A$. Suppose the mechanism can first expose the not-yet-safe agent $i$ through her own unverified attempt, with $B_i=A\setminus\{i\}$ comprising all other attempted public participants and no other new names or device releases before assessment, as in Proposition~\ref{prop:interim_necessity}(ii). Let $m_i$ be the probability, conditional on $i$ being exposed, that this step is reached, no leak occurs before or during it through assessment, and her own public action completes. The independent leak process leaves the other participants' completion draws unchanged on this event. Then
\[
\delta^{\mathrm h}_i+m_i\,\varepsilon^{|B_i|}\;\le\;r_i.
\]
The inequality combines two sources of unprotected exposure: a harmful leak, and going public while relying on others who may not follow through. Both count against the same tolerance for unprotected exposure, $r_i$.
\item[(iii)] (\emph{No unprotected exposure.}) Suppose $\Cs(S)\subsetneq\Cp(S)$. At $r=0$, harmful leaks have probability zero. For mechanisms whose public-coalition transitions still obey E1's action, permission, and receipt rules, with scheduled rounds as above, almost surely any first crossing of $\Cs$ contains Theorem~\ref{thm:necessity}'s private conditional-release stage. Scheduled public rounds cannot supply that crossing, by Proposition~\ref{prop:interim_necessity}(iii). The conclusion holds for almost every seed, and it does not cover an unauthorized leak that itself publishes a protective coalition, which lies outside E1's permission rules.
\end{enumerate}
\end{proposition}

A small public leak can expose too few people to protect one another, whereas a larger public disclosure can provide protection under monotonicity. Stealing the whole list for the opposition need not protect anyone, so the number of stolen names alone does not determine $\delta^{\mathrm h}_i$. Custody safeguards reduce this risk; larger, reliable answering groups with positive slack reduce follow-through risk. The inequality records how these two sources of risk use the same individual tolerance, without solving the cost of reducing either one.

\begin{proof}[Proof of Proposition~\ref{prop:robust_frontier}]
(i) On the no-leak event, the mechanism eventually publishes the endorsements of $C\setminus C_0$, forming the target $C$, which protects every member. The vanguard is protected throughout. Each new member $i\in C\setminus C_0$ is exposed almost surely: either a leak exposes her too early or the protecting release exposes her as intended. Her conditional probability of unprotected exposure is therefore exactly the harmful-leak contribution $\delta^{\mathrm h}_i$. Because leak probability is at most $\delta$, $\delta^{\mathrm h}_i\le\delta$. This proves both criteria.

(ii) Condition on $X_i=\{i\text{ is exposed}\}$. Let $H_i$ be the event defining $m_i$: the step is reached, it remains leak-free through assessment, and her own public action completes. The group $B_i$ is fixed before the independent draws, the leak process is independent of those draws, and $H_i\subseteq X_i$. Thus, on an $m_i\varepsilon^{|B_i|}$ share of $X_i$, agent $i$ becomes public while no member of $B_i$ follows through. These are all the other attempted public participants, and no device release precedes assessment, so the assessed coalition is exactly $C\cup\{i\}\notin\Pfam_i$. This event is disjoint from harmful first identification by a leak, including a leak during the round; that route contributes $\delta^{\mathrm h}_i$. Hence
\[
\Pr(i\text{ unprotected}\mid X_i)\ \ge\ \delta^{\mathrm h}_i+m_i\varepsilon^{|B_i|}.
\]
Acceptability now gives the displayed inequality. The same reasoning allows us to add the contributions from disjoint exposure events at these steps.

(iii) When $r=0$, every nonnegative risk term in part (ii) must vanish: harmful first identification by a leak and the counted unverified singleton exposure have probability zero. For the E1 public-coalition mechanisms specified in part (iii), include the independent leak process in the master seed. Proposition~\ref{prop:interim_necessity}(iii) then supplies the almost-every-seed finite-history argument and the refinement of scheduled rounds, requiring a private conditional-release stage at any first crossing almost surely. Probability-zero harmful leakage alone does not establish this conclusion for an external disclosure that bypasses E1's permission or receipt rules. Nor does the argument constrain null seeds or guarantee completion on every retry path. Along any sequence $r^n\to0$, the inequality implies $\delta_{i,n}^{\mathrm h}\to0$ and $m_{i,n}\varepsilon^{|B_i|}\to0$ for every fixed exposure step. At positive tolerance, however, neither term need equal zero.
\end{proof}

\section{Confinement and Verified Attainment}\label{sec:dominance_supp}

How large a coalition can an acceptable mechanism expose? Even when participants accept some risk, requiring each exposed person to have a positive chance of protection puts a common bound on the people who can become public. A contract can attain that bound when it can collect the necessary endorsements.

Use the finite population, monotone protection, and self-protecting vanguard of the article, with the conditional-risk and exposure-assessment conventions of Section~\ref{sec:acceptable_risk_supp}. Let $M$ be a mechanism acceptable at tolerances $r_i<1$ for every recruit. Write $E(M)$ for the vanguard together with all recruits whom $M$ exposes with positive probability, and $C_\infty$ for its final public coalition. Its \emph{final safe mass} is $m(\{i\in C_\infty:C_\infty\in\Pfam_i\})$, the mass of public participants whom the final coalition protects.

For attainment, additionally assume leak-free private custody, all-or-nothing release, and costless, invisible solicitation. Every supporter attempts to sign when solicited and keeps trying if her endorsement does not arrive. Submissions complete independently with probability $1-\varepsilon\in(0,1)$ on each attempt. The administrator publishes once collection is complete. These collection assumptions do not overcome persistent refusal to sign.

\begin{theorem}[Confinement and verified attainment]\label{thm:dominance}
Every mechanism $M$ acceptable at tolerances $r_i<1$ for every recruit has a self-protecting exposed set $E(M)$, so $E(M)\subseteq\Cp(S)$ and its final safe mass is at most $m(\Cp(S))$ almost surely. Under the additional collection assumptions, the contract publishes $D(S)=\Cp(S)\setminus C_0$ in finite time almost surely, producing $\Cp(S)$ without unprotected exposure. It therefore weakly lowers each person's exposure risk and weakly raises final safe mass relative to any such mechanism.
\end{theorem}

\begin{proof}[Proof of Theorem~\ref{thm:dominance}]
\emph{Confinement.} Fix $i\in E(M)\setminus C_0$. Acceptability and $r_i<1$ give a positive conditional probability that $i$ is protected at her first-exposure assessment. This is the end of her scheduled round or her ordinary E1 exposure event. Since the public set takes finitely many values, some full assessed coalition $X_i$ containing $i$ occurs with positive probability and satisfies $X_i\in\Pfam_i$. Every recruit in $X_i$ is public by that assessment, so each is exposed with positive probability and $X_i\subseteq E(M)$. Monotone protection therefore gives $E(M)\in\Pfam_i$. For $i\in C_0$, the starting condition and $C_0\subseteq E(M)$ give the same conclusion. Thus $E(M)$ is self-protecting, and Theorem~\ref{thm:largest} implies $E(M)\subseteq\Cp(S)$. Because there are finitely many recruits, almost every realized public coalition is a subset of $E(M)$. Its final safe mass is consequently at most $m(E(M))\le m(\Cp(S))$. This argument uses protection at first exposure, not protection acquired in a later round.

\emph{Attainment.} The contract solicits all supporters and repeats each request until the endorsement arrives. A given endorsement remains missing after $k$ attempts with probability $\varepsilon^k\to0$. Since $S$ is finite, all endorsements arrive in finite time almost surely. Private custody keeps the public coalition at the self-protecting vanguard until publication. The single release then forms $\Cp(S)$, which protects every member by Theorem~\ref{thm:largest}. Every person's unprotected-exposure risk is zero, and final safe mass equals $m(\Cp(S))$ almost surely. Together with confinement, this proves the comparison.
\end{proof}

The bound depends on each recruit requiring some chance of protection. A person willing to accept certain unprotected exposure, $r_i\ge1$, can be a risky first mover outside $\Cp$. The attainment claim requires continued signing attempts and free, private retries. It holds almost surely, not on a path where an endorsement never arrives. The comparison ranks exposure risk and final safe mass, rather than delay, cost, or welfare.

\section{Public Participation under Regret-Free Safety}\label{sec:selection_supp}

This section asks whether strategic public participation reaches the cascade coalition when participants choose only actions that guarantee protection. It states and proves the public-selection result summarized in the article's Section~\ref{sec:scope_axioms}. The signing result and its proof are in the article.

The initial public state is $C_0$. At each time point, each supporter $i\in S$ whose statement is not yet public chooses whether to try to publish it under her name. Ordinary practical obstacles may keep that choice from producing a public statement, so $z_{it}\in\{0,1\}$ records whether person $i$ actually becomes public at date $t$. Conditional on choosing to attempt publication, she receives a completion draw with $\Pr(z_{it}=1)=p_i\in(0,1)$. Completion is independent across people and dates and does not depend on the public history. Because the supporter set $S$ is finite, every attempt to go public completes with probability at least $p=\min_{i\in S}p_i>0$ when $S$ is nonempty. With $S=\varnothing$, the public coalition remains $C_0=\Cs(S)$. An incomplete attempt costs nothing, is unobserved, and does not change the public coalition. The possibility that others do not follow through makes a public attempt risky, as in Section~\ref{sec:res}, while independent completion lets repeated safe attempts reach the cascade coalition almost surely, without assuming that they do. The attempts are unverified before exposure because a person cannot confirm that anyone else's statement has become public before her own does.

When her attempt completes, person $i$ receives a continuing payoff $u_i(C)$ that weakly rises as the coalition grows. The payoff is nonnegative exactly when coalition $C$ protects her. In the fight-or-fold interpretation, this payoff averages over the opposition's type and the outcome of retaliation. The comparisons below therefore concern each possible pattern of public participation, with payoffs evaluated before retaliation is resolved. A person whose statement remains unpublished receives zero. Total payoffs are expected discounted sums of these per-period payoffs, with discount factor $\beta\in(0,1)$.

A Markov strategy depends only on the current public coalition. It is \emph{robustly admissible} if it keeps the person safe under every pattern of completed and incomplete attempts that the model allows. A robust Markov-perfect equilibrium requires each person to choose her best strategy among these safe strategies. \emph{Strict safe surplus}, $u_i(C)>0$ whenever $C\in\Pfam_i$, means that safe expression is strictly better than silence. It therefore removes any reason to wait once speaking is safe.

\begin{theorem}[Sequential safe-expression selection]\label{thm:selection}
Under the follow-through process above, Assumptions~\ref{ass:M} and~\ref{ass:U}, and strict safe surplus, every robust Markov-perfect equilibrium has the same action rule. At each public state $C$, exactly the agents in $T(C)\setminus C$ try to go public. The public set therefore follows the greedy path, $C_t=T(C_{t-1})$, apart from incomplete attempts and retries. No equilibrium path reaches outside $\Cs$. Nor can an equilibrium \emph{strategically} stall below it: at every public state, each agent who is safe to join tries to go public immediately rather than waiting. Because attempts complete with probability bounded below by the $p>0$ above, the public coalition reaches $\Cs$ almost surely in finite time.
\end{theorem}

Theorem~\ref{thm:cascade} established that $\Cs$ is the largest coalition a regret-free public process can reach, and the selection theorem shows that the process actually reaches it, with no strategic delay. Every person in $T(C)\setminus C$ remains safe even if hers is the only attempt to go public that completes, so strict surplus makes immediate action better than waiting.

The conclusion depends on both the payoff and the safety standard. If safe expression produces exactly zero surplus, waiting may also be optimal, and if people trust promises that others will act later, that trust can support coordinated movement beyond $\Cs$. This \emph{public participation game} assumes that strategies depend only on the current public coalition, that promises of future action carry no weight, silent people do not benefit from others' expression, and incomplete attempts are costless and unobserved. It also gives each person only two possible public statuses: either her statement has been published under her name, or it has not. The theorem therefore selects among robustly admissible strategies, which impose the worst-case safety standard instead of deriving it from expected-utility maximization. A person who assigns small probability to being stranded and values expression highly enough might rationally attempt to cross $\Cs$, which is the positive-tolerance margin of Section~\ref{sec:positive_tolerance_frontier}. Because silent people gain nothing from others' expression here, free riding on a coalition's success lies outside this game. If incomplete attempts become visible, they can change protection requirements and later behavior; the theorem does not cover that environment.

The proof compares acting now with waiting by running both choices on the same sequence of completion draws. The next lemma shows that the path with earlier participation always contains the waiting path.

\begin{lemma}[Path coupling]\label{lem:coupling}
Under the state-independent follow-through process of this section, fix a public state $C$, an agent $i\in T(C)\setminus C$, and a Markov profile that is greedy at every public state strictly containing $C$ and arbitrary but robustly admissible at $C$. Run two copies of the process from $C$ on identical completion draws at every date: in copy $A$, agent $i$'s date-$0$ attempt completes; in copy $W$, agent $i$ waits whenever the state is $C$. Then the public sets satisfy $D^A_t\supseteq D^W_t\cup\{i\}$ for every $t\ge0$, realization by realization.
\end{lemma}

\begin{proof}[Proof of Lemma~\ref{lem:coupling}]
Measure each state at the end of a date. At date~$0$, both copies begin at $C$. Every person other than $i$ takes the same action in both copies and receives the same completion draw. Let $R_0$ be the common set of those whose attempts complete. Then $D^A_0=C\cup\{i\}\cup R_0\supseteq C\cup R_0\cup\{i\}=D^W_0\cup\{i\}$.

Now suppose $D^A_t\supseteq D^W_t\cup\{i\}$. Consider any $j\ne i$ who attempts in copy $W$ at state $D^W_t$ and is not already public in copy $A$. If $D^W_t=C$, robust admissibility at $C$ requires $j\in T(C)\setminus C$. If $D^W_t\supsetneq C$, greedy play requires $j\in T(D^W_t)\setminus D^W_t$. In either case, $j\in T(D^W_t)$. Since $D^W_t\subseteq D^A_t$, Lemma~\ref{lem:mono} gives $j\in T(D^A_t)\setminus D^A_t$. Moreover, $D^A_t\supseteq C\cup\{i\}\supsetneq C$, so play in copy $A$ is greedy and $j$ attempts there too. The two copies use the same completion draw for each person and date. Thus anyone who becomes public in copy $W$ at date $t+1$ either becomes public in copy $A$ at the same date or was already public there. The inclusion is preserved by induction.
\end{proof}

\begin{proof}[Proof of Theorem~\ref{thm:selection}]
Fix a robust Markov-perfect equilibrium, a public state $C$, and a person $i\in S\setminus C$. First suppose $i\notin T(C)$. Then $C\cup\{i\}\notin\Pfam_i$. Public attempts are unverified, so Assumption~\ref{ass:U} allows only $i$'s attempt to complete. That realization exposes her in a coalition that does not protect her. Attempting is therefore not robustly admissible, and $i$ cannot attempt in equilibrium.

Now suppose $i\in T(C)\setminus C$, so $C\cup\{i\}\in\Pfam_i$. By Lemma~\ref{lem:reduction}, every realization $R$ in which her attempt completes satisfies $C\cup R\supseteq C\cup\{i\}$. Assumption~\ref{ass:M} gives $C\cup R\in\Pfam_i$. Attempting is therefore safe in every realization and is robustly admissible. Waiting pays zero in the current period.

Proceed by backward induction on $|S\setminus C|$. Assume that play is greedy at every strict superset of $C$. Compare attempting with waiting under identical completion draws at every date. If $i$'s attempt does not complete and no one else follows through, the state remains $C$ and the same comparison begins again. If someone else follows through while $i$ does not, the two paths coincide from then on: $i$'s strategies differ only at $C$, and public states never shrink. It remains to compare the case in which $i$ follows through with the case in which she waits.

After $i$ follows through, Lemma~\ref{lem:coupling} gives $D'_t\supseteq D_t\cup\{i\}$ at every later date. The lemma applies because the induction hypothesis supplies greedy play at every strict superset of $C$. When the waiting path eventually names $i$ in $D_t\cup\{i\}$, the completion path pays $u_i(D'_t)\ge u_i(D_t\cup\{i\})$ by monotonicity of $u_i$. Before that date, the completion path pays $u_i(\cdot)\ge0$, while waiting pays zero. Assumption~\ref{ass:M} also preserves safety along the completion path because $C\cup\{i\}\in\Pfam_i$. Completion is therefore weakly better in every realization of public participation, with payoffs evaluated as above. With probability $p_i>0$, an attempt completes immediately and produces a payoff of at least $u_i(C\cup\{i\})>0$. Strict safe surplus therefore makes attempting strictly optimal in expectation.

At state $C$, the unique equilibrium action therefore has exactly the people in $T(C)\setminus C$ attempt. It follows that $C_t\subseteq T(C_{t-1})$. Starting from $C_0$, Lemma~\ref{lem:mono} and induction keep the entire path inside $\Cs$. The path cannot stop at $C_0\subseteq C\subsetneq\Cs$. If $T(C)=C$, monotonicity of $T$ and $C_0\subseteq C$ would imply $T^n(C_0)\subseteq C$ for every $n$, and hence $\Cs\subseteq C$, a contradiction. Thus $T(C)\setminus C$ is nonempty at every such proper subset of $\Cs$, and those people attempt. Repeated attempts complete independently with probability at least $p>0$. Every recruit in $\Cs\setminus C_0$ is therefore named in finite time almost surely, and the public set reaches $\Cs$ almost surely.
\end{proof}

The article's Proposition~\ref{prop:fullsign} separately treats the contract's one-shot signing game under the costless signing conditions. Under strict safe surplus, its dominance argument selects signing by every recruit in $\Cp(S)\setminus C_0$, without relying on the dynamic public game's completion process. The two-person example in Section~\ref{sec:scope_axioms} shows why an unreimbursed signing cost can make the all-abstain equilibrium strict. Proposition~\ref{prop:robust_frontier} bounds exposure risk from leakage. It does not select participation under that risk.

\part*{\normalfont\large\bfseries II. Robustness and Empirical Interpretation}
\addcontentsline{toc}{part}{II. Robustness and Empirical Interpretation}

The preceding results leave practical questions about what signers know, what an administrator can safeguard, and how the opposition chooses its targets. This part considers how people report their protection requirements, what threatens private custody, and how selective retaliation affects the interpretation of public expression. The experimental evidence supports the behavioral premises of the comparison and does not test its full institutional characterization.

\section{Reporting Protection Requirements}\label{sec:reporting_requirements}

The rule that publishes the largest protective coalition needs a protection requirement for each signer. This section asks what happens when signers supply those requirements themselves and when an uncertain signer reports cautiously.

Let each signer $i\in V\subseteq S$ report a requirement $\hat\rho_i$ with her signature, and let the administrator apply the rule to the reports: it selects the largest full coalition $C_0\subseteq C\subseteq C_0\cup V$ such that $m(C)\ge\hat\rho_j$ for every $j\in C\setminus C_0$, then releases $C\setminus C_0$. The vanguard's fixed true requirements satisfy $\rho_j\le b$ and require no reports. To evaluate that release, use the same payoff convention as in the article's signing game in Section~\ref{sec:scope_axioms}: a newly published person earns $u_i(C)$, which is nonnegative exactly when $m(C)\ge\rho_i$, and an unpublished recruit earns zero. True requirements therefore determine the payoff $u_i(C)$ even when reports determine publication. For part~(i), additionally assume that $u_i(C)$ is nondecreasing as the coalition expands. Part~(i) also uses the signing game's strict safe surplus condition: $u_i(C)>0$ whenever $C$ protects her.

\begin{proposition}[Reporting requirements: truthful reports and safe errors]\label{prop:report}
In the scalar protection case under Assumptions~\ref{ass:M} and~\ref{ass:atomic}, with every supporter signing and the others' reports held fixed:
\begin{enumerate}
\item[(i)] (\emph{Truthful reporting.}) For a signer who knows $\rho_i$, reporting $\hat\rho_i=\rho_i$ is weakly better than any other report against every profile of others' reports, and under strict safe surplus it is strictly better than some other report against some profile.
\item[(ii)] (\emph{Safe errors.}) Any report $\hat\rho_i\ge\rho_i$ guarantees that $i$ is never published unprotected, the full public coalition and new release are weakly decreasing in $\hat\rho_i$, and only a report below $\rho_i$ can publish her unprotected.
\item[(iii)] (\emph{Class menus.}) If the architect announces a menu of classes whose membership is verifiable or assigned by the architect, each with a requirement at least as large as the true requirement of any member of that class, then a signer's class stands in for a report of the kind in part (ii). If signers choose their class freely, part (ii) guarantees protection whenever the chosen class requirement is at least their true requirement. The administrator can also round a submitted requirement up to an available class requirement, preserving this guarantee for truthful or conservative submissions. Part (i)'s truth-telling conclusion does not automatically extend to a rounded menu.
\end{enumerate}
\end{proposition}

With exact requirements, truthful reporting is optimal. With a coarse menu, rounding a truthful requirement upward preserves protection but may exclude someone who could safely participate. For a signer who does not know her requirement but knows an upper bound, reporting that bound fails safe: it can exclude her from a protecting coalition but cannot place her in an unprotecting one. Uncertainty about protection therefore reduces the final safe mass rather than safety.

The same logic supports simpler publication rules. A preannounced total threshold is the one-class case. Class-based rules, such as a quota of tenured signers, extend the idea to composition requirements that the scalar protection case does not cover.

There is an institutional analogue in the experiment: the subjects who chose higher disclosure thresholds when they held the controversial position demanded more surrounding support before accepting attribution, although their thresholds were strategic choices rather than reports of known upper bounds on protection requirements \citep{braghieriThresholdDisclosureCollective2026}.

\begin{proof}[Proof of Proposition~\ref{prop:report}]
Fix $i\in V$ and the others' reports, and write $D(\hat\rho_i)$ for the new endorsements released when $i$ reports $\hat\rho_i$. Reported families are upper sets in mass and leave the vanguard's requirements unchanged, so Assumption~\ref{ass:M} holds for them and Theorem~\ref{thm:largest} applies: $C_0\cup D(\hat\rho_i)$ is the union of all full coalitions $C_0\subseteq C\subseteq C_0\cup V$ with $m(C)\ge\hat\rho_j$ for every $j\in C\setminus C_0$. Every such coalition also protects the vanguard because $m(C)\ge b\ge\rho_j$ for $j\in C_0$. To compare reports, note that raising $\hat\rho_i$ removes some coalitions containing $i$ from this family and adds none, so $D$ and $C_0\cup D$ are weakly decreasing in $\hat\rho_i$. I first prove the safety claim in part (ii), then use it to compare payoffs in part (i).

(ii) If $\hat\rho_i\ge\rho_i$ and $i\in D(\hat\rho_i)$, then $b+m(D(\hat\rho_i))\ge\hat\rho_i\ge\rho_i$, so $i$ is protected and earns $u_i(C_0\cup D(\hat\rho_i))\ge0$; if $i\notin D(\hat\rho_i)$ she earns zero. Either way she is never published unprotected, and by monotonicity a higher report can only remove her from a coalition. If instead $\hat\rho_i<\rho_i$, any profile of others' reports under which $i\in D(\hat\rho_i)$ with $b+m(D(\hat\rho_i))<\rho_i$ publishes her unprotected, and such profiles exist whenever some full coalition $C_0\subseteq C\subseteq C_0\cup V$ containing $i$ has mass in $[\hat\rho_i,\rho_i)$.

(i) Compare the truthful report with a higher report $\hat\rho_i'>\rho_i$. Monotonicity gives $D(\hat\rho_i')\subseteq D(\rho_i)$. If $i\in D(\hat\rho_i')$, then $i\in D(\rho_i)$ and both full coalitions protect her, so $u_i(C_0\cup D(\rho_i))\ge u_i(C_0\cup D(\hat\rho_i'))$ because $u_i$ is nondecreasing. If $i\notin D(\hat\rho_i')$, she earns zero, which is at most her truthful payoff by part (ii). Now compare with a lower report $\hat\rho_i''<\rho_i$, so $D(\hat\rho_i'')\supseteq D(\rho_i)$. If $i\notin D(\hat\rho_i'')$, then $i\notin D(\rho_i)$ and both payoffs are zero. If $i\in D(\hat\rho_i'')$ and $b+m(D(\hat\rho_i''))\ge\rho_i$, then $C_0\cup D(\hat\rho_i'')$ satisfies every constraint under the truthful report as well, so it is contained in $C_0\cup D(\rho_i)$ and the two coincide. If $i\in D(\hat\rho_i'')$ and $b+m(D(\hat\rho_i''))<\rho_i$, she earns $u_i(C_0\cup D(\hat\rho_i''))<0$, below any truthful payoff. Truthful reporting is therefore weakly better against every profile. For strictness under strict safe surplus, let every other signer report zero. If $b+m(V)\ge\rho_i$, then $D(\rho_i)=V$ and $u_i(C_0\cup V)>0$, while a report above $b+m(V)$ excludes her and pays zero. If instead $\rho_i>b+m(V)$, truthful reporting excludes her and pays zero, whereas reporting zero releases all of $V$ and gives her a strictly negative payoff. Thus strict comparison is available even when no full coalition of available signers and the vanguard can protect her.

(iii) When membership is verifiable or assigned, a class requirement at least as large as every member's true requirement acts as a report $\hat\rho_i\ge\rho_i$ for each member, so part (ii) applies. When a signer chooses her class, the selected requirement is the report used by the release rule, and the same safety conclusion holds if it is at least $\rho_i$. Rounding a truthful or conservative submission upward to an available class requirement preserves that inequality. This establishes protection, rather than the incentive to submit truthfully: rounding can change the released coalition even when the submitted requirement is exact.
\end{proof}

\section{Custody in Practice}\label{sec:custody_practice}

Private collection must protect stored endorsements, avoid identifying supporters during recruitment, and make promised publication credible. This section examines those requirements and the threats that can undermine them, expanding the article's Sections~\ref{sec:custody_frontier} and~\ref{sec:scope_axioms}.

Custody safeguards aim to prevent the opposition from identifying supporters before a protective release. One response is to hold nothing longer than necessary and delete failed deposits promptly. Another is to separate certifying the threshold from releasing identities. This helps exactly when the condition can be checked without a readable list of names existing anywhere.

Threshold cryptography can provide this separation: endorsements are encrypted or secret-shared across parties, the condition is computed by multiparty computation or threshold decryption \citep{evansPragmaticIntroductionSecure2018,breuerSeCritMassThresholdSecret2024,mangipudiCollusionDeterrentThresholdInformation2023}, and zero-knowledge proofs verify claims about what is held \citep{kosbaHawkBlockchainModel2016}. With these safeguards, a contract can ensure that no single party, the administrator included, can read the list before the condition is met \citep{breuerSeCritMassThresholdSecret2024,mangipudiCollusionDeterrentThresholdInformation2023}, and the release itself can be automatic \citep{kosbaHawkBlockchainModel2016}. That shifts the residual risk from the administrator's discretion to credentials, key custody, code correctness, and the devices and connections people sign from, so premature disclosure requires collusion among key holders, compromised endpoints, or a flaw in the code rather than one administrator's betrayal.

Where such tools are unavailable, custody depends on administrators' ability and willingness to keep endorsements confidential.

Recruitment creates another opportunity for information to escape. The model assumes that soliciting and signing generate no attributable information. When they do---when signing, or the metadata around it, reveals to the opposition that a person endorsed---the model records a leak. Danger from merely being approached, before any endorsement exists, lies outside the model. Because the contract's guarantee covers only what enters custody, the boundary that matters is whether discovering, reaching, or signing the contract creates attributable information the opposition can get.

Not all solicitation crosses this boundary. Announcing that a contract exists need not create any: a broadcast advertisement, or a site anyone can visit, solicits no one in particular, and even hostile publicity spreads the address. A targeted ask can create such information, since it reveals who was asked, and so can a public site through traffic or signing metadata, even without targeted recruitment. A site with open enrollment avoids targeted requests but does not guarantee privacy.

Private signing need not be costless to attract participants. By reducing expected retaliation costs, the contract can make participation more attractive, but signing must still be worthwhile compared with remaining silent. When these risks deter signing, the coalition the contract can publish, $\Cp(V)$, can shrink with the signer set (Theorem~\ref{thm:tunnel}).

On this recruitment margin the call to action and the contract can both broadcast. The call to action's remaining edge is that it holds nothing, which is why it stays the relevant competitor whenever custody cannot be trusted.

\runinhead{Responses available against either institution.} The comparison holds fixed the protection that each public coalition provides, isolating what private collection changes. If the opposition's response changes those protection conditions equally under both institutions, the comparison applies in the changed environment. For example, a response aimed at the whole eligible class---raising class-wide penalties, conducting a general purge, or adopting a chilling policy---can be applied to either the contract or public organizing. The changed environment must remain monotone and leave the vanguard protected for Theorems~\ref{thm:cascade} and~\ref{thm:tunnel} to apply. Such a response can narrow the gap between the cascade and the tunnel but cannot reverse their ranking. Targeted retaliation, by contrast, needs names, and the contract withholds them until the coalition can protect its members. This leaves contract-specific responses, including infiltration and false contracts, for the following discussion.

\runinhead{Infiltration and false contracts.} Two responses exploit features of the contract itself. An infiltrator may sign to inflate the count and plan to recant after release. This is related to the Sybil attack, in which an adversary multiplies identities to swing a count, but authentication forces the infiltrator to use a real identity. However, it is not the case that \textit{anyone} may serve as an infiltrator. Contract architects may hedge against this by excluding people publicly committed to the opposition from the eligible class, and so prevent them from signing. Therefore, infiltrators must be those whose position is publicly uncertain. In the case where such an infiltrator recants, she is indistinguishable from an endorser who caved to the opposition's retaliation---which both limits the damage, since the coalition loses some protection but perhaps not her whole weight, and discourages infiltration in the first place.

A false contract instead harvests names for the opposition. It either publishes a short list or never publishes as promised. In either case the harvested names are in hostile custody, a failure with a leak rate of one. Protection against this threat depends on the signer's judgment about who is running the contract, perhaps a vanguard she recognizes or an institution with a reputation to lose. The guarantee is conditional on that judgment.

The companion paper \citep{cashmanConditionalDisclosureCoordination} treats the two threats differently. It models false-contract risk as an entry decision: a signer participates only when her trust in the administrator clears a threshold that rises with the harm a leak would do. For bad-faith endorsements, it instead studies robust publication design: when immediate recantation is bounded and only continuing commitment supplies protection, a sharp margin above the minimum safe size guarantees that the remaining coalition still protects its members.

\runinhead{Competing contracts.} Competing contracts may split signatures below every threshold, so a common platform may need a rule for pooling compatible campaigns. Because failed deposits are destroyed, carrying an endorsement to a compatible successor requires the signer's standing permission.

\section{The Fight-or-Fold Model: Per-Member Result}\label{app:opponent}

Section~\ref{sec:target} makes retaliation an all-or-none choice against the coalition. The reversal also holds when the opposition, still modeled as a single decision-maker, decides whom to punish one member at a time. To model this choice, suppose punishing a member of a coalition with mass $m$ costs $\kappa_T \psi(m)$, where $\psi$ is continuous, strictly increasing, and positive. Punishment prevents the harm $d_i>0$ that member $i$ would impose on the opposition. The opposition therefore punishes $i$ exactly when $d_i\ge\kappa_T \psi(m)$. As in the coalition-level model, its cost type is drawn once and held fixed when comparing coalitions. Here $\kappa_T$ has a continuous, strictly increasing cdf $\Lambda$ with full support on $(0,\infty)$. This extension does not impose the article's bound $\kappa_T<d$.

In the article, a coalition can survive a fight. Here punishment suppresses the targeted member's expression, while an unpunished member's expression remains public. The proposition has four parts: acceptance is a mass threshold; the reversal holds member by member; already-visible members are punished with weakly lower probability after the release; and punishment concentrates on the most damaging members.

\begin{proposition}[Per-member reversal and selective retaliation]\label{prop:permember}
Let $p_i(m)=\Lambda\!\big(d_i/\psi(m)\big)$ be the probability that member $i$ of a full public mass-$m$ coalition is punished, let $h_i>0$ and $e_i\ge0$ for all $i\in C_0\cup S$, and let $i$ be protected iff $e_i\ge p_i(m)\,h_i$. Maintain $e_i\ge p_i(b)h_i$ for $i\in C_0$, so the initial vanguard is self-protecting; recruits accept release exactly when protected. Then: (i) acceptance is a mass threshold: $i$ accepts iff $m\ge\mu_i$. When $e_i=0$, $\mu_i=\infty$. When $0<e_i<h_i$, $\mu_i=\psi^{-1}\!\big(d_i/\Lambda^{-1}(e_i/h_i)\big)$. If the argument exceeds $\psi$'s range, set $\mu_i=\infty$, so the agent never accepts; if it falls below the range, set $\mu_i=0$. When $e_i\ge h_i$, $\mu_i=0$. At a common tolerance $e/h$, the cutoff $\mu_i$ is weakly increasing in $d_i$. The induced protection families are upward sets, with $\mu_i\le b$ for $i\in C_0$, and $\Cs,\Cp$ exist by Lemma~\ref{lem:interim}; (ii) for every hidden supporter $i\in\Cp\setminus\Cs$, replacing public organizing with expansion to $\Cp$ raises observed retaliation on her from $0$ to $p_i(m(\Cp))\in(0,1)$ and lowers suppression of her expression from $1$ to the same interior number: the reversal holds member by member; (iii) every already-visible member $j\in\Cs$ is punished with weakly lower probability after the release, $p_j(m(\Cp))\le p_j(m(\Cs))$; (iv) at any full public mass, punishment concentrates on the most damaging members: $p_i(m)$ is increasing in $d_i$.
\end{proposition}

\begin{proof}[Proof of Proposition~\ref{prop:permember}]
(i) The protection condition $e_i\ge p_i(m)h_i$ is equivalent to $\Lambda(d_i/\psi(m))\le e_i/h_i$. If $e_i=0$, it fails at every finite mass: $d_i/\psi(m)>0$, and full support gives $\Lambda(d_i/\psi(m))>0$. Hence $\mu_i=\infty$. If $e_i\ge h_i$, the condition holds for every $m$ because $\Lambda\le1$. For $0<e_i<h_i$, continuity, strict monotonicity, and full support give protection exactly when $d_i/\psi(m)\le\Lambda^{-1}(e_i/h_i)$, or equivalently when $\psi(m)\ge d_i/\Lambda^{-1}(e_i/h_i)$. This defines a mass threshold $\mu_i$ that rises with $d_i$ and falls with $e_i/h_i$. Protection is upward in mass and therefore in set inclusion. The initial condition gives $\mu_i\le b$ for $i\in C_0$, so $C_0$ is self-protecting. Lemma~\ref{lem:interim} then supplies the lattice objects.

(ii) Under public organizing, a hidden supporter never becomes public. The data record no punishment against her, but her expression is completely suppressed. Observed retaliation is therefore $0$ and suppression is $1$. Under expansion to $\Cp$, which is self-protecting under the induced families, she accepts publication. The opposition punishes her exactly when $\kappa_T\le d_i/\psi(m(\Cp))$. This event has probability $p_i(m(\Cp))=\Lambda(d_i/\psi(m(\Cp)))\in(0,1)$ because $d_i>0$ and $\Lambda$ is interior on $(0,\infty)$. Her expression is suppressed exactly when she is punished. Observed retaliation rises from $0$ to $p_i(m(\Cp))$, while suppression falls from $1$ to the same probability.

(iii) Because $m(\Cp)\ge m(\Cs)$ and $\psi$ is increasing, $d_j/\psi(m(\Cp))\le d_j/\psi(m(\Cs))$. Applying the increasing function $\Lambda$ proves the claim.

(iv) The result follows directly because $\Lambda$ is increasing.
\end{proof}

The reversal therefore also holds when the opposition can punish individual members rather than fight the coalition as a whole. Beyond the reversal, part (iv) produces the selective-discharge pattern found in the union evidence: the opposition targets the members who impose the greatest harm, and the minimum coalition mass at which a member accepts release, $\mu_i$, rises with that harm. In union organizing campaigns, the discharge of union activists is documented as a distinct employer tactic and described as a signal meant to deter other supporters \citep{bronfenbrennerNoHoldsBarred2009,schmittDroppingAxIllegal2009}. Consistent with that interpretation, estimates that translate discharge counts into a per-activist risk rest on the explicit assumption that firings target activists \citep{schmittDroppingAxIllegal2009}. Estimates of illegal organizer discharge remain contested \citep{schmittDroppingAxIllegal2009,lalondeHardTimesUnions1991,sherkTruthImproperFirings2007}.

\section{Experimental Evidence and Measurement}\label{sec:experiment_supp}

Section~\ref{sec:target} of the article describes the experiment in \citet{braghieriThresholdDisclosureCollective2026}, in which 298 students voted under private, public, or threshold-disclosure rules. This section records the point estimates and sets out measurement questions for conditional publication.

Linking names to votes raised abstention from 21.7 to 37.1 percent and reduced expression of the socially controversial position from 43.4 to 22.9 percent. Under threshold disclosure, the point estimates for participation and controversial expression were close to those under private voting, while one-third of participants had their votes publicly revealed.

Voters taking the controversial position chose a median disclosure threshold of 75 percent, against 45 percent among other voters. These thresholds concern the share of agreeing votes, including votes that remain private. Unlike the randomized treatment comparisons, the difference in chosen thresholds by vote direction is descriptive.

\runinhead{What to measure.} Beyond these behavioral premises, evaluating the institution requires observing both public participation and the support that remains private. The model separates five quantities of interest.
\begin{enumerate}
\item Whether and when a coalition is released. Under conditional release the publication condition is either met or not, so release is a discrete event that arrives with a delay, and the release probability and the delay are informative about the distance from the vanguard to $y^+$ rather than to $y^-$, given a model of how endorsements arrive.
\item The share of the eventual coalition that is public at the initial release against the share that joins afterward. Public organizing shows names arriving one at a time and stalling; the contract shows a block at the release date followed by public joiners whom the block protects, and Theorem~\ref{thm:necessity} says that anything reaching past the barrier while guaranteeing protection at publication must show such a block.
\item The composition of the release by exposure and private benefit, which Corollary~\ref{cor:composition} predicts.
\item Retaliation against people who become public, alongside suppression over the eligible population, counting those who remain silent. Proposition~\ref{thm:rd} shows why more observed attacks can accompany less suppression; successful punishments alone also omit suppressed expression that never became public.
\item The separate effects of information and conditional publication. A low-stakes design could vary credible information about support independently of whether signed endorsements are published immediately or only as a protecting group. Comparing release, composition, and exposure before protection across these conditions would distinguish improved turnout forecasts from the value of collected commitments. Follow the eligible population in every condition, including nonparticipants and campaigns that never release.
\end{enumerate}

\part*{\normalfont\large\bfseries III. Further Extensions}
\addcontentsline{toc}{part}{III. Further Extensions}

These extensions ask what changes when the output is a count, when coalition composition matters, or when a larger group can harm a member. Each can be read separately after the article; none is needed to prove its main results.

\section{Credible Polls: Supporting Result}\label{sec:count_release}

A credible poll uses private custody to certify support without publishing endorsements under their authors' names. The allies who would act on that support receive the certified full-coalition mass, including the vanguard. Only a confidentiality-bound administrator sees the recruits' identities. The following corollary states what the count certifies.

\begin{corollary}[Credible poll]\label{prop:count}
In the nonatomic scalar protection case, let $\Pfam_i=\{C\ni i:m(C)\ge\rho_i\}$ on full coalitions, with $\rho_i\le b$ for $i\in C_0$. Replace publication of $D(V)=\Cp(V)\setminus C_0$ by publication of the count $b+m(D(V))$. The rule certifies a self-protecting full-coalition mass, at least $\rho_i$ for every member of $C_0\cup D(V)$, and certifies $y^+$ when $V=S$. It directly names no recruit, so the named public coalition remains $C_0$ and the rule adds no naming-based safety constraint. This certifies support, not protection for a later self-identifying act.
\end{corollary}

\begin{proof}[Proof of Corollary~\ref{prop:count}]
The threshold families are monotone in mass, so Theorem~\ref{thm:largest} makes $\Cp(V)$ self-protecting. Because $C_0$ and $D(V)$ are disjoint, the published count is $b+m(D(V))=m(\Cp(V))$, which meets every member's requirement. At $V=S$, equation~\eqref{eq:granovetter} identifies this mass with $y^+$. Publication of a count does not itself publish anyone's endorsement, leaving the self-protecting vanguard as the named public coalition. Theorem~\ref{thm:necessity} concerns transitions that newly name recruits, so a credible poll is not a public-only route past $\Cs$.
\end{proof}

The corollary assumes private accumulation of endorsements. It does not analyze observable signing or the participation incentives it creates. Certifier knowledge alone is not a harmful leak, but disclosure that identifies a signer to the opposition without protection enters $\delta_i^{\mathrm h}$ in Proposition~\ref{prop:robust_frontier}. Count publication also need not preserve anonymity. Side information, a small candidate set, or distinctive individual masses can identify signers, and a count covering the whole eligible population identifies everyone. Those inferences require a model of audience information. The certificate's value for later public action likewise depends on how allies and the opposition respond and does not follow from the coalition calculation alone.

\section{Heterogeneity and Coalition Expansion}\label{sec:gap}

With heterogeneous protection requirements, their distribution determines how much larger a coalition the tunnel can form than the public cascade. This section develops the fixed-point geometry and comparative statics used in the article. Let $H_S(y)=m\{i\in S:\rho_i\le y\}$ and $G_S(y)=b+H_S(y)$. The least and greatest fixed points, $y^-\le y^+$, are the public masses of $\Cs$ and $\Cp$ in the scalar versions of Theorems~\ref{thm:largest} and~\ref{thm:cascade}.

\begin{proposition}[Exposure barriers and coalition expansion]\label{prop:gap}
In the scalar nonatomic case, with continuous $H_S$:
\begin{enumerate}
\item[(i)] the additional public mass assembled by the tunnel is the fixed-point distance, $m(\Cp)-m(\Cs)=y^+-y^-$;
\item[(ii)] (\emph{Barriers of positive depth.}) The inequality $y^-<y^+$ is equivalent to $G_S$ having multiple fixed points. An \emph{exposure barrier of positive depth}---a level $y>y^-$ at which the protected mass falls short of the mass needed to sustain it, $G_S(y)<y$, with a recrossing above it ($G_S(y')\ge y'$ for some $y'>y$)---is sufficient for $y^-<y^+$ and is also necessary when the lower crossing $y^-$ is \emph{locally blocked} ($G_S(y)<y$ for all $y$ in a right-neighbourhood of $y^-$). Without local blockage, multiple fixed points can instead arise through a critical-mass or tangency case, in which the cascade mass $y^-$ is not robust to an upward perturbation but no barrier of positive depth separates it from $y^+$;
\item[(iii)] raising $H_S$ only on $(y^-,\infty)$ leaves the cascade mass $y^-$ unchanged and weakly raises the largest protective mass $y^+$; raising $H_S$ on $[b,y^-]$, or raising the vanguard $b$ with the prospective recruit population fixed and the vanguard self-protecting in each environment, weakly raises $y^-$. These distributional comparisons concern masses, not the identities of coalition members.
\end{enumerate}
\end{proposition}

\begin{proof}[Proof of Proposition~\ref{prop:gap}]
(i) At the fixed points, $\Cp=C_0\cup\{i\in S:\rho_i\le y^+\}$ and $\Cs=C_0\cup\{i\in S:\rho_i\le y^-\}$. Hence $m(\Cp)=b+H_S(y^+)=y^+$ and $m(\Cs)=b+H_S(y^-)=y^-$. Subtracting gives the result.

(ii) By definition, $y^-<y^+$ exactly when $G_S$ has at least two fixed points. For sufficiency, take $y_0>y^-$ such that $G_S(y_0)<y_0$, and suppose $G_S(y')\ge y'$ for some $y'>y_0$. Continuity of $G_S$ and the intermediate value theorem give a fixed point in $(y_0,y']$. Therefore $y^+>y_0>y^-$.

Now impose local blockage. If $y^-<y^+$ and $G_S(y)<y$ on $(y^-,y^-+\nu)$, choose any $y_0\in(y^-,\,y^-+\min\{\nu,\,y^+-y^-\})$. Then $G_S(y_0)<y_0$, and the graph recrosses the diagonal at $y'=y^+$. This is an exposure barrier of positive depth.

Local blockage is necessary for this converse. Let $b=1$ and $G_S(y)=1+\sqrt{y-1}$ on $[1,2]$. The fixed points are $1$ and $2$, but $G_S(y)>y$ throughout $(1,2)$. Thus $y^-<y^+$ even though no barrier of positive depth separates them. Here $y^-=b$, so the cascade does not leave the vanguard, but any upward perturbation at the unstable lower fixed point can start growth.

(iii) The iteration that reaches $y^-$ evaluates $G_S$ only at masses no greater than $y^-$. Thus $y^-$ depends only on $H_S$ over $[b,y^-]$. Raising $H_S$ only on $(y^-,\infty)$ leaves $y^-$ unchanged and weakly raises $G_S$ above it, so the greatest fixed point $y^+$ weakly rises. Raising $H_S$ on $[b,y^-]$, or raising the vanguard $b$, raises $G_S$ along the public iteration and weakly raises $y^-$. The function $H_S$ records masses and says nothing about which named person has which requirement, so these comparisons of masses imply neither unchanged membership nor set inclusion across environments.
\end{proof}

This is the standard geometry of a threshold cascade \citep{granovetterThresholdModelsCollective1978}. Failure privacy matters for the support beyond the cascade coalition, not for total latent support: a movement can have broad private backing but little public expression, and the more support is concentrated beyond that coalition, the more the contract adds. In the union application, these are the most exposed workers, whom only the tunnel reaches.

In a finite coalition, each person has positive mass, and the cascade and the tunnel destination use different thresholds. A recruit of mass $a_i$ can safely join existing public mass $y$ exactly when $y+a_i\ge\rho_i$, so her \emph{activation threshold} $\check\rho_i:=\rho_i-a_i$ governs public entry. Writing $K_S(y)=m\{i\in S:\check\rho_i\le y\}$, the cascade iterates $b+K_S(y)$ from $b$, while the tunnel destination $\Cp$ still follows the final requirements $H_S(y)=m\{i\in S:\rho_i\le y\}$ through $y=b+H_S(y)$. The distinction disappears in the nonatomic limit $a_i\to0$, and a person with high mass can start further growth despite a high protection requirement because her own mass helps satisfy it.

\subsection{Information}\label{sec:design}

The preceding geometry holds the protection offered by a given coalition fixed. Information could instead reduce the expected cost of speaking by changing what the opposition believes---an employer who learns that most of the staff agree may retaliate less, which in the scalar protection case lowers each person's requirement---but it need not make being the first to join worthwhile. Let an intervention select a belief $\mu$ from a fixed set $\mathcal M$, and let $\Pfam_i(\mu)$ list the coalitions that protect person $i$ under that belief. Maintain the starting condition that $C_0$ protects its members under each belief being compared. If $C_0\cup\{i\}\notin\Pfam_i(\mu)$ for every attainable $\mu$, no intervention that changes only what the audience believes, such as a credible poll, lets $i$ join the vanguard alone, though it may still start the cascade by making someone else a safe first mover. If some attainable belief does make $i$ safe with the vanguard, information suffices.

Information and private custody therefore solve different parts of the problem. By changing audience beliefs, information changes which coalitions protect every member, and, under the regret-free standard with protection held fixed, private custody changes which protective coalition can be reached. In the scalar protection case, favorable audience beliefs lower requirements from $\rho_i$ to $\rho_i-\gamma\mu$, where $\gamma\ge0$ measures the strength of this effect. The change moves both fixed points. For any fixed belief, the path still determines whether expression reaches the lower or upper crossing.

Write $\Gamma_\mu$, $\Cs(\mu)$, and $\Cp(\mu)$ for the coalition operator, cascade coalition, and largest self-protecting coalition in the protection environment $\Pfam_i(\mu)$.

\begin{remark}[Information and the path constraint]\label{cor:info}
At fixed $\mu$, a belief-only intervention that releases no named coalition leaves the coalition fixed points unchanged: a public-only regret-free process remains bounded by $\Cs(\mu)$, while the Social Assurance Contract implements $\Cp(\mu)$. The observation is immediate from the fixed-belief theorems and is recorded for the design discussion in the article.
\end{remark}

\subsection{Exogenous Conditioning Statistics: a Boundary Case}\label{sec:bbf}

The tunnel can expand the coalitions that can form safely only when appearing with others changes whether a person is safe. To mark that boundary, this subsection asks what happens when safety instead depends on an external statistic that coalition membership cannot change. At each value of the external statistic, anyone who could safely appear with others could also safely join the vanguard alone. The result compares which coalitions are safe at a given value of the statistic. It does not model participation by people who do not yet know that value.

\begin{remark}[The exogenous-statistic limit]\label{rem:bbf}
Threshold majority voting with disclosure conditions release on a statistic rather than on a protecting coalition \citep{braghieriThresholdDisclosureCollective2026}. The realized vote share in that model is endogenous to participation and voting, so the proposition below does not formally nest it. The proposition instead isolates a stricter boundary: at any fixed value of a conditioning statistic that is primitive and independent of who acts, the custody stage cannot expand the coalitions that can form safely through the cascade. When safety depends on who else is exposed, coalition composition and the vanguard again determine which recruits can join, so the tunnel may form a larger protective coalition than the cascade.
\end{remark}

\begin{proposition}[The exogenous-statistic limit]\label{prop:bbf}
Suppose each protection family conditions only on an exogenous statistic: there is a random variable $\sigma$, realized independently of the mechanism and of who is named, and thresholds $\sigma^\ast_i$ for all $i\in C_0\cup S$ with $\Pfam_i=\{C\ni i:\sigma\ge\sigma^\ast_i\}$ on full coalitions. Restrict attention to realizations and vanguards satisfying $\sigma\ge\sigma_i^\ast$ for every $i\in C_0$, so the initial coalition remains self-protecting. Then: (i) $\Gamma(C)=C^\ast(\sigma):=C_0\cup\{i\in S:\sigma\ge\sigma^\ast_i\}$ for every full $C$, the operator has the unique fixed point $C^\ast(\sigma)$, and $\Cs(S)=\Cp(S)=C^\ast(\sigma)$; (ii) for fixed $S$ and $\sigma$, the additional recruits $C^\ast(\sigma)\setminus C_0$ are vanguard- and history-independent, while the full coalition includes the initial vanguard; (iii) at each realized value of $\sigma$, failure privacy does not expand the set of coalitions that can form safely: every recruit in $C^\ast(\sigma)\setminus C_0$ is safe when joining the vanguard alone. With signer set $V$, the contract releases the new endorsements $D(V)=V\cap C^\ast(\sigma)$, giving full outcome $C_0\cup D(V)$. Conversely, $\Cs\subsetneq\Cp$ requires protection to depend on co-exposure.
\end{proposition}

\begin{proof}
(i) For a full coalition $C\ni i$, $C\in\Pfam_i$ exactly when $\sigma\ge\sigma^\ast_i$, regardless of who else belongs to $C$. Vanguard members satisfy this condition by hypothesis. Hence $\Gamma(C)=C_0\cup\{i\in S:\sigma\ge\sigma_i^\ast\}=C^\ast(\sigma)$ for every full $C$. The operator has the unique fixed point $C^\ast(\sigma)$, so its least and greatest fixed points coincide. Every recruit $i\in C^\ast(\sigma)\setminus C_0$ is safe when joining $C_0$ alone. Thus $T(C_0)=C^\ast(\sigma)$: one application of the cascade operator gives the full protective coalition at that value of the statistic. For fixed $S$ and $\sigma$, the entrant set depends on neither the vanguard nor the history, but the full coalition retains $C_0$. With available signers $V$, the same argument gives $\Cp(V)=C_0\cup(V\cap C^\ast(\sigma))$, proving parts (ii)--(iii). Conversely, $\Cs\subsetneq\Cp$ would require some $i\in\Cp\setminus\Cs\subseteq S$ such that $\Cp\in\Pfam_i$ but $\Cs\cup\{i\}\notin\Pfam_i$. That is impossible when protection is independent of coalition composition.
\end{proof}

\section{Scope of Monotone Protection}\label{sec:scope}

The lattice results require monotone protection: adding public participants cannot make an existing member less safe. Material retaliation does not guarantee this property. This section studies what remains when monotonicity fails.

Without monotone protection, there may be no largest self-protecting coalition. Let $C_0=\varnothing$, $S=\{1,2,3\}$, and take
\[
\Pfam_1=\{\{1,2\},\{1,3\}\},\qquad
\Pfam_2=\{\{1,2\}\},\qquad
\Pfam_3=\{\{1,3\}\}.
\]
The coalitions $\{1,2\}$ and $\{1,3\}$ are both self-protecting, but neither contains the other, and their union does not protect its members. After all three agents sign, the administrator must choose whether to exclude agent 2 or agent 3. The greatest-coalition implementation in Theorems~\ref{thm:largest}--\ref{thm:tunnel} therefore fails.

Without monotonicity, a small, curated contract could perhaps still check protection without revealing the coalition to a proposed entrant by privately re-evaluating each potential new coalition. However, this means the mechanism must then collect and process potentially complex compatibility constraints, choose among conflicting protective coalitions, and perhaps reject an entrant, remove an incumbent, or seek renewed consent. That is a discretionary coalition-selection problem, and it does not scale into a practical mechanism.

\runinhead{The nonmonotone frontier.} Monotonicity can fail in at least two ways. In social-image models, exposure costs depend on audience beliefs and perceived norms \citep{benabouIncentivesProsocialBehavior2006,bursztynMisperceivedSocialNorms2020}. Adding a disliked co-signer can raise the reputational cost borne by everyone else. A new member's record or vulnerability can also reduce coalition resilience $q(C)$ (Section~\ref{sec:target}) and make the group easier to defeat. In either case, the union of two self-protecting coalitions can fail to protect some of their members, so a largest self-protecting coalition need not exist. The relevant question becomes which protective coalitions are accessible.

Let $C_0$ be the irreversible self-protecting vanguard. Let $\Sigma\subseteq\{C:C_0\subseteq C\subseteq C_0\cup S\}$ contain the admissible public sets, including $C_0$. Every member of an admissible set, including every vanguard member, must be protected. Each public step adds one recruit and must leave the new set in $\Sigma$. The \emph{safely reachable coalitions} $\mathcal K(C_0)$ are the coalitions reachable from $C_0$ through these safe one-person steps.

Under monotone protection, Lemma~\ref{lem:reduction} shows that safety after each individual completion guarantees safety under every realization. Without monotonicity, a later addition can make an earlier member unsafe. This extension retains ordinary E1 event-by-event assessment. The scheduled-round convention does not apply here. The next assumption therefore strengthens Assumption~\ref{ass:U} for public same-date groups of unverified attempts. It is a separate assumption, not a consequence of nonmonotonicity.

\begin{assumption}[Same-date group admissibility]\label{ass:batch}
For a public transition in which a set $E$ of unverified attempts by not-yet-public recruits is made, the admissible realization set contains, for every subset $E'\subseteq E$ and every ordering of $E'$, the realization in which exactly the members of $E'$ complete, one at a time, in that order. A device action releasing completed commitments of previously unattributed recruits already in private custody is instead governed by Assumption~\ref{ass:atomic}.
\end{assumption}

The analysis takes $\Sigma$ as known, so to implement coalition $C$ the administrator must know $\Sigma$, the vanguard $C_0$, and the target $C$. With those inputs given, the next proposition characterizes reachability. It neither elicits protection sets nor chooses among incomparable protective coalitions, and outside the monotone domain there may be no largest destination.

\begin{proposition}[Accessibility necessity]\label{thm:accessibility}
Maintain the permission-respecting action space and causal timing preceding Lemma~\ref{lem:dichotomy}, mechanisms as in Definition~\ref{def:mechanism}, and Assumptions~\ref{ass:atomic} and~\ref{ass:batch}. For any admissible family $\Sigma$ with irreversible self-protecting vanguard $C_0$ as above: (i) a coalition is the outcome of some regret-free public-only process from $C_0$ if and only if it lies in the safely reachable coalitions $\mathcal K(C_0)$; (ii) every $C\in\Sigma\setminus\mathcal K(C_0)$ requires a private conditional-release stage to implement, and taking private custody of completed commitments from $C\setminus C_0$ before releasing those names in one release event implements it. A private conditional-release stage is therefore necessary, and custody followed by one release is sufficient, for exactly the admissible coalitions outside the reachable set. When $\Sigma$ is generated by monotone protection (Assumption~\ref{ass:M}), $\max\mathcal K(C_0)=\Cs$ and $\max\Sigma=\Cp$, recovering Theorem~\ref{thm:necessity} for full coalitions.
\end{proposition}

\begin{proof}
(i) A regret-free public-only process exposes names irreversibly and conditions transitions only on completed public actions. Consider a transition in which a set $E$ attempts to become public. No commitment in private custody controls this transition, so Assumption~\ref{ass:batch} treats it as an unverified public same-date group. Regret-free safety must hold for every possible subset and order of completions. Each prefix of each completion order, added to the current public set, must therefore belong to $\Sigma$. Replace every realized same-date group with one such safe sequence of individual completions and join these sequences across transitions. The result is a one-person-at-a-time chain from $C_0$ to the final coalition. Hence the outcome lies in $\mathcal K(C_0)$.

Conversely, suppose a chain shows that $C\in\mathcal K(C_0)$. Have each new member attempt alone in the order given by the chain. Every intermediate set belongs to $\Sigma$, so the process is regret-free and reaches $C$.

(ii) Now take $C\in\Sigma\setminus\mathcal K(C_0)$. No safe one-person public chain reaches it. By part (i), unverified public attempts alone cannot carry the process outside the reachable set. In the maintained permission-respecting action space, an event newly attributing several recruits can eliminate partial-completion outcomes only by using one device action based on their completed commitments that remain private. This is Lemma~\ref{lem:custody}; endorsements already public at that transition may be reprinted without such private custody. Every regret-free implementation of $C$ therefore contains a private conditional-release stage.

For sufficiency, collect completed private commitments from every member of $C\setminus C_0$, then release all of their names in one device action. If the action does not complete, the public set remains $C_0$. If it completes, the public set becomes $C$. Both sets belong to $\Sigma$ and protect all their members, including the vanguard, and Assumption~\ref{ass:atomic} rules out partial release. The mechanism therefore implements $C$. Under Assumption~\ref{ass:M}, the largest accessible and admissible full coalitions are $\Cs$ and $\Cp$.
\end{proof}

Without monotone protection, accessibility replaces the lattice characterization, and choosing among incomparable protective coalitions requires an objective and knowledge of the protection constraints.

\ifSubfilesClassLoaded{%
  \let\section\ACbodysection
  \bibliographystyle{aea}
  \bibliography{safe_coalitions}
}{ }

\end{document}

\fi

\end{document}